\documentclass[a4paper,onecolumn,11pt]{quantumarticle}
\pdfoutput=1
\usepackage[dvipsnames]{xcolor}

\usepackage{amsmath}
\usepackage{amssymb}
\usepackage{amsfonts}
\usepackage{mathrsfs}

\usepackage{amsthm}
	\newtheorem{thm}{Theorem}[section]
	\newtheorem{lem}[thm]{Lemma}
	\newtheorem{prop}[thm]{Proposition}
	\newtheorem{cor}[thm]{Corollary}

	\theoremstyle{definition}
	\newtheorem{rmk}[thm]{Remark}

    \newtheorem{main}{Main result}

\DeclareMathOperator{\R}{\mathbb{R}}
\DeclareMathOperator{\C}{\mathbb{C}}

\DeclareMathOperator{\scS}{\mathsf{S}}

\newcommand{\mG}{\mathcal{G}}
\newcommand{\mL}{\mathcal{L}}

\newcommand{\dd}{\mathrm{d}}

\newcommand{\id}{1\!\!1}

\newcommand{\err}{\mathcal{E}}

\def\Xint#1{\mathchoice
   {\XXint\displaystyle\textstyle{#1}}%
   {\XXint\textstyle\scriptstyle{#1}}%
   {\XXint\scriptstyle\scriptscriptstyle{#1}}%
   {\XXint\scriptscriptstyle\scriptscriptstyle{#1}}%
   \!\int}
\def\XXint#1#2#3{{\setbox0=\hbox{$#1{#2#3}{\int}$}
     \vcenter{\hbox{$#2#3$}}\kern-.5\wd0}}
\def\dashint{\Xint-}

\usepackage[utf8]{inputenc}
\usepackage[english]{babel}
\usepackage[pdftex]{graphicx}
\usepackage{enumerate}
\usepackage{ytableau}

\definecolor{lgr}{gray}{0.93}
\usepackage{mdframed}
\mdfdefinestyle{MyFrame}{%
    outerlinewidth=0pt,
    linecolor=lgr,
    innertopmargin=8pt,
    innerbottommargin=8pt,
    innerrightmargin=8pt,
    innerleftmargin=8pt,
    backgroundcolor=lgr}

\usepackage[numbers,compress,square,comma]{natbib}
\begin{document}
\bibliographystyle{ourstyle}
\title{Dynamical quantum phase transitions from random matrix theory}

\author{David P\'erez-Garc\'ia}
\email{dperezga@ucm.es}
\affiliation{Departamento de An\'alisis Matem\'atico y Matem\'atica Aplicada,\protect\\ Universidad Complutense de Madrid, 28040 Madrid, Spain}

\author{Leonardo Santilli}
\email{santilli@tsinghua.edu.cn}
\affiliation{Yau Mathematical Sciences Center, Tsinghua University, Beijing, 100084, China}
\affiliation{Departamento de Matem\'{a}tica, Grupo de F\'{i}sica Matem\'{a}tica, Faculdade de Ci\^{e}ncias,\protect\\ Universidade de Lisboa, 1749-016 Lisboa, Portugal}

\author{Miguel Tierz}
\email{tierz@mat.ucm.es}
\affiliation{Departamento de An\'alisis Matem\'atico y Matem\'atica Aplicada,\protect\\ Universidad Complutense de Madrid, 28040 Madrid, Spain}

\begin{abstract}
We uncover a novel dynamical quantum phase transition, using random matrix theory and its associated notion of planar limit. We study it for the isotropic XY Heisenberg spin chain. For this, we probe its real-time dynamics through the Loschmidt echo. This leads to the study of a random matrix ensemble with a complex weight, whose analysis requires novel technical considerations, that we develop. We obtain three main results: 1) There is a third order phase transition at a rescaled critical time, that we determine. 2) The third order phase transition persists away from the thermodynamic limit. 3) For times below the critical value, the difference between the thermodynamic limit and a finite chain decreases exponentially with the system size. All these results depend in a rich manner on the parity of the number of flipped spins of the quantum state conforming the fidelity. 
\end{abstract}
\maketitle
\tableofcontents

\section{Introduction}
Understanding the dynamics of quantum many body systems is one of the current most salient scientific challenges. New types of {\it phases} emerge from this perspective depending e.g. on how the system's evolution generates correlations or how fast it thermalizes. These include phases characterized by fast or slow scrambling, fulfilling the Eigenstate Thermalization Hypothesis (ETH) \cite{srednicki1994chaos,Deutsch:2018,shiraishi2017systematic,Mori:2017qhg}, phases corresponding to Many Body Localization (MBL) \cite{Nandkishore:2014kca,Vasseur:2016jtj,imbrie2016many,imbrie2017local,Parameswaran:2018enu,abanin2019colloquium}, with many body scars \cite{Bernien:2017,turner2018weak,Serbyn:2020wys}, among a large number of possible scenarios involving some sort of weakly ergodicity breaking \cite{turner2018weak,Serbyn:2020wys,sala2020ergodicity}.\par
New phases have concomitant new types of phase transitions. Among these, {\it dynamical quantum phase transitions (DQPTs)} \cite{Heyl:Ising,karrasch2013dynamical,Hickey:2014,vajna2014disentangling,Heyl:Broken,Kriel:2014,vajna2015topological,budich2016dynamical,Schmitt2015,Heyl:universality,Sharma:2015,Zhang:2016,Sharma:2016,Puskarov,Zunkovic:2018,Halimeh:2016,Banerjee:2016ncu,Karrasch:Potts,Zhou:nonHermitian,Guardado-Sanchez:2018,Heyl:2018zzb,Bandyopadhyay,Halimeh:2018a,Mishra:2020a,Halimeh:2018b,Khatun:2019,Pedersen:2020xuk,DeNicola:2020ddc,Zamani:2020,Peotta:2020tzu,Bao:2021bhc,Cheraghi:2021nea,Okugawa:2021jis,Halimeh:2021aeh,Naji:2021dur,Jafari:2021xty,Gonzalez:2022nyf,VanDamme:2022loc,Qin:2022ycz,Corps:2022hml,Mondal:2022xba} (for reviews \cite{Heyl:2017blm,zvyagin2016dynamical,heyl2019dynamical,Marino:2022eiw}) are still far from being well understood, partly because numerical simulations of time evolution in interacting quantum systems are extremely challenging, and partly because there are very few examples for which one can rigorously prove the existence of a DQPT. The first problem starts to be resolved thanks to the outstanding advance of quantum simulators, either digital or analogue, which can accurately simulate the long time evolution of some target Hamiltonian. Indeed, quantum simulations of DQPT have been already reported for instance in \cite{Bloch:2010,Flaschner:2018,Jurcevic:2017,zhang2017observation,Guo:2019,Wang:2019,Tian:2019,Nie:2019pxk}.\par
In this work, we contribute to the second problem, by finding and describing a new type of DQPT which differs significantly from all previous examples and, at the same time, can be analyzed rigorously using Random Matrix Theory (RMT). As it is standard in the context of DQPTs, we will consider the Loschmidt echo and the associated dynamical free energy \cite{Jalabert:2001,Hahn:1950,Gorin:2006}. We will show that, for the isotropic XY model, the dynamical free energy undergoes a third order phase transition in the spirit of, but different than, the celebrated Gross--Witten--Wadia (GWW) phase transition in lattice gauge theory \cite{Gross:1980he,Wadia:1980cp,Wadia:2012fr}. We will also show that signatures of the transition can already be observed in rather small system sizes.

\section{Main results}
A probe of dynamical phase transitions is the Loschmidt echo, or fidelity \cite{Gorin:2006}. A system of $L$ qubits is prepared in the initial state $\lvert \psi_0 \rangle$ and evolved forward in time through a Hamiltonian $H$. One is usually interested in the Loschmidt amplitude
\begin{equation*}
	\mG (t)  = \left\langle \psi_0 \left\lvert e^{-i t  H   } \right\rvert \psi_0 \right\rangle  ,
\end{equation*}
and the corresponding probability, the Loschmidt echo:
\begin{equation*}
	\mL (t) = \lvert \mG (t) \rvert^2 .
\end{equation*}
A DQPT is signalled by the non-analytic behaviour of the dynamical free energy $f(t) \propto - \ln \mL (t)$.\par
In this paper we will focus on the Loschmidt echo defined by the isotropic XY Hamiltonian\footnote{Also referred to as XX model.} 
\begin{equation}
\label{eq:HXY}
	H_{\text{XY}} = - \sum_{j=0} ^{L-1}	\left(  \sigma_j ^{-} \sigma_{j+1} ^{+} + \sigma_j ^{+} \sigma_{j+1} ^{-} \right)
\end{equation}
with periodic boundary conditions (PBC) and initial state the domain wall state
\begin{equation}
	\lvert \psi_0 \rangle = \lvert \underbrace{ \downarrow  \downarrow \dots \downarrow }_{N} \underbrace{ \uparrow  \uparrow \dots \uparrow }_{L-N} \rangle .
\label{eq:psi0}
\end{equation}
Therefore, $L$ denotes the size of the chain and $N$ is the number of flipped spins in the state. The corresponding Loschmidt amplitude is 
\begin{equation}
\label{eq:GXX}
	\mG_N (t) = \langle \psi_0 \lvert e^{-i t H_{\text{XY}}} \lvert \psi_0 \rangle  ,
\end{equation}
the echo is $\mL_N (t) = \lvert \mG_N (t) \rvert^2$, and we define the dynamical free energy as 
\begin{equation}
\label{eq:defDFE}
	f_N( \tau ) = - \frac{1}{2 N^2 } \ln  \mL_N (t) ,
\end{equation}
where $t$ and $\tau$ are related through scaling:
\begin{equation}
\label{eq:deftau}
	t= N \tau .
\end{equation}
Notice that $f_N( \tau )$ is normalized by $N^2$, not by the total number $L$ of qubits, hence it significantly differs from previous approaches \cite{LeClair:1995,Heyl:Ising}. The Loschmidt return rate $- \ln \mL_N (t)$ is not an extensive quantity here and remains finite in the thermodynamic limit (see Section \ref{sec:Methods}).\par
We will define $f( \tau )$ as the result of taking first the thermodynamic limit $L\rightarrow \infty$ and then the limit $N\rightarrow\infty$,  when we consider only {\it even} values of $N$ to compute the limit. 

\begin{figure}[t]
\centering
	\includegraphics[width=0.67\textwidth]{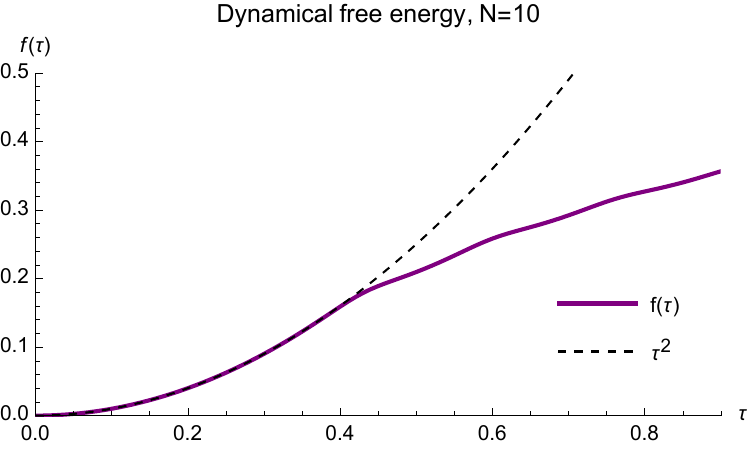}%
\caption{Dynamical free energy. $N=10$ and $L \gg 2N$.}
\label{fig:dfe}
\end{figure}\par
Our main results are:\par
\begin{mdframed}[style=MyFrame]
\begin{main}\label{main1}
 There is a third order dynamical phase transition in $f(\tau)$ occurring at $\tau_{\rm cr}\approx 0.33137$.
\end{main}
\begin{main}\label{main2}
 The third order phase transition persists in the thermodynamic limit $L,N\rightarrow \infty$ with $t=N\tau$ and constant ratio $L/N >1.2$.
\end{main}
\begin{main}\label{main3}
For $\tau<\tau_{\rm cr}$ the difference in $f_N(\tau)$ between considering the thermodynamic limit and a chain of length $L$ decreases exponentially with $L-N$ if $L>  2N$.
\end{main}
Moreover, we prove that the phase transition is robust against the presence of impurities in the preparation of the initial state \eqref{eq:psi0}.
\end{mdframed}\par
The first result means that the dynamical free energy $f(\tau)$ is not analytic, but rather has a discontinuity in its third derivative at $\tau_{\rm cr}$.
Furthermore, although taking $L \to \infty$ first is a convenient simplification, by main result \ref{main2} the phase transition remains valid at finite density $N/L$.
We stress that, whilst main results \ref{main1} \& \ref{main2} require $N \to \infty$ taking only even values, main result \ref{main3} also holds for the odd case. This will be shown in detail below.\par
All of our main results are mathematically sound, and are proven (or substantiated) using RMT. Dealing with real-time evolution requires the development of new techniques, which constitutes the main technical contribution of this paper.\par
\medskip
The connection between the XY model and RMT is well understood, but so far it mostly dealt with static quantities, like thermal states \cite{Perez-Garcia:2013lba,Stephan:2013,Pozsgay:2013,Perez-Garcia:2014aba,Stephan:2017,Santilli:2019wvq}. The Loschmidt echo (and generalizations thereof) has of course been considered previously. However, if, on one hand, the RMT formulation of \eqref{eq:GXX} (see Section \ref{sec:Methods}) was thoroughly studied in \cite{Krapivsky:2017sua}, the scaling limit \eqref{eq:deftau} remained insofar unaddressed. On the other hand, the vast majority of previous works relied on analytic continuation to the real-time echo from its thermal version \cite{Pozsgay:2013,Viti:2016,Stephan:2021yvk}. As a byproduct of our main result \ref{main1}, this approach is only valid for $\tau< \tau_{\text{cr}}$. Note also that a breakdown of the analytic continuation was already demonstrated in \cite{Piroli:2016fpu,Piroli:2018amn}.\par
For all these reasons, the phase transition eluded prior scrutiny.

\subsection{Differences with standard DPQTs}
The new dynamical phase transition presented here has many {\it exotic} characteristics.\par
The first one is that the DQPT emerges from the scaling limit \eqref{eq:deftau}, in analogy with the planar limit of Quantum Field Theory (QFT) and RMT \cite{Brezin:1977sv}. The transition reflects a non-analytic change in the limit $f (\tau)$ of the dynamical free energy \eqref{eq:defDFE} as a function of $\tau$ (Figure~\ref{fig:dfe}).\par
The second one is that the phase transition is third order. This is in principle unusual from the perspective of condensed matter physics, where most phase transitions are first or second order \cite{Sachdev:book}, and this seems to be the setting of previously reported DQPTs \cite{Canovi:2014,Heyl:2017blm,hamazaki2021exceptional}.\par
Third order phase transitions have been found in the study of integrable hyperbolic Richardson--Gaudin models, relevant for both superfluidity and superconductivity \cite{Rombouts:2010,Lerma:2011hp}. See \cite{eisele1983third} for earlier and  \cite{choi2021phase,chakravarty2021critical} (and references therein) for recent discussions of third-order transitions. A common thread is the lack of experimental realization so far, contrarily to the outcome of our work.\par
In RMT, third-order phase transitions are widespread. The first and best known example is the GWW transition \cite{Gross:1980he,Wadia:1980cp,Wadia:2012fr}, and a number of third-order transitions are indeed reduced to it \cite{majumdar2014top}. The GWW model was originally introduced in lattice Yang--Mills theory in two dimensions \cite{Bars:1979xb} and has since then found overarching multidisciplinary applications. It also has played a pivotal role in solving the long-standing problem of the longest increasing subsequence in combinatorics \cite{Johansson:1998,BDJ:1999}.\par 
The Loschmidt amplitude \eqref{eq:GXX} is intimately related to the GWW model: the latter may be seen as the imaginary-time version of \eqref{eq:GXX}. Notwithstanding this parallel, the analytic form of $f(\tau)$ as well as the mechanism beyond the DQPT will make clear that the third order transition we unveil is novel and radically distinct from the GWW one.\par
The third main difference is about how the DQPT we present is related to the zeros of the Loschmidt echo. DQPTs \cite{Heyl:2017blm} are usually associated to zeros in the Loschmidt echo which are attained in the thermodynamic limit. Here, as mentioned above, the Loschmidt echo is strictly positive and analytic for every even $N$ in the thermodynamic limit. The non-analyticity arises only when $N\to \infty$.

\subsection{The case of odd \texorpdfstring{$N$}{N}}
\begin{figure}[ht]
\centering
	\includegraphics[width=0.45\textwidth]{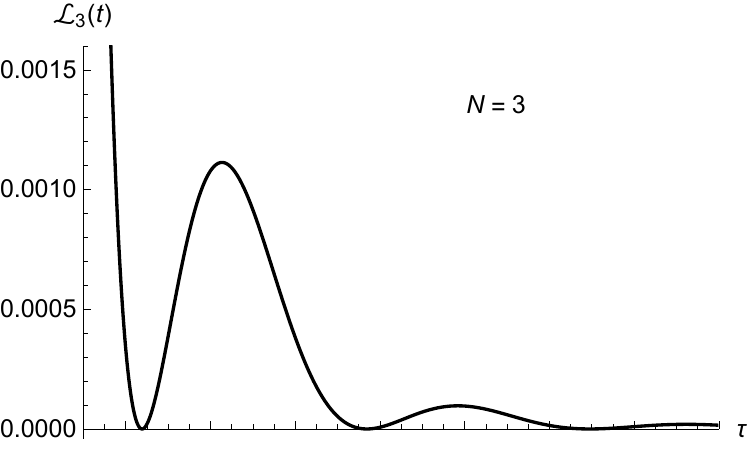}\hspace{0.05\textwidth}
	\includegraphics[width=0.45\textwidth]{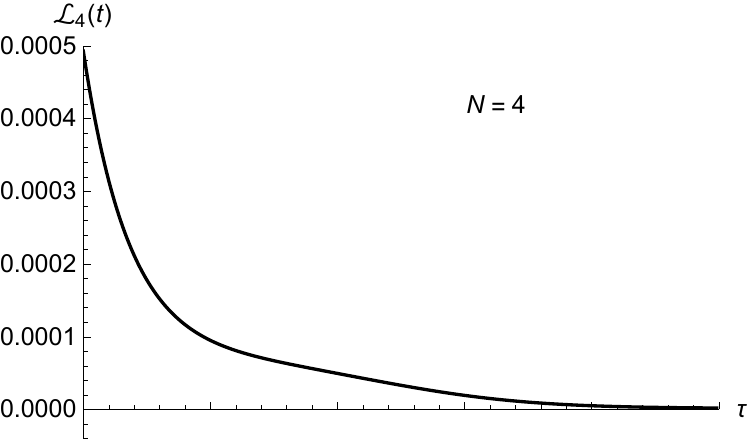}
\caption{Different behavior of the Loschmidt echo $\mL_N (t)$ at even and odd $N$, exemplified in $N=3,4$. $\mL_N (t)$ has zeros at finite values of $t$ for odd $N$, whereas it is positive and decays asymptotically as $e^{-t^2}$ for even $N$.}
\label{fig:Ltvst_N}
\end{figure}\par

We will also analyze the case of odd $N$ (meaning that we take $N\rightarrow\infty$ considering only odd values of $N$). This case has an extreme behavior regarding the relationship to zeros of the Loschmidt amplitude in stark contrast to the case of even $N$: the Loschmidt amplitude becomes zero at finite $L$, $N$ and $t$ --- a behaviour already observed in \cite{Krapivsky:2017sua}. Such zeros, which translate into divergences of the dynamical free energy, are responsible for the phase transition in the odd case (Figure~\ref{fig:Ltvst_N}). This implies a difference in how the phase transition is approached if one takes $N$ even or odd, see Figure~\ref{fig:fN3and4}. In the even case, $\mL_N (t)$ decays with $t$ and the dynamical free energy $f(\tau)$, after the phase transition, is {\it smaller} than the expected curve if there were no phase transition. In the odd case however, $f(\tau)$ is {\it larger} and $\mL_N (t)$  oscillates and has zeros. It is remarkable that, despite this mathematical fact, we can prove the following 
\begin{mdframed}[style=MyFrame]
\begin{main}\label{main4}
$\tau_{\rm cr}$ is the same for the even and odd cases.
\end{main}
\end{mdframed}

\begin{figure}[hb]
\centering
	\includegraphics[width=0.45\textwidth]{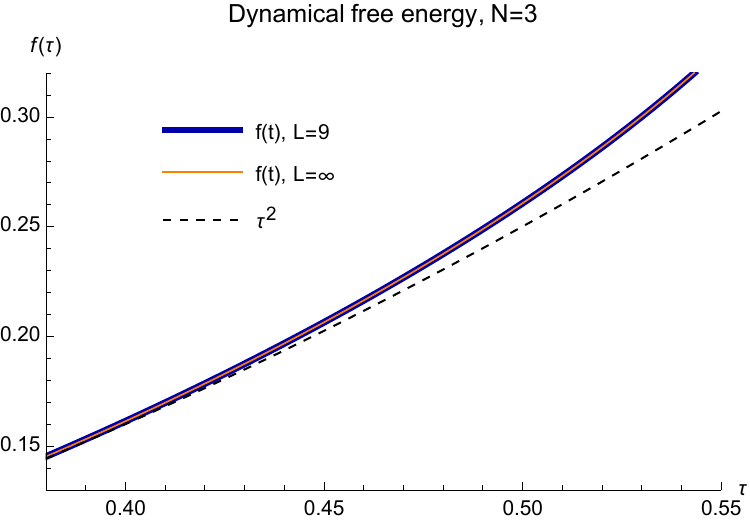}%
	\hspace{0.05\columnwidth}
	\includegraphics[width=0.45\textwidth]{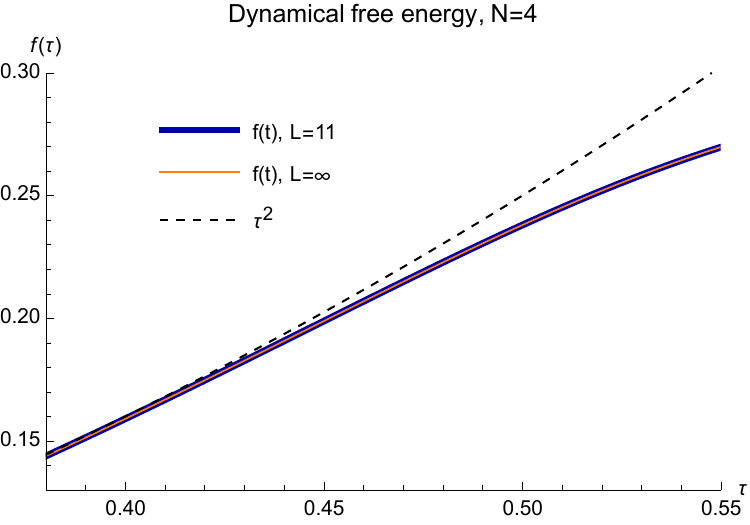}
\caption{Dynamical free energy. Left: $N=3,L =9$. Right: $N=4,L =11$. The bold blue line is the exact value for the given $L$, the thin orange line is the exact value at $L= \infty$ and the black dashed line is the extrapolation if there were no phase transition.}
\label{fig:fN3and4}
\end{figure}\par

\subsection{Experimental accessibility of the phase transition}

First of all, the presence of an impurity in the preparation of the initial state modifies $\mL_N (t)$ by a polynomial term in $t$, as we prove in Section \ref{sec:S3}. Thus, it would not spoil the visibility of the DQPT. One may therefore replace \eqref{eq:psi0} with any superposition of states with a number of impurities much smaller than $N$, and the analysis still holds.\par
Moreover, thanks to our main result \ref{main3}, one can observe clear signatures of the phase transition at very low values of $N$ and $L$.  For instance, as one can see in Figure~\ref{fig:fN3and4}, there is no difference in taking $L=9$ or the thermodynamic limit for $N=3$. In that case, $\mL_{3} (t=1.29) \approx 0.034$ and one already observes a $4.5\%$ deviation from the expected curve if there were no phase transition. Note that $\mL_{3} (t)$ is the probability value of returning to the state $\lvert\downarrow  \downarrow \downarrow  \uparrow  \uparrow\dots \uparrow  \rangle$ at time $t$. This can be estimated by measuring in the basis $\{\lvert \downarrow\rangle, \lvert \uparrow\rangle\}$ each one of the spins.\par
Recently, thanks to the quantum algorithm proposed in \cite{lu2021algorithms} to obtain expectation values on thermal states, the computation of Loschmidt echoes has emerged as a key subroutine in quantum computing. Motivated by that, it is shown in \cite{Yang:2023nak} that mild assumptions on the noise allow to implement easily error mitigation techniques, just by rescaling, for the study of the Loschmidt echo via a discretized time evolution.
Using those error mitigation tools, together with the fact that the total number of qubits is very small ($L=9$), the required evolution time is very short ($t=1.29$) and the total deviation to be detected is sufficiently large ($1.5\times 10^{-3}$), the experimental observation of the phase transition seems reasonably within reach in current digital quantum computers.\par
As for analogue quantum simulations, recent improvements in optical lattice microscopes have allowed to study very accurately time evolution in fermionic chains after a quench in cold atoms experiments \cite{Bloch:2017,Vijayan:2020}. Since the isotropic XY interaction is well-known to be equivalent to free fermions \cite{Lieb:1961fr},  one may also expect to observe the DQPT we uncover here in optical lattices. Furthermore, the implementation of a DQPT via cold atoms in an optical cavity was reported in \cite{Muniz:2020}.

\section{Techniques and sketch of the proofs}
\label{sec:Methods}

The map from the isotropic XY Loschmidt amplitude \eqref{eq:GXX} with PBC to a matrix model is central in our discussion \cite{Bogoliubov:2009txC,Perez-Garcia:2013lba}. At $L=\infty$ it is 
\begin{align}
	\mG_N (t) = \frac{1}{N!} \int_{[0,2 \pi]^N} \frac{\dd \theta_1 \cdots \dd \theta_N}{(2 \pi)^N} \prod_{1 \le a < b \le N}  \left\lvert  2 \sin \left( \frac{\theta_a - \theta_b}{2} \right) \right\rvert^2 ~
	\prod_{a=1} ^{N} e^{- i 2 t  \cos \theta_a} . \label{eq:GU=GWW} 
\end{align}
It follows from an old and well-known result in RMT \cite{Andreief} that the integral \eqref{eq:GU=GWW} equals the determinant of a $N\times N$ Toeplitz matrix:
\begin{equation}
	\mG_N (t) =  \det_{1 \le a,b \le N } \left[  i^{b-a} J_{a-b} (-2t) \right] , \label{eq:GU=Tdet} 
\end{equation}
where $J_{\nu}$ are the Bessel functions of the first kind.\par
The matrix model corresponding to a finite chain is presented in Section \ref{sec:S1}, and amounts to replace the integral \eqref{eq:GU=GWW} with a Riemann sum. It can be equivalently written as a Toeplitz determinant, as well.\par

\subsection{Main result 1: Dynamical phase transition in the even \texorpdfstring{$N$}{N} case}

We first consider a chain with $N$ even and $L \to \infty$. Corrections due to finite system size are under analytic control and are discussed below.\par
To set up the large $N$ limit, express \eqref{eq:GU=GWW} in the form 
\begin{align}
    \mG_{N} (t) & = \frac{1}{N!} \int_{[0,2 \pi]^N} \frac{\dd \theta_1 \cdots \dd \theta_N}{(2 \pi)^N} e^{-N^2 S (\theta_1, \dots, \theta_N )} , \label{eq:GexpSeff}  \\
     S (\theta_1, \dots, \theta_N ) & = \frac{1}{N} \sum_{a=1} ^{N} \frac{i t}{N} \cos \theta_a - \frac{1}{N^2} \sum_{1 \le a \ne b \le N} \ln \left\lvert  2 \sin \left( \frac{\theta_a - \theta_b}{2} \right) \right\rvert . \notag
\end{align}
It is customary to introduce the density of eigenvalues $\rho (\theta) = \frac{1}{N} \sum_{a=1}^{N} \delta (\theta - \theta_a )$ and rewrite $S (\theta_1, \dots, \theta_N ) $ in the functional form 
\begin{align}
    S [\rho] = i \tau \int \dd \theta \rho (\theta ) \cos \theta -  \int \dd \theta \rho (\theta )   \dashint \dd \varphi \rho (\varphi )  \ln \left\lvert  2 \sin \left( \frac{\theta - \varphi }{2} \right) \right\rvert , \label{eq:Seffrho}
\end{align}
where $\dashint$ is the principal value integral. For the two pieces to be same order in $N$, we have replaced $t/N$ with $\tau$ according to \eqref{eq:deftau}. This leads to the so-called \emph{planar limit} of the matrix model \cite{Brezin:1977sv}.\par
When $N \gg 1$, we expect that \eqref{eq:GexpSeff} is dominated by the saddle points of $S[\rho]$ (see Remark \ref{rmk:convergence} below for more details). The problem thus reduces to determine the distribution $\rho_{\ast} (\theta)$ that extremizes \eqref{eq:Seffrho}, and use it to compute 
\begin{equation*}
    f (\tau) = - \lim_{N \to \infty} \frac{1}{2N^2} \ln \mL_{N} (t) = \Re S[\rho_{\ast}] .
\end{equation*}
The Euler--Lagrange equation from the extremization of \eqref{eq:Seffrho} is 
\begin{equation}
\label{eq:ComplexSPEUN}
	\dashint \frac{\dd u}{u } \frac{z+u}{z-u} \varrho (z) = \tau \left( z - z^{-1} \right) ,
\end{equation}
where we have passed to exponential variables, $(e^{i \theta}, e^{i \varphi}) \to (z,u)$, and $\varrho (e^{i \theta}) = \rho_{\ast} (\theta)$.\par
A novel feature of \eqref{eq:Seffrho} is the presence of an imaginary coefficient. As a consequence, \eqref{eq:ComplexSPEUN} does not admit a standard solution with $z \in U(1)$. To solve the complexified extremization problem \eqref{eq:ComplexSPEUN} is one of the main technical achievements of this work.\par
We relax the assumption on supp$(\varrho)$, allowing it to be a curve $\Gamma \subset \mathbb{C}$ subject to the requirement $\lim_{\tau \to 0} \Gamma = U(1)$. The details are provided in Section \ref{sec:PT}, and the upshot is that $\Gamma$ is the connected component of 
\begin{equation}
\label{eq:defGammaU}
	\Re \left\{   \ln z + i \tau \left( z - \frac{1}{z} \right) \right\} = 0 
\end{equation}
with appropriate initial condition.\par
Letting the system evolve in time, the dynamical free energy is non-analytic when a complex zero of $\varrho (z)$ hits $\Gamma$. Hence, there exists a critical time $\tau_{\text{cr}}$ (see Section \ref{sec:PT}), which is the unique positive solution to
\begin{equation}
\label{eq:criticaleqtau}
	\sqrt{1 + 4 \tau_{\mathrm{cr}} ^2 } - \log \left(\frac{  1 \pm \sqrt{1 + 4 \tau_{\mathrm{cr}} ^2 } }{2 \tau_{\mathrm{cr}} } \right) = 0 ,
\end{equation}
i.e. $\tau_{\text{cr}} \approx 0.33137171$.\par
Solving \eqref{eq:ComplexSPEUN} for $\varrho (z)$ and using \eqref{eq:defGammaU}, we evaluate 
\begin{equation*}
	f( \tau ) = \tau^2 , \quad \tau < \tau_{\text{cr}} .
\end{equation*}
Moreover, for $\tau > \tau_{\text{cr}}$ we can extract $f (\tau )$ it in the two regimes $\tau \to \tau_{\text{cr}}$ and $\tau \to \infty $. In the former limit,  
\begin{equation*}
	\lim_{\tau \to \tau_{\text{cr}}} \partial_{\tau} ^2 f (\tau ) = 2 , \quad \quad \lim_{\tau \to \tau_{\text{cr}}} \partial_{\tau} ^3 f (\tau ) \ne 0 ,
\end{equation*}
implying that the DQPT is third order. Conversely, for $\tau \to \infty $ we get $f (\tau ) \propto \ln \tau $ (in agreement with Figure~\ref{fig:dfe}).\par
The late time asymptotic behaviour of $f(\tau)$, as well as the critical value $\tau_{\text{cr}}$, guarantee that the DQPT we uncover is not merely a Wick rotation of the GWW phase transition.\par

\begin{rmk}[Complex saddles]\label{rmk:convergence}
    To derive our main result \ref{main1} we have assumed that, in the regime $N \gg 1$, the logarithm of the matrix model converges to its saddle point value. In imaginary time, not only this is a widely used fact in physics, but for the GWW model the convergence has even been established mathematically \cite{Johansson:1998,BDJ:1999}.\par 
    In real (or complex) time dynamics, the convergence to the complex saddles is generically assumed in the physics literature; see for instance \cite{Copetti:2020dilC}, where the GWW model at complex time was considered. The convergence to the saddle point value, albeit not mathematically proven, is supported by the following facts.
    \begin{itemize}
        \item[(i)] We are interested in $f(\tau)$, which is real by definition. It is thus believed that a standard extremization argument will ultimately apply.
        \item[(ii)] There exists only one saddle point configuration (up to permutations) in both phases. This is a peculiarity of third order phase transitions, as opposed to the more familiar first and second order ones, characterized by the presence of multiple critical points. At its unique saddle point, $S[\rho]$ turns out to be real. This avoids us the subtleties of comparing complex values of the matrix model at different saddle points.
        \item[(iii)] \emph{A posteriori}, we will benchmark our findings against the mathematical work \cite{DeanoC}, finding perfect agreement.
    \end{itemize}
\end{rmk}

\subsection{Main results 2 \& 3: Finite chain effects}
A finite chain length $L < \infty $ replaces \eqref{eq:GU=GWW} with a discrete ensemble, in which the $N$ eigenvalues are distributed among $L$ sites. We refine the large $N$ analysis for the discrete ensemble and consider the scaled thermodynamic limit $L,N,t \to \infty$ with fixed rates $\tau = t/N$ and 
\begin{equation}
\label{eq:defell}
    \ell =L/N \ge 1 .
\end{equation}
We shall show in Section \ref{sec:S4} that a finite ratio $1 < \ell < \infty$ induces a DQPT at a certain time $\tau_{\star} (\ell )$ (the system is trivial at $\ell=1$). The dependence of this new critical time $\tau_{\star}$ on $\ell$ is given explicitly. We impose $\tau_{\star} (\ell) > \tau_{\text{cr}}$, to require that this transition induced by a finite ratio \eqref{eq:defell} does not invalidate the DQPT at $\tau_{\text{cr}}$ discussed so far. The inequality is satisfied for $\ell \gtrsim 1.2$. In other words, the hypotheses that lead to the main result \ref{main1} continue to hold at finite $\ell$ as long as $\ell \gtrsim 1.2$.\par
Furthermore, we quantify the validity of the $L=\infty$ approximation. We prove that, in the scaled thermodynamic limit \eqref{eq:defell} and in the phase $\tau < \tau_{\text{cr}}$, the error in neglecting the dependence on $\ell < \infty$ is exponentially small in $L-N$ if $\ell \ge 2$. The derivation, detailed in Section \ref{sec:S4}, builds on the mathematical work \cite{BaikLiuC} and on the theory of orthogonal polynomials.\par
We also complement the analytic estimate with quantitative numerical evidence. Define the error 
\begin{equation}
\label{def:errorNL}
	\err_N (\ell, \tau) = \ln \frac{ \mL_N (t)\vert_{2N<L<\infty}}{\mL_N (t)\vert_{L=\infty}} .
\end{equation}
The closer $\err_N $ is to zero, the better is the infinite chain approximation. Figure~\ref{fig:finitechainE} shows that \eqref{def:errorNL} converges quickly to 0 as $\ell$ is increased, and that the $L= \infty$ approximation is more accurate away from the critical point.
\begin{figure}[ht]
\centering
	\includegraphics[width=0.45\textwidth]{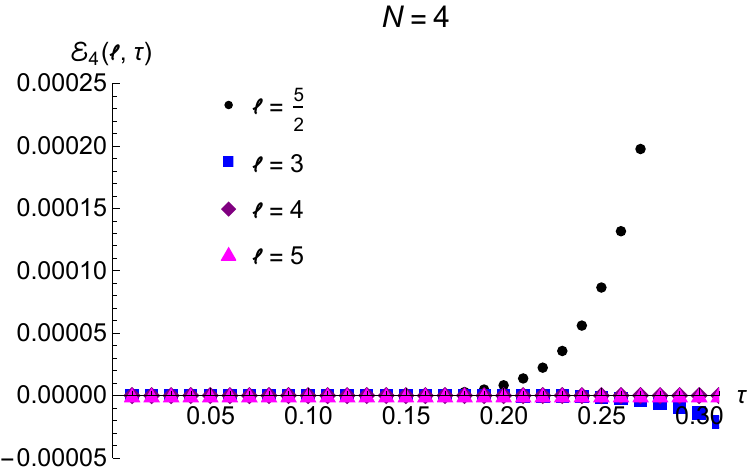}\hspace{0.05\textwidth}
	\includegraphics[width=0.45\textwidth]{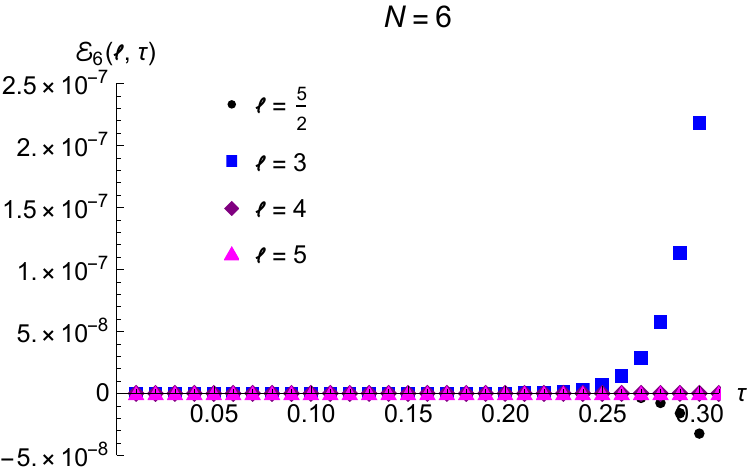}
\caption{Exact estimate of the error $\err_N$ at various $\ell$. Left: $N=4$, Right: $N=6$.}
\label{fig:finitechainE}
\end{figure}\par
We can moreover analyze the combined effect of finite $L$ and $N$ versus our asymptotic analysis. A quantitative probe of such net effect is 
\begin{equation}
\label{eq:errorRdef}
	\mathcal{R}_{N} (\ell, \tau) = 1 - \frac{ f (\tau)}{ f_{N,L} (\tau) } ,
\end{equation}
where $f$ (resp. $f_{N,L}$) is the dynamical free energy \eqref{eq:defDFE} in the large $N$ limit at $L=\infty$ computed analytically (resp. the dynamical free energy at finite $N$ and $2N<L<\infty$). The quantity \eqref{eq:errorRdef} is shown in Figure~\ref{fig:finitechainRvst}-\ref{fig:finitechainRvsL}. We infer from the plots that $f_{N,L} (\tau)$ approaches $f (\tau)$ monotonically from below, and that the discrepancy between the finite chain results and the asymptotic analysis are almost entirely due to finite $N$ effects, compared to which the $\ell < \infty$ effects are negligible if $\ell >2$. These conclusions agree with the fully analytic main result \ref{main3}.\par
\begin{figure}[ht]
\centering
	\includegraphics[width=0.45\columnwidth]{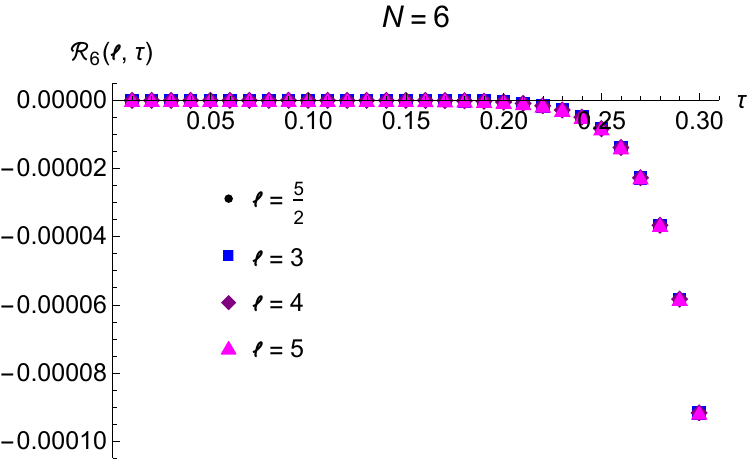}\hspace{0.05\columnwidth}
	\includegraphics[width=0.45\columnwidth]{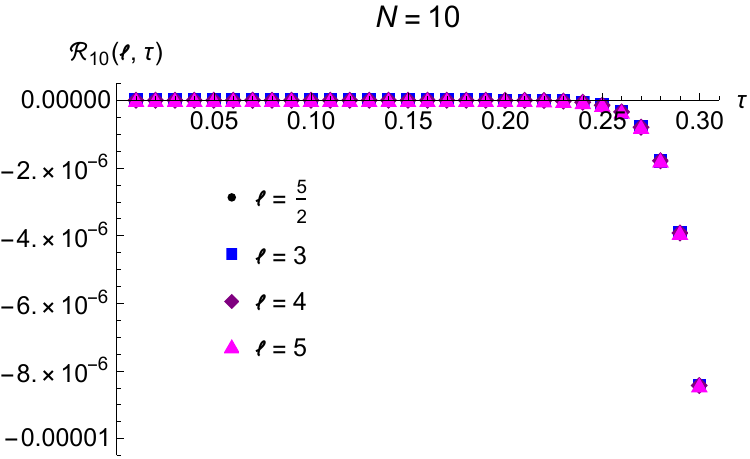}
\caption{Exact estimate of the error $\mathcal{R}_N$ at various $\ell$. Left: $N=6$, Right: $N=10$.}
\label{fig:finitechainRvst}
\end{figure}\par
\begin{figure}[ht]
\centering
	\includegraphics[width=0.45\columnwidth]{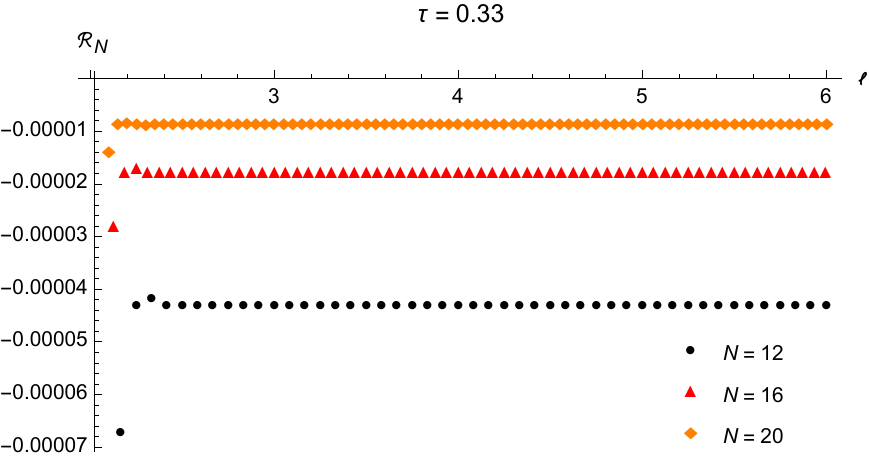}
\caption{Error $\mathcal{R}_N$ as a function of $\ell$ at $\tau\approx \tau_{\text{cr}}$.}
\label{fig:finitechainRvsL}
\end{figure}\par

\subsection{Main result 4: Odd \texorpdfstring{$N$}{N} and quantum speed limit}
Quantum speed limits (QSLs) pinpoint underlying time scales of physical processes by producing lower bounds on the expectation value of an observable or on the rate of change of a quantum state \cite{mandelstam1991uncertainty,margolus1998maximum,ness2021observing} (for a review \cite{Deffner:2017cxz}). We have seen above that, for every odd $N$, $\mL_N (t)$ oscillates with damped amplitude and has zeros, as shown in Figure~\ref{fig:Ltvst_N}. Thus, we focus here on the earliest time at which an evolved state $\lvert \psi (t) \rangle = e^{i t H } \lvert \psi_0 \rangle$ is orthogonal to the initial state $\lvert \psi_0 \rangle $ \cite{vaidman1992minimum}. In the current model, let us denote the corresponding value of $\tau$ as $\tau_{N} ^{\text{QSL}}$, that is 
\begin{equation*}
    \tau_{N} ^{\text{QSL}} = \min \left\{ \tau > 0 \ : \ \mG_N (t) = 0 \right\} ,
\end{equation*}
where we recall that $t$ and $\tau$ are related through \eqref{eq:deftau}. By construction, $\tau_{N} ^{\text{QSL}}$ is intimately related to the zeros of the Loschmidt echo \cite{Zhou:2021bfo}.\par
We show the existence of a finite limiting value $\tau^{\text{QSL}}$ of the QSL, namely 
\begin{equation*}
    \exists \ 0< \tau^{\text{QSL}} < \infty \ : \quad \lim_{k \to \infty} \tau_{2k+1} ^{\text{QSL}} = \tau^{\text{QSL}} .
\end{equation*}
The proof relies on the existence, for $\tau <  \tau_{N} ^{\text{QSL}} $, of a system of orthogonal polynomials on the unit circle, that are naturally associated to the Toeplitz determinant \eqref{eq:GU=Tdet}. The lower bound $ \tau^{\text{QSL}} > 0$ can be shown applying Szeg\H{o}'s strong limit theorem \cite{Szegoth}, which is the classical result for the asymptotics of a wide class of Toeplitz determinants, including \eqref{eq:GU=Tdet}, complemented with analytic arguments.\par
Moreover, we obtain that the limit value is precisely 
\begin{equation}
\label{eq:tQSL=tcr}
    \tau^{\text{QSL}} = \tau_{\text{cr}} ,
\end{equation}
which is the content of our main result \ref{main4}. The proof of \eqref{eq:tQSL=tcr} is technical and we defer it to Section \ref{sec:Nodd}. It uses random matrix theory, exploiting also a connection with the Toda lattice \cite{Adler99}.\par
Evidence for the identity \eqref{eq:tQSL=tcr} is given in Figure~\ref{fig:tauQSLNodd}, which shows the first $27$ values of $\tau^{\text{QSL}} _{2k+1}$ computed numerically.
\begin{figure}[ht]
\centering
	\includegraphics[width=0.67\columnwidth]{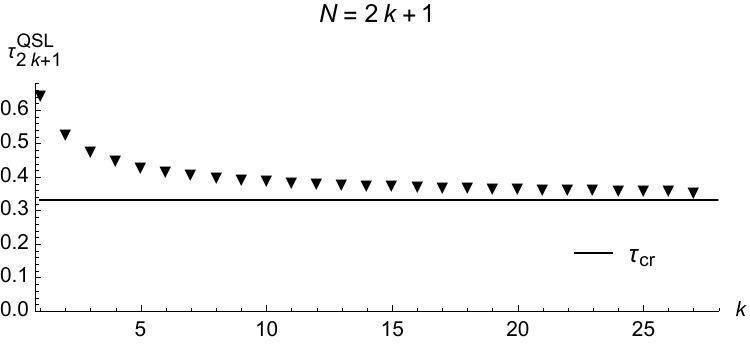}
\caption{Values of $\tau^{\mathrm{QSL}} _{2k+1}$ for $k=1, \dots, 27$. The horizontal asymptote is $\tau_{\mathrm{cr}}$.}
\label{fig:tauQSLNodd}
\end{figure}\par

\part*{Detailed statements and proofs}
\addcontentsline{toc}{part}{Detailed statements and proofs}

\section{Setup: Loschmidt echo and random matrix theory}
\label{sec:S1}
\subsection{The spin chain}
		Consider a one-dimensional spin chain with $L$ qubits, $L \gg 1$. We label the qubits $k=1, \dots, L$ and use the notation $\left\{ \lvert \uparrow \rangle_k ,  \lvert \downarrow \rangle_k \right\}$ for the natural basis in the $k^{\text{th}}$ copy of the single-qubit Hilbert space $\mathscr{H}$. The Hilbert space of the spin chain is
  \begin{equation*}
       \mathscr{H}^{\otimes L} := \bigotimes_{j=1} ^{L} \mathrm{Span}\left\{ \lvert \uparrow \rangle_j ,  \lvert \downarrow \rangle_j \right\}   ,
  \end{equation*}
  the $L^{\text{th}}$ tensor product of the single-qubit Hilbert space, with $\dim \left( \mathscr{H}^{\otimes L}\right) = 2^L$.\par
  The spin chain is prepared in the initial state 
		\begin{equation*}
			\lvert \Uparrow \rangle := \bigotimes_{k= 1} ^{L} \lvert \uparrow \rangle_k = \lvert  \uparrow  \uparrow \dots  \uparrow \uparrow \rangle ,
		\end{equation*}
		normalized as $\langle \Uparrow \vert \Uparrow \rangle = 1$. It is possible to assemble $N\le L $ spin-flip operators acting on $N$ distinct copies of the single-qubit Hilbert space into string operators
		\begin{equation}
		\label{eq:stringopsigma}
			\sigma^{-} _{k_1 \cdots k_N} = \sigma_{k_{N}} ^{-} \cdots \sigma_{ k_2 } ^{-} \sigma_{ k_1 } ^{-} ,
		\end{equation}
		$k_i \ne k_j$ if $i \ne j $ (and likewise for $\sigma^{+} _{k_1 \cdots k_N}$). Thanks to the commutation relations, we can restrict without loss of generality to configurations with 
		\begin{equation*}
			k_1 < k_2 < \cdots < k_N .
		\end{equation*}
		Acting with the operators \eqref{eq:stringopsigma} on $\lvert \Uparrow \rangle $ produces the states 
		\begin{align}
			\lvert N, \kappa \rangle & := \sigma_{k_{N}} ^{-} \cdots \sigma_{ k_2 } ^{-} \sigma_{ k_1 } ^{-} \lvert \Uparrow \rangle \label{eq:stateNkappa} \\
			& = \lvert \uparrow \uparrow \dots  \uparrow  \underbrace{ \downarrow}_{k_1^{\text{th}}} \uparrow  \dots  \uparrow  \underbrace{\downarrow}_{k_2^{\text{th}}} \uparrow  \dots \uparrow  \underbrace{\downarrow}_{k_{N}^{\text{th}}} \uparrow  \dots \uparrow  \uparrow \rangle . \notag
		\end{align}
		These states are labelled by the integer $N \in \mathbb{N}$ and a partition 
		\begin{equation*}
			\kappa = \left\{ \kappa_1 \ge \kappa_2 \ge \cdots \ge \kappa_N \ge 0 \right\} \subset \mathbb{N}^N
		\end{equation*}
		of length at most $N$. (The dependence on $L$ is left implicit). The partition $\kappa$ is related to the labels $\left\{ k_j \right\}_{j=1, \dots, N}$ of the qubits acted upon through 
		\begin{equation*}
			\kappa_j = k_{N-j+1} -(N -j+1), \quad j=1, \dots, N.
		\end{equation*}
		Notice that it is important to specify $N$, not only $\kappa$, because those partitions with length$(\kappa)<N$ give distinct operators \eqref{eq:stringopsigma} for distinct values of $N$.\par
  The states $\lvert N, \kappa \rangle$ are multi-domain wall states, consisting of strings of $\uparrow $ alternating with strings of $\downarrow $. In particular, the empty partition $\kappa = \emptyset $ corresponds to the state 
		\begin{align*}
			\lvert \psi_0 \rangle := \lvert N, \emptyset \rangle & = \sigma_{N} ^{-} \cdots \sigma_{ 1 } ^{-} \lvert \Uparrow \rangle \\
			& = \lvert \underbrace{ \downarrow \downarrow \dots \downarrow }_{N}  \underbrace{ \uparrow  \dots \uparrow \uparrow }_{L-N} \rangle ,
		\end{align*}
		which has a single domain wall, for a chain with open boundary conditions (OBC), or a two-sided domain wall for a chain with periodic boundary conditions (PBC).

\subsection{Quantum quench}
\label{sec:derivquench}
		We consider a quantum quench protocol in which we begin with a state $\lvert N , \kappa \rangle $ and a trivial Hamiltonian and, at $t=0$, we suddenly turn on the spin-$\frac{1}{2}$ Heisenberg XX Hamiltonian  
  \begin{equation*}
    H_{\text{XY}} = - \sum_{j=0} ^{L-1}	\left(  \sigma_j ^{-} \sigma_{j+1} ^{+} + \sigma_j ^{+} \sigma_{j+1} ^{-} \right) .
  \end{equation*}
  We are interested in the return probability, that goes under the name of Loschmidt echo, or fidelity:
		\begin{equation}
			\mL_{N,L} (t; \kappa ) = \lvert \langle N, \kappa \lvert e^{-it H_{\text{XY}}} \rvert N, \kappa \rangle \rvert^2 .
		\end{equation}
		The corresponding probability amplitude, the Loschmidt amplitude, is
		\begin{equation}
        \label{eq:defamplitudeGNLtkappa}
			\mG_{N,L} (t; \kappa ) = \langle N, \kappa \lvert e^{-it H_{\text{XX}}} \rvert N, \kappa \rangle .
		\end{equation}\par
		Let us recall that, in general, the Loschmidt echo measures, by definition, the overlap of a state $e^{it H_0} \lvert \psi \rangle $, evolved forward in time with an unperturbed Hamiltonian $H_0$, with a state $e^{i t H} \lvert \psi \rangle $ evolved with a perturbed Hamiltonian $H$. See \cite{Gorin:2006} for a review. Here we focus on one of the simplest possible setups: the initial Hamiltonian is trivial, $H_0 = 0$, the perturbed Hamiltonian is $H=H_{\text{XY}}$ and the state is \eqref{eq:stateNkappa}.\par
  \medskip
  Assuming an infinitely long chain, $L\to \infty$, the Loschmidt amplitude is given by \cite{Bogoliubov:2009txC}
		\begin{equation}
		\label{eq:GNintegralgen}
			\mG_{N,L=\infty} (t; \kappa ) = \int_{G} \dd U ~ \chi_{\kappa} (U) \chi_{\kappa} (U^{-1}) ~  e^{-it \mathrm{Tr} \left( U+U^{-1} \right) } ,
		\end{equation}
		where $\chi_{\kappa}$ is the character of the representation $\kappa$ of $G$, e.g. a Schur polynomial if $G=SU(N)$, and $\dd U$ is the normalized Haar measure on $G$, such that 
  \begin{align*}
      \langle N, \kappa \vert N, \kappa \rangle  &= \mG_N (0) \\
        &= \int_{G} \dd U ~ \chi_{\kappa} (U) \chi_{\kappa} (U^{-1}) = 1 .
  \end{align*}
The derivation of \eqref{eq:GNintegralgen} is based on the observation that the Loschmidt amplitude \eqref{eq:defamplitudeGNLtkappa} satisfies the differential equation 
\begin{equation}
\label{eq:diffeqBogoliubov}
    i \frac{\dd \mG_{N,L} (t; \kappa) }{\dd t } = \sum_{a=1}^{N} \left[ \mG_{N,L} (t; \kappa + \Box_a )  + \mG_{N,L} (t; \kappa - \Box_a ) \right] .
\end{equation}
In the latter expression, $\kappa \pm \Box_a$ means the partition obtained by appending or removing a box to the $a^{\text{th}}$ row of the Young tableaux representing $\kappa $. Equivalently, 
\begin{equation*}
    (\kappa_1, \kappa_2, \dots, \kappa_a , \dots, \kappa_N ) \pm \Box_a := (\kappa_1, \kappa_2, \dots, \kappa_a \pm 1 , \dots, \kappa_N ) .
\end{equation*}
For instance, 
\begin{equation*}
    \begin{aligned}
        & \kappa = (4,3,1,1) \\
        & \begin{ytableau} \ & \ & \ & \ \\  \ & \ & \ \\  \ \\ \
        \end{ytableau}
    \end{aligned}\qquad\qquad
    \begin{aligned}
        & \kappa + \Box_2= (4,4,1,1) \\
        & \begin{ytableau} \ & \ & \ & \ \\  \ & \ & \ & \ \\  \ \\ \
        \end{ytableau}
    \end{aligned}\qquad\qquad
    \begin{aligned}
        & \kappa - \Box_2= (4,2,1,1) \\
        & \begin{ytableau} \ & \ & \ & \ \\  \ & \  \\  \ \\ \
        \end{ytableau}
    \end{aligned}
\end{equation*}\par
\begin{rmk}[Time evolution of domain wall state]
    Remember that one is free to choose $\kappa$, and hence the initial state $\lvert N, \kappa \rangle$, as part of the data that specify the system. Once the initial state is prepared, formula \eqref{eq:diffeqBogoliubov} governs the time evolution of the system, expressed as a superposition in the basis \eqref{eq:stateNkappa}.\par
    The single-domain wall state $\lvert \psi_0 \rangle $ corresponds to the empty partition $\kappa = \emptyset$, for which the character $\chi_{\emptyset} =1$ is the identity function in \eqref{eq:GNintegralgen}. In this initial state, the time evolution equation \eqref{eq:diffeqBogoliubov} shows that the first order variation of the Loschmidt amplitude $\mG_{N,L} (t; \kappa)$ is given by the nearest-neighbour hopping of a qubit located at the boundary of the domain wall. As a consistency check, this behaviour can easily be retrieved by direct analysis of the time evolution in the Hamiltonian formalism.
\end{rmk}
A solution to \eqref{eq:diffeqBogoliubov} at imaginary time $t=-i\beta $ was given in \cite{Bogoliubov:2006C,Bogoliubov:2007C} in the form of a determinant. The derivation is algebraic and does not rely on reality of the parameters, thus it applies straightforwardly to the case at hand. The specific form of the determinant solving \eqref{eq:diffeqBogoliubov} is sensitive to the boundary conditions, so let us pause to discuss these first. We will come back and complete the proof of \eqref{eq:GNintegralgen} afterwards.\par
\medskip
In this work, we consider two types of boundary conditions:
\begin{itemize}
\item[i)] Periodic boundary conditions (PBC for short);
\item[ii)] Absorbing boundary conditions on the left and Open boundary conditions on the right (henceforth ABC for short).
\end{itemize}
By \emph{absorbing} boundary conditions, we mean that we set to zero the probability amplitude of a hop of the $0^{\text{th}}$ qubit to its left \cite{Bogoliubov:2006C}.\par
We may also allow for open boundary conditions (OBC) on both sides of the chain. As we will explain below, neglecting this choice results in no loss of generality for what concerns our results.\par
   It follows from \cite{Bogoliubov:2009txC} that $\mathcal{G}_{N,L} (t;\kappa)$ solving \eqref{eq:diffeqBogoliubov} with PBC is a Toeplitz determinant. For ABC, one gets instead a Toeplitz+Hankel determinant. We do not write down their form explicitly because we will not use it in the rest of the work. More details on the derivation will be given for the $\kappa= \emptyset$ case below.\par
   Formula \eqref{eq:GNintegralgen} follows from applying the celebrated Andr\'eief's identity \cite{AndreiefC,MeetAC} to these determinants.
   \begin{lem}[Andr\'eief's identity]
   \label{lem:AI}
   Let $\dd \mu (z)$ be a Borel measure on $\mathscr{D}$ and $\left\{ \Phi_{a-1} \right\}_{a=1} ^{N}$ and $\left\{ \Psi_{a-1} \right\}_{a=1} ^{N}$ be two collections of $\mu$-integrable functions on $\mathscr{D}$. Then
   \begin{multline*}
    \frac{1}{N!} \int_{\mathscr{D}} \dd \mu (z_1) \cdots \int_{\mathscr{D}} \dd \mu (z_N) ~ \det_{1 \le a,b \le N} \left[ \Phi_{a-1}  (z_b) \right] \det_{1 \le a,b \le N} \left[ \Psi_{a-1}  (z_b) \right] \\ = \det_{1 \le a,b, \le N} \left[ \int_{\mathscr{D}} \dd \mu(z) ~ \Phi_{a-1}  (z) \Psi_{b-1}  (z) \right] .
   \end{multline*}
   \end{lem}
   Equivalently, one could use the generalization of the Heine--Szeg\H{o} identity in \cite{BumpDiaconisC}. In each case, we arrive at \eqref{eq:GNintegralgen} with 
\begin{itemize}
\item[i)] $G=U(N)$ for PBC;
\item[ii)] $G=USp(2N)$ for ABC.\footnote{We denote $USp(2N) := Sp(N) \cap U(2N)$ the group of unitary-symplectic $2N \times 2N$ matrices.}
\end{itemize}
Furthermore, it is possible to realize a chain with OBC at both ends by breaking the $U(1)$ translation invariance of the PBC chain, cutting it open. It has the effect of removing the center symmetry (which is a 1-form symmetry, in the gauge theory parlance), thus the integration in \eqref{eq:GNintegralgen} is over $PSU(N)$.
\begin{rmk}[OBC and global form of $G$]
It is oftentimes quoted in the literature that, for OBC, the pertinent group is $SU(N)$. From the analysis of the symmetries of the model we get that the resulting group should be centerless, hence its global form is $PSU(N)$. This subtlety is not relevant for the ensuing discussion.
\end{rmk}\par

\subsection{Preparing the quench in the simplest domain wall state}
We focus on the simplest initial state with domain wall, 
\begin{equation*}
\lvert \psi_0 \rangle := \lvert N, \emptyset \rangle .
\end{equation*}
To reduce clutter, we denote 
    \begin{equation*}
        \mG_{N,L} (t) := \mG_{N,L} (t; \kappa=\emptyset )  \qquad \mL_{N,L} (t) := \mL_{N,L} (t; \kappa=\emptyset ) .
    \end{equation*}\par
\medskip
The determinants expressing the Loschmidt amplitude with, respectively, PBC and ABC are \cite{Bogoliubov:2009txC,Perez-Garcia:2013lba}:
\begin{subequations} 
\begin{align}
	\mG^{(\text{\rm PBC})}_{N,L=\infty} (t) &=  \det_{1 \le a,b \le N } \left[  i^{b-a} J_{a-b} (-2t) \right], \label{eq:GU=Tdetlong} \\
	\mG^{(\text{\rm ABC})}_{N,L=\infty} (t) &= \det_{1 \le a,b \le N } \left[  i^{b-a} J_{a-b} (-4t) - i^{-b-a} J_{a+b} (-4t) \right] , \label{eq:GSp=THdetlong}
\end{align}
\label{eq:G=detlong}
\end{subequations}
where $J_{\nu}$ are the Bessel functions of the first kind. For trivial $\kappa = \emptyset$, the integral \eqref{eq:GNintegralgen} becomes simply 
\begin{subequations}
\begin{align}
    \mG^{(\text{\rm PBC})}_{N,L=\infty} (t)  \ &= \ \left. \langle \psi_0 \lvert e^{-it H_{\text{XY}}} \rvert \psi_0 \rangle \right\rvert_{\text{\rm PBC}} \ = \ \int_{U(N)} \dd U ~ e^{-it \mathrm{Tr} \left( U + U^{-1} \right) } \\
    \mG^{(\text{\rm ABC})}_{N,L=\infty} (t) \ &= \ \left. \langle \psi_0 \lvert e^{-it H_{\text{XY}}} \rvert \psi_0 \rangle \right\rvert_{\text{\rm ABC}} \ = \ \int_{USp(2N)} \dd U ~ e^{-it \mathrm{Tr} \left( U + U^{-1} \right) } .
\end{align}
\label{eq:Gmatrixintegralform}
\end{subequations}
It is customary in random matrix theory to reduce integrals of the form \eqref{eq:Gmatrixintegralform} to integrals over eigenvalues, thanks to Weyl's integration formula.
\begin{lem}[Weyl's integration formula]
    Let $G$ be a compact, connected Lie group, $\mathbb{T}_G$ a maximal torus and $\mathrm{W}_G$ the Weyl group of $G$ determined by $\mathbb{T}_G$. Denote by $\dd U$ the normalized Haar measure on $G$ and $\dd \mathsf{t} $ the Lebesgue measure on the torus. Besides, let $\Phi$ be a class function on $G$. Then 
    \begin{equation*}
        \int_G \dd U ~\Phi (U)  = \frac{1}{\lvert \mathrm{W}_G \rvert }\int_{\mathbb{T}_G} \dd \mathsf{t} ~ \Phi (\mathsf{t}) \left\lvert \Delta_{\text{\rm Weyl}}  (\mathsf{t}) \right\rvert^2 ,
    \end{equation*}
    where, denoting $\mathrm{Roots}_+ (G)$ the set of positive roots of $G$, 
    \begin{equation*}
        \Delta_{\text{\rm Weyl}}  (\mathsf{t}) := \prod_{\alpha \in \mathrm{Roots}_+ (G)} \left( e^{\alpha (\mathsf{t})/2} - e^{-\alpha (\mathsf{t})/2} \right) .
    \end{equation*}
\end{lem}
Weyl's integration formula allows to rewrite \eqref{eq:Gmatrixintegralform} as the eigenvalue integrals
\begin{subequations} 
\begin{align}
	\mG^{(\text{\rm PBC})}_{N,\infty} (t) &=  \frac{1}{N!} \int_{[0,2 \pi]} \frac{\dd \theta_1}{2 \pi} \cdots \int_{[0,2 \pi]} \frac{\dd \theta_N}{2 \pi} ~  \left\lvert \Delta_{U(N)} (\theta_1, \dots, \theta_N ) \right\rvert^2 ~ e^{- i 2 t \sum_{a=1}^{N} \cos \theta_a} , \label{eq:GU=GWWlong} \\
	\mG^{(\text{\rm ABC})}_{N,\infty} (t) &=  \frac{1}{N!} \int_{[0,\pi]} \frac{\dd \theta_1}{2 \pi} \cdots \int_{[0,\pi]} \frac{\dd \theta_N}{2 \pi} ~  \left\lvert \Delta_{Sp(N)} (\theta_1, \dots, \theta_N ) \right\rvert^2 ~ e^{- i 4 t \sum_{a=1}^{N} \cos \theta_a} , \label{eq:GSp=GWWlong}
\end{align}
\label{eq:G=MMlong}
\end{subequations}
where the Vandermonde factors are \cite{ForresterC}
\begin{subequations} 
\begin{align}
	\Delta_{U(N)} (\theta_1, \dots, \theta_N ) & = \prod_{1 \le a < b \le N}  2 \sin \left( \frac{\theta_a - \theta_b}{2} \right)  , \\
	\Delta_{Sp(N)} (\theta_1, \dots, \theta_N) &= \prod_{1 \le a < b \le N}  \left( \cos \theta_a - \cos \theta_a \right)  \prod_{a=1}^{N} 2 \sin \theta_a .
\end{align}
\end{subequations}
As claimed in the previous discussion, the equality between \eqref{eq:G=MMlong} and \eqref{eq:G=detlong} follows from identifying in the Vandermonde determinants a form suitable for applying Lemma \ref{lem:AI} to \eqref{eq:G=MMlong}, and computing the integrals 
\begin{equation*}
    \oint_{U(1)} \frac{\dd z}{2 \pi i z} z^{\nu} e^{i t (z+z^{-1})} = \frac{1}{i^{\nu}} J_{\nu} (-2t) .
\end{equation*}
An alternative derivation of the matrix model representation \eqref{eq:Gmatrixintegralform} is based on a Jordan--Wigner transformation, see for instance \cite{Stephan:2017}.\par
\begin{rmk}[Loschmidt echo and gauge theory]
   The matrix model \eqref{eq:GU=GWWlong} for the spin chain with PBC is closely related to the Gross--Witten--Wadia (GWW) model \cite{Gross:1980he,Wadia:1980cp,Wadia:2012fr}, originally introduced in two-dimensional lattice gauge theory. More precisely, the GWW model is the imaginary-time version of \eqref{eq:GU=GWWlong}, i.e. replacing $t \mapsto -i \beta$. The inverse temperature $\beta$ plays the role of the inverse Yang--Mills coupling in the gauge theory interpretation \cite{Gross:1980he}. The imaginary-time version of \eqref{eq:GSp=GWWlong} is the analogue of the GWW model with $USp(2N)$ gauge group.\par
   This gauge theory analogy will partly inspire the methods to study \eqref{eq:G=MMlong} in later sections.
\end{rmk}
From the expressions \eqref{eq:Gmatrixintegralform} we easily recover the universal, short-time behaviour of the Loschmidt echo.
\begin{thm}
    At leading order in the short time expansion, the Loschmidt echo has the universal behaviour
    \begin{equation*}
	\mL_{N,\infty} (t \ll 1 ) \approx 1- t^2 \left\langle \mathrm{Tr} \left( U+U^{-1} \right)^2 \right\rangle^{\mathrm{conn}.} _G ,
\end{equation*}
with the average over $G=U(N)$ for PBC and over $G=Usp(2N)$ for ABC. The superscript indicates the connected correlator
\begin{equation*}
    \left\langle A \cdot B \right\rangle^{\mathrm{conn}.} _G :=  \left\langle A \cdot B \right\rangle_G - \left\langle A \right\rangle_G \left\langle  B \right\rangle_G .
\end{equation*}
\end{thm}
\begin{proof}
Expanding the integrand in \eqref{eq:Gmatrixintegralform} at small $t$ one gets 
    \begin{equation*}
        \mG_{N, \infty} (\pm t) \approx \int_{G} \dd U ~ \left[ 1 \mp it  \mathrm{Tr} \left( U+U^{-1} \right) + \frac{1}{2} (\pm it)^2   \mathrm{Tr} \left( U+U^{-1} \right)^2 \right] .
    \end{equation*}
    Recall that we are using the normalized Haar measure on the Lie groups, so $\mG_N (0)=1$. The above expansion then implies
    \begin{align*}
        \mG_{N, \infty} (t) \mG_{N, \infty} (-t) & \approx \left[  1 - it \left\langle \mathrm{Tr} \left( U+U^{-1} \right) \right\rangle_G  + \frac{(-it)^2}{2} \left\langle \mathrm{Tr} \left( U+U^{-1} \right)^2 \right\rangle _G \right] \\
        & \quad \times \left[  1 + it \left\langle \mathrm{Tr} \left( U+U^{-1} \right) \right\rangle_G  + \frac{(it)^2}{2} \left\langle \mathrm{Tr} \left( U+U^{-1} \right)^2 \right\rangle _G \right] \\
        & = 1 + t^2 \left( \left\langle \mathrm{Tr} \left( U+U^{-1} \right) \right\rangle_G  \right)^2 - 2 \times \frac{t^2}{2} \left\langle \mathrm{Tr} \left( U+U^{-1} \right)^2 \right\rangle _G .
    \end{align*}
    Observing that 
    \begin{equation*}
        \mL_{N, \infty} (t) =  \mG_{N, \infty} (t) \mG_{N, \infty} (-t) 
    \end{equation*}
    concludes the proof.
\end{proof}

\subsection{Loschmidt echo on a finite chain}
		The expressions \eqref{eq:G=MMlong} hold for an infinitely long chain to which we apply the operators \eqref{eq:stringopsigma}. More precisely, \eqref{eq:G=MMlong} must be understood as the $L \to \infty$ limit of a chain of fixed length (or circumference) in which the number $L$ of qubits is progressively increased. Note that the actual length of the chain is irrelevant, because it drops out from any probability amplitude by normalization.\par
  \medskip
  The finite-$L$ matrix models, that give rise to \eqref{eq:G=MMlong} upon taking $L \to \infty $, are discrete analogues of \eqref{eq:G=MMlong}, in which the $N$ eigenvalues occupy $L$ discrete sites. The angular variables $\left\{ \theta_a \right\}_{a=1, \dots, N}$ are discretized into $\left\{ \theta (s_a) \right\}_{a=1, \dots, N}$, where 
		\begin{equation}
		\label{eq:discrthetaofs}
			\theta (s) = \frac{2 \pi }{L} (s-1), \qquad s \in \left\{ 1, \dots, L \right\} .
		\end{equation}
  \begin{itemize}
  \item[i)] In a finite chain of $L$ qubits with PBC, the circular topology $U(1)$ is replaced with the discrete topology
		\begin{equation*}
			 \left\{ z \in \mathbb{S}^1 \ : \ z^L = 1  \right\} \cong \mathbb{Z}_L ,
		\end{equation*}
  where we describe $\mathbb{Z}_L$ as a the multiplicative group $\left\{ e^{i 2 \pi (s-1)/L} , \ s=1 \dots, L \right\}$.
  \item[ii)] Likewise, in a finite chain of $L$ qubits with ABC, the semicircle topology $U(1)/\mathbb{Z}_2 ^{\mathsf{P}}$ is replaced by the multiplicative group $\left\{ e^{i \pi (s-1)/L} , \ s=1 \dots, L \right\}$. Here $\mathbb{Z}_2 ^{\mathsf{P}}$ is the parity symmetry $z \mapsto -z$.
\end{itemize}
Therefore, the matrix models for $L < \infty$ are the Riemann sums:
		\begin{subequations} 
\begin{align}
	\mG^{(\text{\rm PBC})}_{N,L< \infty} (t)  &=  \frac{1}{L^N N!} \sum_{s_1=1}^{L} \cdots \sum_{s_N=1}^{L}  \left\lvert \Delta_{U(N)} (\theta (s_1), \dots , \theta (s_N) ) \right\rvert^2 ~ e^{- i 2 t \sum_{a=1}^{N} \cos \theta (s_a) } , \label{eq:GU=discrlong} \\
	\mG^{(\text{\rm ABC})}_{N,L< \infty} (t)  &=  \frac{1}{L^N N!} \sum_{s_1=1}^{L} \cdots \sum_{s_N=1}^{L}  \left\lvert \Delta_{Sp(N)} (\theta (s_1), \dots , \theta (s_N) ) \right\rvert^2 ~ e^{- i 4 t \sum_{a=1}^{N} \cos \theta (s_a) } . \label{eq:GSp=discrlong}
\end{align}
\label{eq:G=MMdiscrlong}
\end{subequations}
		The discrete ensembles \eqref{eq:G=MMdiscrlong} admit a determinant presentation, as well. 
\begin{prop}
    For every $k \in \mathbb{Z}$ and $w \in \mathbb{C}$, let 
    \begin{equation*}
        \mathfrak{I}_k ( w ) := \frac{1}{L} \sum_{s=1}^{L} e^{- w \cos \theta (s)} e^{i k \theta (s) } .
    \end{equation*}
    Then 
    \begin{subequations} 
\begin{align}
	\mG^{(\text{\rm PBC})}_{N,L=\infty} (t) &=  \det_{1 \le a,b \le N } \left[ \mathfrak{I}_{a-b} (i 2t) \right],  \\
	\mG^{(\text{\rm ABC})}_{N,L=\infty} (t) &= \det_{1 \le a,b \le N } \left[   \mathfrak{I}_{a-b} (i4t) -  \mathfrak{I}_{a+b} (i4t) \right] .
\end{align}
\label{eq:discr=detlong}
\end{subequations}
\end{prop}
\begin{proof}
    It is a direct application of Andreief's identity (Lemma \ref{lem:AI}) with measure 
    \begin{equation*} 
    \mu (z) = e^{-it (z+z^{-1})} \frac{1}{L} \sum_{s=1} ^{L} \delta (z-e^{i \theta (s)} ) .
    \end{equation*}
\end{proof}\par
\begin{rmk}
    The presentations of the echos as determinants in \eqref{eq:G=detlong} and \eqref{eq:discr=detlong} is especially convenient for numerical evaluation at finite $N,L$ and arbitrary time $t$.
\end{rmk}
        \medskip
		For practical purposes, instead of working with \eqref{eq:GSp=GWWlong}-\eqref{eq:GSp=discrlong}, when considering the ABC chain we use the change of variables $x_a = \cos \theta_a$, which produces 
		\begin{align}
			\mG^{(\text{\rm ABC})}_{N,L=\infty} (t)  & = \frac{1}{N!} \int_{[-1,1]^N}  \prod_{1 \le a < b \le N}  (x_a-x_b)^2 ~ \prod_{a=1}^{N} e^{-  i 4 t  x_a }  \sqrt{ 1- x_a^2 } ~ \frac{\dd x_a}{2 \pi } , \label{eq:GSPreallong} \\
			\mG^{(\text{\rm ABC})}_{N,L<\infty} (t) &=  \frac{1}{L^N N!} \sum_{x_1 \in \scS_L } \cdots \sum_{x_N \in \scS_L }  \prod_{1 \le a < b \le N}  (x_a-x_b)^2 ~ \prod_{a=1}^{N} e^{- i 4 t  x_a}  \sqrt{ 1- x_a^2 } , \label{eq:GSPrealdiscr}
		\end{align}
		where the finite-$L$ model has domain 
		\begin{equation*}
			 \scS_L = \left\{ x \in [-1,1] \ : \ \left( x+ i \sqrt{1-x^2} \right)^L = 1  \right\} .
		\end{equation*}
  The function in the definition of $\scS_L$ comes from the logarithmic form of the arccosine function, and then writing $z=e^{i \mathrm{ArcCos} (x)}$.\par
		\medskip
  \begin{rmk}[Finite echo in the thermodynamic limit]
      We stress that a very important consequence of the simple quench setup we consider is the finitude of the Loschmidt echo in the thermodynamic limit $L \to \infty$. This is in sharp contrast with most quench protocols, in which one typically has $\ln \mL (t) \approx - L \tilde{f} (t)$ at $L \gg 1$, for some $L$-independent function $\tilde{f}$ \cite{Heyl:2017blm}.
  \end{rmk}

\subsection{Determinants of Bessel functions}
Looking back at the expressions \eqref{eq:G=detlong} and plotting them as functions of $t$ for several values of $N \in \mathbb{N}$, we uncover the following pattern:
\begin{itemize}
    \item For $N\in 2 \mathbb{N}+1$, $\mL_{N,L} (t)$ decays as $e^{-t^2}$ near $t \gtrsim 0$, then vanishes at a given $t=t^{\text{QSL}}_N$ and shows damped oscillations from that point on;
    \item For $N \in 2\mathbb{N}$, $\mL_{N,L} (t)$ decays as $e^{-t^2}$ near $t \gtrsim 0$, then reduces its slope and has an asymptotic decay to zero with small oscillations on top of that trend.
\end{itemize}
(See also \cite{Krapivsky:2017sua}, where such behaviour was first noticed). We stress that the behaviour of $\mL_{N,L=\infty} (t)$, including the oscillations for $N$ odd, are a consequence of the analytic properties of the Bessel functions $J_{\nu} (2t)$:
		\begin{itemize}
			\item Exponentially decreasing near $t \gtrsim 0$;
			\item Oscillating at $t \gg 1$.
		\end{itemize}
		Crucial is their assembly in the Toeplitz determinant \eqref{eq:GU=Tdetlong} or the Toeplitz$+$Hankel determinant \eqref{eq:GSp=THdetlong}.\par
For the most part of the body of this note, we consider $N \in 2 \mathbb{N}$, that is, an even number of spins $\lvert \downarrow \rangle$. The discussion of $N \in 2\mathbb{N}+1$ is deferred to Section \ref{sec:Nodd}.\par
\medskip		
		As a side remark, we notice that the factors of $i^{a-b}$ (where $i = \sqrt{-1}$) in \eqref{eq:GU=Tdetlong} appear in such a way that 
		\begin{equation}
		\label{eq:BesselwithandwithoutI}
			\mL_N ^{\text{(PBC)}} (t) = \left\lvert \det_{1 \le a,b \le N} \left[ J_{a-b} (2t) \right] \right\rvert^2
		\end{equation}
		$\forall \ N \in \mathbb{N}$, that is, they drop out of the determinant, cf. Figure~\ref{fig:BelleJdets}.\par
		\begin{figure}[ht]
			\centering
				\includegraphics[width=0.45\textwidth]{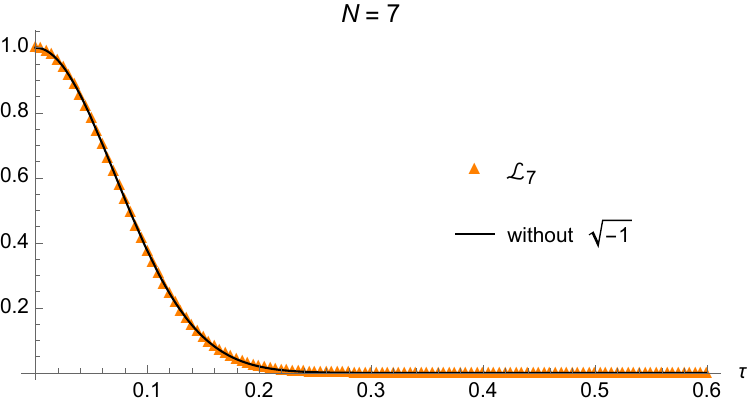}\hspace{0.05\textwidth}
				\includegraphics[width=0.45\textwidth]{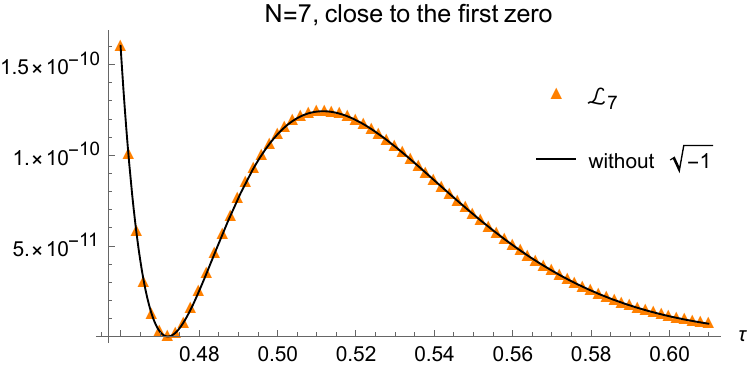}
			\caption{The identity \eqref{eq:BesselwithandwithoutI} at $N=7$.}
			\label{fig:BelleJdets}
			\end{figure}

\section{Dynamical quantum phase transition}
\label{sec:PT}
		In this section we discuss the large $N$ planar limit of the models \eqref{eq:GU=GWWlong}-\eqref{eq:GSPreallong}, and prove the existence of a third order dynamical quantum phase transition.\par
		For concreteness, we begin with the analysis of the $L=\infty$ chain, and come back to the experimentally accessible $L< \infty$ case in the next section.\par
  \medskip
  The content of this section is organized as follows.
  \begin{itemize}
    \item In Subsection \ref{sec:planarlimit} we introduce the planar limit and set up the problem at large $N$.
    \item Then, we take a detour to discuss the large $N$ third order phase transition taking place at imaginary time. In Subsection \ref{sec:PBCGWW} we review the large $N$ solution of $\mG_{N,\infty} ^{\text{(PBC)}} (-i \beta)$, which is the GWW model, and the associated third order phase transition. This is a review of known results, which we report here for completeness and to introduce the tools we will use in later subsections. After that, in Subsection \ref{sec:ABCGWW} we extend the result to the ABC chain in imaginary time, that is, to the large $N$ limit of $\mG_{N,\infty} ^{\text{(ABC)}} (-i \beta)$. The latter model was recently studied in \cite{Kimura:2022mavC}.
    \item We then proceed with the analysis of interest, with real time. We discuss the large $N$ planar limit of the chain with ABC in Subsection \ref{sec:DQPTABCLinfty} and with PBC in Subsection \ref{sec:DQPTPBCLinfty}. We will prove therein that the chain undergoes a third order DQPT. These two subsections form the core of the work and contain many of the main technical achievements.
  \end{itemize}
Because of their intrinsically more technical nature, each of the Subsections \ref{sec:PBCGWW} to \ref{sec:DQPTPBCLinfty} is segmented into several steps, to orient the reader through the computations.

\subsubsection*{Statement of the results}
    We will uncover a third order dynamical quantum phase transition in the planar large $N$ limit, $N,t \to \infty$ with $\tau= t/N$ fixed. The existence of this transition is insensitive to the choice of boundary conditions. Moreover, it will be manifest in the derivation that the phase transition we find \emph{is not} merely a rotation of the GWW phase transition from imaginary to real time. This is confirmed both by the critical value $\tau_{\text{cr}} = 0.33137 \dots$ and the explicit form of the dynamical free energy.\par
    Very much in the paradigm of dynamical quantum phase transitions \cite{Heyl:2017blm}, we observe a change in the analytic properties of the dynamical free energy in the complex plane.

  \subsection{Planar limit of matrix models}
  \label{sec:planarlimit}

  Let us begin by setting up large $N$ planar limit of the models \eqref{eq:GU=GWWlong}-\eqref{eq:GSPreallong} as well as their imaginary-time version. Standard reviews on the large $N$ limit of matrix models include \cite{DiFrancesco:1993C,Marino:2004C,Eynard:2015aeaC}. For unitary matrix models see e.g. \cite{Mandal:1989ryC,Jain:2013pyC,Santilli:2019wvq,Santilli:2020uehC}.\par
  To treat real and imaginary time at once, throughout the present subsection we define the quantities 
  \begin{subequations}
  \begin{align}
        \hat{\mG}_{N, \infty} ^{(\text{\rm PBC})} (w) &= \frac{1}{N!} \int_{[0, 2 \pi]^N}  \left\lvert \Delta_{U(N)} (\theta_1, \dots, \theta_N ) \right\rvert^2 ~ \prod_{a=1}^{N} \frac{\dd \theta_a}{2 \pi} e^{- w \sum_{k \ge 1} \frac{c_k}{k} \cos (k \theta_a) } , \label{eq:hatGPBC} \\
        \hat{\mG}_{N, \infty} ^{(\text{\rm ABC})} (w) &= \frac{1}{N!} \int_{[-1,1]^N}  \prod_{1 \le a<b \le N} (x_a-x_b)^2 ~ \prod_{a=1}^{N} \frac{\dd x_a}{2 \pi} \sqrt{1-x_a^2} ~e^{- 2 w \sum_{k \ge 1} \frac{c_k}{k} T_k (x_a) } , \label{eq:hatGABC}
  \end{align}   
  \label{eq:hatGxidef}
  \end{subequations}
  where $\left\{ c_k \right\}_{k \ge 1} $ is a finite collection of real coefficients and $\left\{ T_k \right\}_k$ are the Chebyshev polynomials of the first kind. For later reference, $\left\{ U_k \right\}_k$ will be the Chebyshev polynomials of the second kind. Here $w \in \C $ is a complex parameter, and the real- and imaginary-time dynamics are recovered substituting
  \begin{align*}
    \text{real time :}& \qquad w \mapsto i t \in i \mathbb{R}_{\ge 0}, \\
    \text{imaginary time :}& \qquad w \mapsto \beta \in \mathbb{R}_{\ge 0}.
  \end{align*}
  Moreover, the original models of interest correspond to the coefficients 
  \begin{equation*}
        c_k= 2 \delta_{k,1} . 
  \end{equation*}\par
  \medskip
  We consider the planar limit \cite{Brezin:1977sv} of the matrix models \eqref{eq:hatGxidef}. The planar limit is a large $N$ limit originally introduced by 't Hooft for the study of Quantum Chromodynamics \cite{tHooft:1973alwC} and has since then proved crucial in a wide variety of problems in high-energy physics and beyond. It owes its name to the fact that only planar Feynman diagrams are retained in this large $N$ limit.\par
We start by rewriting 
\begin{subequations}
		\begin{align}
			\hat{\mG}_{N, \infty} ^{(\text{\rm PBC})} (w) &= \frac{1}{N!} \int_{[0, 2 \pi]^N} \frac{\dd \theta_1 \cdots \dd \theta_N}{ (2\pi)^N} ~ \exp \left( - N^2 S_{\text{eff}} ^{(\text{\rm PBC})} (\theta_1, \dots, \theta_N) \right) , \\
   \hat{\mG}_{N, \infty} ^{(\text{\rm ABC})} (w) &= \frac{1}{N!} \int_{[-1,1]^N} \frac{\dd x_1 \cdots \dd x_N}{ (2\pi)^N} ~ \exp \left( - N^2 S_{\text{eff}} ^{(\text{\rm ABC})} (x_1, \dots, x_N) \right)
		\end{align}
  \label{eq:expintegrandMM}
  \end{subequations}
  where the ``effective actions'' are 
  \begin{align*}
      S_{\text{eff}} ^{(\text{\rm PBC})} (\theta_1, \dots, \theta_N) & = \frac{1}{N} \sum_{a=1} ^{N} \left[ \frac{w}{N} \left( \sum_{k \ge 1} \frac{c_k}{k} \cos (k \theta_a) \right) - \frac{1}{N} \sum_{b \ne a} \ln \left\lvert 2 \sin \left( \frac{\theta_a - \theta_b}{2} \right) \right\rvert \right] \\
       S_{\text{eff}} ^{(\text{\rm ABC})} (x_1, \dots, x_N) & = \frac{1}{N} \sum_{a=1} ^{N} \left[ \frac{2w}{N} \left( \sum_{k \ge 1} \frac{c_k}{k} T_k (x_a) \right) - \frac{1}{N} \ln \sqrt{1-x_a^2} - \frac{1}{N} \sum_{b \ne a} \ln \left\lvert x_a - x_b \right\rvert \right] .
  \end{align*}
  They are obtained by passing the Vandermonde determinants in the exponential.\par
  Observe that, in the limit $N \to \infty$, every sum over the indices $a,b=1, \dots, N$ grows linearly in $N$, thus the terms $\frac{1}{N} \sum_{a=1}^{N}$ have a well-defined large $N$ limit. For the two pieces in the effective action to compete at the same order in $N$, for $N \gg 1$, one is led to consider the \emph{'t Hooft scaling} of the parameter $w$, with 
  \begin{equation}
  \label{eq:generictHooft}
    \xi := \frac{w}{N} \text{ fixed as $N \to \infty$}.
  \end{equation}
\begin{rmk}
    In the original gauge theory framework, $w$ is inversely proportional to the squared Yang--Mills coupling $g_{\text{\tiny YM}}^2 $, so that $\frac{1}{\xi} \propto g_{\text{\tiny YM}}^2  N$ is what is customarily called \emph{'t Hooft coupling} in the gauge theory literature.\par
    For the imaginary-time dynamics, $w \mapsto \beta$, the planar limit corresponds to send the inverse temperature $\beta \to \infty$ linearly with $N$. In other words, the 't Hooft scaling intertwines the large $N$ and low temperature limits, instead of taking them independently.
\end{rmk}
We conclude that, with the scaling \eqref{eq:generictHooft}, the integrands in \eqref{eq:expintegrandMM} are suppressed at large $N$ as $e^{-N^2 S_{\text{eff}}} $, with the effective actions having a finite large $N$ limit. Therefore, in the large $N$ limit, the integrals \eqref{eq:expintegrandMM} are dominated by the stationary points of the integrands, which are the saddle points of the effective action.\par
This leads to the consideration of the system of $N$ saddle point equations:
\begin{subequations}
    \begin{align}
        \text{PBC :}& \qquad \qquad \frac{\partial  S_{\text{eff}} ^{(\text{\rm PBC})}}{\partial \theta_a} = 0 \qquad \forall \ a=1, \dots , N ; \\
         \text{ABC :}& \qquad \qquad \frac{\partial  S_{\text{eff}} ^{(\text{\rm ABC})}}{\partial x_a} = 0 \qquad \forall \ a=1, \dots , N.
    \end{align}
\end{subequations}
These are the Euler--Lagrange equations for \eqref{eq:expintegrandMM}. Explicitly, we have the saddle point equations 
\begin{align*}
   \text{PBC :}& \qquad \qquad  \frac{1}{N}\sum_{b \ne a} \cot \left( \frac{\theta_a - \theta_b}{2} \right)  = - \xi \sum_{k \ge 1} c_k \sin (k \theta_a) , \\
   \text{ABC :}& \qquad \qquad  \frac{2}{N}\sum_{b \ne a}  \frac{1}{x_a - x_b}  = 2 \xi \sum_{k \ge 1} c_k U_{k-1} (x_a) + \frac{1}{N} \frac{x_a}{\sqrt{1-x_a^2}} ,
\end{align*}
$\forall a= 1, \dots, N$. Note the factor of $2$ coming out of differentiating the double sum. In the PBC case, this factor of 2 cancels a factor $\frac{1}{2}$ from the derivative of the Vandermonde term.
\begin{rmk}
\label{remarkABCsub}
     Note that the square-root term in the ABC effective action is sub-leading at large $N$ and will not play a role in determining the saddle points. As a matter of fact, our results in the planar limit apply to any matrix model
     \begin{equation*}
         \frac{1}{N!} \int_{-1} ^{1} \frac{\dd x_1}{2\pi} \cdots \int_{-1} ^{1} \frac{\dd x_N}{2\pi} \prod_{1 \le a < b \le N} (x_a-x_b)^2 ~ \prod_{a=1}^{N} e^{-i4t x_a + \varphi (x_a) }
     \end{equation*}
     where $\varphi: [-1,1] \to \C$ 
     \begin{itemize}
         \item[($i$)] has no branch cuts on $(-1,1)$, and
         \item[($ii$)] has coefficients that do not depend on $N$.
     \end{itemize}
     As particular instances, this accounts for replacing the integral over $USp(2N)$ in \eqref{eq:GSp=GWWlong} with integration over $SO(2N)$ or $SO(2N+1)$.
\end{rmk}\par
\medskip
In the large $N$ limit, the discrete index $a=1,\dots,N$ can be effectively replaced by a continuous index 
\begin{equation*}
    \mathsf{a} = \frac{a}{N}, \qquad 0 <\mathsf{a} \le 1 .
\end{equation*}
This substitution implements the replacement 
\begin{equation*}
    \frac{1}{N} \sum_{a=1} ^{N} F (\theta_a)  \ \mapsto \ \int_0 ^1 \dd \mathsf{a} F (\theta (\mathsf{a})) ,
\end{equation*}
for arbitrary function $F$, where the $N$ eigenvalues are grouped into a function $\theta (\mathsf{a})$ defined according to 
\begin{equation*}
    \theta (\mathsf{a}) = \theta _{\lfloor \mathsf{a}N \rfloor } .
\end{equation*}
The procedure is exactly analogous for the substitution $x_a \mapsto x (\mathsf{a}) = x_{\lfloor \mathsf{a}N \rfloor }$ in the ABC case.\par
One then arrives at the saddle point equations $\forall \  0< \mathsf{a} \le 1$: 
\begin{subequations}
    \begin{align}
  \text{PBC :}& \qquad \qquad  \dashint_{0} ^1 \dd \mathsf{b} \cot \left( \frac{\theta (\mathsf{a})  - \theta (\mathsf{b}) }{2} \right)  = - \xi \sum_{k \ge 1} c_k \sin (k \theta \mathsf{a} ) , \\
   \text{ABC :}& \qquad \qquad  \dashint_0 ^1 \dd \mathsf{b}  \frac{1}{x (\mathsf{a}) - x (\mathsf{b}) } = \xi \sum_{k \ge 1} c_k U_{k-1} (x (\mathsf{a}) ),
\end{align}
\end{subequations}
up to $O(1/N)$ corrections, that are negligible in the large $N$ limit. The symbol $\dashint $ means the Cauchy principal value integral.\par
\medskip
Clearly, it is not handy to deal with a system of $N$ equations in the large $N$ limit. For this reason, it is customary to introduce the eigenvalue densities 
		\begin{align*}
			\rho^{(\text{\rm PBC})} (\theta) = \frac{1}{N}\sum_{a=1} ^{N} \delta (\theta - \theta_a) , \\
			\rho^{(\text{\rm ABC})} (x) = \frac{1}{N}\sum_{a=1} ^{N} \delta (x-x_a) ,
		\end{align*}
		(we will omit the superscript when no confusion can arise) which are normalized by definition and have compact support at large $N$. Besides, it follows from the definition together with the discussion above that
  \begin{equation}
    \lim_{N \to \infty} \rho^{(\text{\rm PBC})} (\theta) = \left( \frac{ \dd \theta (\mathsf{a}) }{\dd \mathsf{a}} \right)^{-1} .
    \label{eq:defrhoderivative}
  \end{equation}
  Thus, in practice, we read the above equation as $\rho^{(\text{\rm PBC})} (\theta)  \dd \theta = \dd \mathsf{a} $ and enforce the replacement 
  \begin{equation*}
    \int_{0} ^{1} \dd \mathsf{b} F(\theta (\mathsf{b})) \ \mapsto \ \int \dd \theta \rho^{(\text{\rm PBC})} (\theta) F(\theta) ,
  \end{equation*}
  and likewise for $\rho^{(\text{\rm ABC})} (x)$.\par
  Putting all the pieces together, we are finally led to the saddle point equations 
 \begin{subequations}
 \begin{align}
		\text{PBC:}& \qquad \qquad \dashint \dd \varphi \rho (\varphi) ~ \cot \left( \frac{\theta- \varphi}{2} \right) = - \xi \sum_{k \ge 1} c_k \sin (k\theta) , \label{eq:SPEgenericPBC} \\
		\text{ABC:}& \qquad \qquad \dashint \dd y \rho (y) ~ \frac{1}{x-y} = \xi \sum_{k \ge 1} c_k U_{k-1} (x) .  \label{eq:SPEgenericABC}
		\end{align}  
\label{eq:SPEgenericmodel}
\end{subequations}
These are integral equations to be solved for the eigenvalue densities $\rho^{(\text{\rm PBC})}$ and $\rho^{(\text{\rm ABC})}$. For the isotropic XY Heisenberg spin chain we are interested with, these equations reduce to 
\begin{subequations}
\begin{align}
		\text{PBC:}& \qquad \qquad \dashint \dd \varphi \rho (\varphi) ~ \cot \left( \frac{\theta- \varphi}{2} \right) = - 2 \xi \sin (\theta) , \label{eq:SPEPBC} \\
		\text{ABC:}& \qquad \qquad \dashint \dd y \rho (y) ~ \frac{1}{x-y} = 2 \xi . \label{eq:SPEABC}
		\end{align}
  \label{eq:SPEchain}
\end{subequations}
The substitutions $w = \beta$ and $w=it$ give the imaginary-time and real-time dynamics, respectively. Because we are dealing with the planar limit, we introduce the scaled quantities 
\begin{subequations}
    \begin{align}
        \text{imaginary time :}& \qquad \qquad \gamma := \frac{\beta}{N} \text{ fixed}; \\
        \text{real time :}& \qquad \qquad \tau := \frac{t}{N} \text{ fixed}.
    \end{align}
\end{subequations}\par
To conclude the computation, we recall from \eqref{eq:expintegrandMM} that, at leading order in the large $N$ limit,
\begin{equation*}
    \hat{\mG}_{N, \infty} (w= \xi N) \approx e^{-N^2 S_{\text{eff}} [\rho]} ,
\end{equation*}
for either choice of boundary condition. In the latter expression, $S_{\text{eff}} [\rho]$ means the effective action evaluated on the eigenvalue density that solves the pertinent equation in \eqref{eq:SPEgenericmodel}. Therefore, it follows by construction of the planar limit that 
\begin{equation*}
    \lim_{N \to \infty} \frac{1}{N^2} \ln \hat{\mG}_{N, \infty} (w= \xi N) = S_{\text{eff}} [\rho] ,
\end{equation*}
with the right-hand side finite.
\begin{rmk}[Normalization of the free energy]
    It follows from this general argument that, in the planar limit, the matrix models we consider behave as $\ln \mG (t) \propto N^2 f(\tau)$, with $f(\tau)$ a continuous function of $\tau=t/N$, independent of $N$. Therefore, we normalize the dynamical free energy of the system according to 
		\begin{equation}
		\label{eq:defDFEtaulargeN}
			f (\tau) = - \frac{1}{2 N^2} \ln \mL_{N,\infty} (t) ,
		\end{equation}
		in agreement with the standard normalization in the random matrix theory literature.\par
    With the aim of reducing confusion among imaginary-time and real-time quantities, we define for imaginary time 
			\begin{equation}
   \label{eq:defGWWFE}
				\mathcal{F} (\gamma) = \lim_{N \to \infty} \frac{1}{N^2} \ln  \mG_{N, \infty} (-i \beta) ,
			\end{equation}
   where $\gamma = \beta/N$ and 
   \begin{equation*}
    \mG_{N, \infty} (-i \beta) = \langle \psi_0 \lvert e^{-\beta H_{\text{XY}}} \rvert \psi_0 \rangle .
   \end{equation*}
	Note the opposite sign in \eqref{eq:defGWWFE} with respect to the dynamical free energy \eqref{eq:defDFEtaulargeN}.
\end{rmk}
Before solving the saddle point equations \eqref{eq:SPEchain} for both imaginary and real time and any choice of boundary conditions, let us note the following.
\begin{lem}
    Let $\rho (z) $ be the eigenvalue density associated to a matrix model with bounded integration domain $\mathscr{D}$. Then 
    \begin{equation*}
        \rho (z) \ge 0 \qquad \forall z \in \mathscr{D} .
    \end{equation*}
\end{lem}
\begin{proof}
    The boundedness of $\mathscr{D}$ by definition implies 
    \begin{equation*}
        d:= \sup_{z_1,z_2 \in \mathscr{D}} \lvert z_1-z_2 \rvert \qquad  0 \le d < \infty .
    \end{equation*}
    Therefore 
    \begin{equation*}
        \frac{ \lvert z_a - z_b \rvert }{\lvert a-b \rvert } \le  \lvert z_a - z_b \rvert \le d  
    \end{equation*}
    which, replacing $a= \mathsf{a} N$ and rearranging, yields 
    \begin{equation*}
        \frac{\lvert \mathsf{a} - \mathsf{b} \rvert }{\lvert z(\mathsf{a} )- z(\mathsf{b}) \rvert } \ge \frac{1}{Nd} >0 .
    \end{equation*}
    Sending $N \to \infty$ and $\mathsf{b} \to \mathsf{a}$ using the characterization \eqref{eq:defrhoderivative}, we conclude that $\rho (z) \ge 0$.
\end{proof}
Because both matrix models we consider have compact support, the eigenvalue densities $\rho ^{(\text{\rm PBC})} (\theta) $ and $\rho ^{(\text{\rm ABC})} (x) $ must be non-negative definite. Solutions to \eqref{eq:SPEchain} that do not satisfy such constraint are to be discarded.

	\subsection{Review of the GWW third order phase transition}
	\label{sec:PBCGWW}
		The celebrated Gross--Witten--Wadia (GWW) model \cite{Gross:1980he,Wadia:1980cp,Wadia:2012fr} is the imaginary-time version of \eqref{eq:GU=GWWlong}, equivalently \eqref{eq:hatGPBC} with $w=\beta$ and $c_k = 2 \delta_{k,1}$. It was originally devised in the study of two-dimensional gauge theory on the lattice. The term $\exp \left[ - \beta \mathrm{Tr} (U + U^{-1}) \right]$ was derived from the Wilson lattice action for a one-plaquette model of lattice two-dimensional gauge theory \cite{Gross:1980he}.\par
		In this subsection we review the derivation of its third order phase transition, as a warm up to present the saddle point techniques. The reader familiar with the subject may move directly to the next subsection.\par
		\medskip
		To begin with, we make a change of variables $\theta_a \mapsto \theta_a - \pi$ for later convenience, which shifts the integration domain to $[-\pi, \pi]^N$ and reflects $\beta \mapsto - \beta$. According to the general discussion of Subsection \ref{sec:planarlimit}, the planar limit of the model is governed by the saddle point equation 
		\begin{equation}
		\label{eq:SPEGWW}
			\dashint_{-\pi} ^{\pi} \dd \varphi ~ \rho (\varphi) \cot \left( \frac{\theta- \varphi}{2} \right) = 2 \gamma \sin \theta  ,
		\end{equation}
		where we recall that $\gamma = \beta/N$ plays the role of the inverse 't Hooft coupling.

		\subsubsection*{Statement of the result for the eigenvalue density}
		The solution to \eqref{eq:SPEGWW} is given by \cite{Gross:1980he,Wadia:2012fr}
		\begin{equation}
				\rho (\theta) = \begin{cases} \frac{1}{2\pi} \left[  1 + 2 \gamma \cos \theta \right] \id_{-\pi \le \theta \le \pi} & \gamma \le \frac{1}{2} , \\ \frac{2 \gamma}{\pi} \cos \frac{\theta}{2} \sqrt{ \frac{1}{2 \gamma}  - \left( \sin \frac{\theta}{2} \right)^2 } \id_{-\theta_0 \le \theta \le \theta_0 } & \gamma > \frac{1}{2} \end{cases}
			\label{eq:completesolrhoGWW}
			\end{equation}
			where $\id$ is the indicator function and $\theta_0$ in the phase $\gamma > \frac{1}{2}$ is the positive angle with $\sin \frac{\theta_0}{2} = \sqrt{1/(2 \gamma )}$.
		
		\subsubsection*{Solution in the first phase}
		In order to solve \eqref{eq:SPEGWW} for $\rho (\theta)$, we begin with the ansatz
		\begin{equation}
		\label{eq:ansatz1}
		    \text{supp} \rho = [-\pi, \pi ] .
		\end{equation}
		Plugging the expansion 
		\begin{equation*}
		    \cot \left( \frac{\theta - \varphi}{2} \right) = 2 \sum_{k=1} ^{\infty } \left[ \cos (k\theta) \sin (k \varphi) - \sin (k\theta) \cos (k \varphi)  \right] ,
		\end{equation*}
		valid for $\varphi \ne \theta $, in \eqref{eq:SPEGWW} we get 
		\begin{equation*}
		    \frac{1}{\pi}\dashint_{-\pi} ^{\pi} \dd \varphi ~ \rho (\varphi) \sum_{k=1} ^{\infty } \left[ \cos (k\theta) \sin (k \varphi) - \sin (k\theta) \cos (k \varphi)  \right] =  \frac{\gamma}{\pi} \sin \theta .
		\end{equation*}
		Using the ansatz \eqref{eq:ansatz1} we can perform the integrals on the left-hand side and get a generating function for the Fourier coefficients of $\rho (\theta)$. Equating with the right-hand side one finds 
		\begin{equation}
		\label{eq:GWWrhoPI}
		    \rho (\theta) = \frac{1}{2\pi} +  \frac{\gamma}{\pi} \cos \theta .
		\end{equation}
		The $0^{\text{th}}$ Fourier coefficient is not fixed by the saddle point equation, but rather is fixed by normalization together with the ansatz \eqref{eq:ansatz1}.\par
		It is straightforward to extend the solution thus obtained to the saddle point equation \eqref{eq:SPEgenericPBC}, finding 
		\begin{equation*}
		    \rho (\theta) = \frac{1}{2\pi} \left[ 1 + 2 \gamma \sum_{k \ge 1} c_k \cos (k \theta) \right] .
		\end{equation*}\par
		The solution \eqref{eq:GWWrhoPI} is subject to the constraint 
		\begin{equation*}
		    \rho (\theta) \ge 0 \qquad \forall -\pi \le \theta \le \pi ,
		\end{equation*}
		which is violated around $\theta = \pm \pi$ for $\gamma > \frac{1}{2}$, and around $\theta=0$ for $\gamma < - \frac{1}{2}$. Because $\gamma$ is a scaled inverse temperature, we only consider the physically meaningful region $\gamma > 0$. Therefore, at $\gamma > \frac{1}{2}$, the solution \eqref{eq:GWWrhoPI}is invalid, indicating a phase transition. We must drop the assumption \eqref{eq:ansatz1} and look for a different eigenvalue density.
		
		\subsubsection*{Solution in the second phase}
		Giving up on the assumption \eqref{eq:ansatz1}, we then make the ansatz that $\rho (\theta) $ is supported on (one or more) arc(s) on the unit circle.\par
		To this aim, it is standard procedure to introduce the \emph{planar resolvent} 
		\begin{equation}
		\label{eq:defplanarres}
			\omega (z) = \int \frac{\dd u}{u} ~ \varrho (u) \frac{z+u}{z-u} , \qquad z \in \mathbb{C} \setminus \text{supp} \varrho ,
		\end{equation}
		where $\text{supp} \varrho = \Gamma \subset \mathbb{S}^1$ and $\varrho (e^{i \theta})= \rho (\theta)$. By construction, 
		\begin{equation}
		\label{eq:rhofromomega}
			4 \pi \rho( \theta) = \omega_{+} (e^{i \theta} ) - \omega_{-} (e^{i \theta} ) ,
		\end{equation}
		where we are adopting the standard shorthand notation 
		\begin{equation*}
		    \omega_{\pm} (z) := \lim_{\varepsilon \to 0^{+}}  \omega (e^{i \theta} \pm i \varepsilon ) . 
		\end{equation*}
		Then, expressing \eqref{eq:SPEGWW} in terms of the exponential variable $z= e^{i \theta} \in \Gamma$, yields 
		\begin{equation}
		\label{eq:jumpSPEomega}
		    \frac{1}{2} \left[ \omega_{+} (z) + \omega_{-} (z) \right] = - i \gamma \left( z + z^{-1} \right) , \qquad z \in \Gamma .
		\end{equation}
		This is a jump equation for $\omega (z)$ along $\Gamma$, supplemented with the condition 
		\begin{equation*}
		    \lim_{\lvert z \rvert  \to \infty} \omega(z) = 1 ,
		\end{equation*}
		which follows from the normalization of the eigenvalue density together with the definition \eqref{eq:defplanarres}. We have thus set up a Riemann--Hilbert problem, whose solution governs the planar limit of the GWW model in the phase $\gamma > \frac{1}{2}$.\par
		The solution for $\omega (z)$ is derived in a standard way. The jump equation \eqref{eq:jumpSPEomega} implies a contour integral representation for the planar resolvent, with integration contour encircling $\Gamma$. The final result then follows from standard contour manipulations, see e.g. \cite{Jain:2013pyC,Santilli:2019wvq,Santilli:2020uehC} for the details.\par
		
		\subsubsection*{Free energy}
			Equipped with the solution \eqref{eq:completesolrhoGWW} for the eigenvalue density, we can compute $\mathcal{F} (\gamma) $ as defined in \eqref{eq:defGWWFE}. More precisely, it is convenient to compute 
			\begin{equation*}
				\frac{ \dd \mathcal{F}^{\text{(PBC)}} }{ \dd \gamma } = 2 \int \dd \theta \rho (\theta ) \cos (\theta) = \begin{cases}  2 \gamma & 0 \le \gamma \le \frac{1}{2} \\ 2- \frac{1}{2\gamma} & \gamma \ge \frac{1}{2} . \end{cases} 
			\end{equation*}
			The result is that $\mathcal{F}^{\text{(PBC)}} (\gamma)$ is continuous and twice differentiable, but it fails to be smooth at $\gamma= \frac{1}{2}$. This is the third order GWW phase transition.

	\subsection{A GWW transition for the ABC chain in imaginary time}
	\label{sec:ABCGWW}
		We now extend the above result to the imaginary-time version of the ABC. 
	\subsubsection*{Statement of the result}	
	The outcome of the analysis is: the model has a third order transition at the same critical value $\gamma_{\text{cr}}= \frac{1}{2}$ as the imaginary-time PBC chain (i.e. the original GWW model). Moreover
		\begin{equation}
		\label{eq:ABCequals2PBC}
			\mathcal{F}^{\text{(ABC)}} = 2 \mathcal{F}^{\text{(PBC)}} .
		\end{equation}
		The result has been recently obtained in \cite{Kimura:2022mavC} using slightly different methods.\par
		
	\subsubsection*{Solution in the first phase}
		Taking the large $N$ limit \emph{after} the change of variables $x_a=\cos \theta_a$ yields the imaginary-time version of \eqref{eq:SPEABC}:
		\begin{equation}
		\label{eq:SPEABCGWWlong}
			\dashint \dd y \rho (y) \frac{1}{x-y} = - 2 \gamma .
		\end{equation}
		In fact, we can easily solve the more general problem 
		\begin{equation*}
			\dashint_{-1}^{1} \dd y ~ \rho (y) \frac{1}{x-y} = - \gamma \sum_{k \ge 1} c_k U_{k-1} (x) ,
		\end{equation*}
		where $U_{k-1}$ are the Chebyshev polynomials of second kind. From the Stieltjes transform of the Chebyshev polynomials 
		\begin{equation*}
			\dashint_{-1}^{1} \dd y  \frac{T_k (y)}{\pi (x-y)\sqrt{1-y^2}} = -  U_{k-1} (x) 
		\end{equation*}
		we infer 
		\begin{equation*}
			\rho (x) = \frac{1}{\pi \sqrt{1-x^2}} \left[ 1+ \gamma \sum_{k \ge 1} c_k T_k (x) \right] .
		\end{equation*}
		Specializing to the case of interest, we have 
		\begin{equation}
		\label{eq:rhoPISpNGWW}
			\rho (x) = \frac{1 +2 \gamma x }{\pi \sqrt{1-x^2}} , \qquad -1 \le x \le 1 .
		\end{equation}
		Again, the solution holds for $0 \le \gamma \le \frac{1}{2}$. Beyond the critical value, the solution ceases to be non-negative close to the edge $x=-1$. Therefore, the GWW transition becomes a hard-edge to soft-edge transition in this setup.
		
	\subsubsection*{Solution in the second phase}
		We look for a so-called ``soft-edge'' density, in which the inverse square-root singularity near $x=-1$ is replaced by a square-root behaviour near a boundary $A>-1$.\par
		The solution to \eqref{eq:SPEABCGWWlong} under these circumstances is known (see e.g. \cite{DeiftBookC}) and goes under the name of Tricomi's formula \cite{TricomiC}. The final answer in the case of interest reads 
		\begin{equation}
			\rho (x)= \begin{cases} \frac{1 +2 \gamma x }{\pi \sqrt{1-x^2}}  \id_{-1 < x < 1 } & 0 \le \gamma \le \frac{1}{2} \\ \frac{2 \gamma}{\pi} \sqrt{\frac{x-A}{1-x}} \id_{A \le x \le 1 } & \gamma > \frac{1}{2} ,\end{cases}
			\label{eq:completesolrhoABCGWW}
		\end{equation}
		where $A= \frac{\gamma-1}{\gamma}$ is fixed by normalization. 
	\begin{rmk}
	The eigenvalue density \eqref{eq:completesolrhoABCGWW} is precisely the image of the GWW density \eqref{eq:completesolrhoGWW}, upon restriction to $[0,\pi]$ and mapping $\cos \theta = x$.
	\end{rmk}
	
	\subsubsection*{Free energy}
		Armed with the eigenvalue density \eqref{eq:completesolrhoABCGWW} we can easily compute $\frac{\dd \mathcal{F}^{\text{(ABC)}} }{\dd \gamma }$, as above. Direct integration gives 
		\begin{equation*}
			\mathcal{F}^{\text{(ABC)}} = 2 \mathcal{F}^{\text{(PBC)}} 
		\end{equation*}
		as advertised in \eqref{eq:ABCequals2PBC} and in agreement with \cite{Kimura:2022mavC}.\par
		In the first phase, the relation \eqref{eq:ABCequals2PBC} is a corollary of Johansson's  extension of Szeg\H{o} strong limit theorem \cite{Johansson:1997C}. Our result is to prove it beyond the phase transition.
		\begin{thm}[\cite{Szegoth,Johansson:1997C}]\label{thm:Szego}
		    Let $\left\{ c_k \right\}_{k \ge 1}$ be a collection of coefficients satisfying 
		    \begin{equation*}
		        \sum_{k\ge 1} \lvert c_k \rvert < \infty \qquad \text{and} \qquad \sum_{k\ge 1} k \lvert c_k \rvert^2 < \infty ,
		    \end{equation*}
		    and define 
		    \begin{align*}
		    \mathcal{Z}_N ^{(\text{\rm PBC})} & := \int_{U(N)} \dd U ~\exp \left[ \sum_{k\ge 1} c_k \mathrm{Tr} (U^k + U^{-k}) \right] , \\
		        \mathcal{Z}_N ^{(\text{\rm ABC})} & := \int_{USp(2N)} \dd U ~\exp \left[ \sum_{k\ge 1} c_k \mathrm{Tr} (U^k + U^{-k}) \right] .
		    \end{align*}
		    It holds that 
		    \begin{align*}
		        \lim_{N \to \infty}  \ln  \mathcal{Z}_N ^{(\text{\rm PBC})} & = \sum_{k \ge 1} k \lvert c_k \rvert ^2 , \\
		        \lim_{N \to \infty} \ln  \mathcal{Z}_N ^{(\text{\rm ABC})} & = 2 \sum_{k \ge 1} k \lvert c_k \rvert ^2  .
		    \end{align*}
		\end{thm}
		\begin{proof}
		The $U(N)$ part of the theorem dates back to Szeg\H{o} \cite{Szegoth}. For its extension to other classical Lie groups, in particular $USp(2N)$, see \cite{Johansson:1997C,Garcia-Garcia:2019uveC}, but be aware of the difference in normalization.\par
		Szeg\H{o}'s theorem and its extension consider the strict $N \to \infty$ limit, with all other parameters fixed. The agreement with the value in the first phase follows by real-analyticity arguments, see \cite{Santilli:2019wvq} for an explicit proof.
		\end{proof}

	\subsubsection*{Folding trick}
		The spin chain model we are considering endows us with a visual, albeit non-rigorous, argument for the factor of 2 in formula \eqref{eq:ABCequals2PBC}.\par
		Imagine to begin with a chain with PBC and $2N+1$ flipped spins, yielding a $U(2N+1)$ matrix integral. Because of the $N^2$-scaling, we have 
		\begin{equation*}
			\mathcal{F}^{\text{(PBC)}} \vert_{N \mapsto 2N+1} = 2^2 \mathcal{F}^{\text{(PBC)}} \vert_{N} 
		\end{equation*}
		at leading order in $N$. Then, we remove the middle one among the $2N+1$ qubits and replace it with an absorbing wall. At the same time, we cut open the chain at the antipodal point, thus obtaining two copies of a chain with $N$ flipped spins and ABC on the left, OBC on the right. Because the two effects of introducing boundaries/interfaces are expected to be sub-leading in $N$, we predict 
		\begin{equation*}
			\mathcal{F}^{\text{(PBC)}} \vert_{N \mapsto 2N+1} = 2 \mathcal{F}^{\text{(ABC)}} \vert_{N } .
		\end{equation*}
		Combining the last two equations we obtain \eqref{eq:ABCequals2PBC}, which is of course only valid at leading order in the planar limit. It is indeed known \cite{Johansson:1997C}, and in agreement with our argument, that the ABC model has $1/N$ corrections, as opposed to the $1/N^2$ corrections in the PBC model.\par
		We emphasize that the spin chain argument predicting \eqref{eq:ABCequals2PBC} is general and does not rely on the specific couplings of the model, as long as it admits a description in terms of matrix integrals. As a matter of fact, the relation holds in the real-time dynamics as well.\par
		\medskip
		To conclude the digression into imaginary-time dynamics and the associated third order phase transition, we compare in Figure~\ref{fig:thermalGWW} the finite $N$ results with the asymptotic formula. The plot shows the perfect agreement of our asymptotic expressions, derived analytically, with the numerical exact result at large but finite $N$.
		\begin{figure}[ht]
			\centering
				\includegraphics[width=0.45\textwidth]{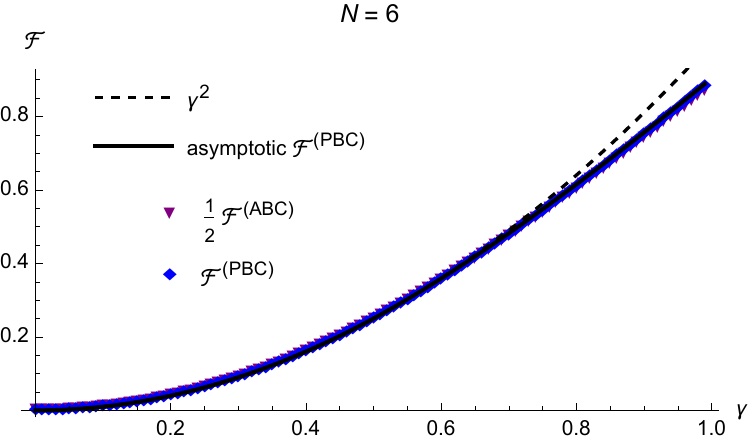}\hspace{0.08\textwidth}
				\includegraphics[width=0.45\textwidth]{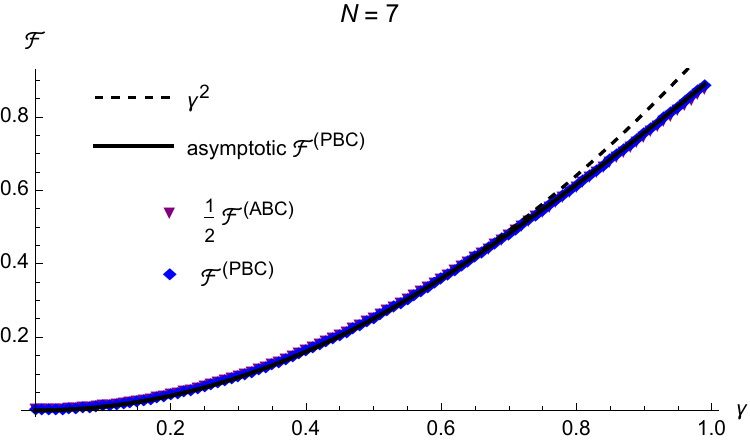}
			\caption{Plot of $\mathcal{F}=\frac{1}{2N^2} \ln \langle \psi_0 \lvert e^{- \beta H_{\text{XY}}} \rvert \psi_0 \rangle  $ at $N=6$ (left) and $N=7$ (right).}
			\label{fig:thermalGWW}
			\end{figure}
	
	\subsection{Dynamical quantum phase transition: ABC}
    \label{sec:DQPTABCLinfty}
		At this stage we can go back to the main subject: solving the real-time dynamics at leading order in $N \gg 1$. We begin from the chain with ABC and solve \eqref{eq:SPEABC} with $\xi = i \tau $ first, that is 
		\begin{equation}
		\label{eq:SPEABCit}
		    \dashint \dd y \frac{\rho (y) }{x-y} = i 2 \tau 
		\end{equation}
		The chain with PBC is discussed in the next subsection.\par
		Because the right-hand side of \eqref{eq:SPEABCit} is imaginary, the standard techniques do not apply, as it is obvious that they would produce a real-valued left-hand side. We therefore need to extend the methods utilized above to account for the imaginary values of the parameters in the effective action.\par
		We begin with a bird-eye discussion of complex actions and their saddle points, and then go on and solve the planar limit of the Loschmidt echo for the chain with ABC.\par
		\begin{rmk}[A minus sign]\label{rmk:sign}
		To avoid dragging cumbersome factors of $(-1)$ in the ensuing computations, we make a change of variables $\tilde{x}_a=-x_a$ in the matrix model \eqref{eq:GSPreallong} (and drop the tilde henceforth). In practice, it is tantamount to solve \eqref{eq:SPEABCit} with replacement $\tau \mapsto - \tau $.
		\end{rmk}

		\subsubsection*{Complex saddles}
		We have noted that the saddle point equation \eqref{eq:SPEABCit} requires a complex solution for $\rho (x)$. It follows from the effective action 
		\begin{align*}
		    \lim_{N \to \infty} S_{\text{eff}} ^{(\text{\rm ABC})} (x_1, \dots, x_N) & = \lim_{N \to \infty} \left[ - \frac{1}{N^2} \sum_{1 \le a\ne b \le N} \ln \left\lvert x_a-x_b \right\rvert + \frac{1}{N} \sum_{a=1} ^{N} i 4 \tau x_a \right] \\
		    & = i 4 \tau \int_0 ^1 \dd \mathsf{a} ~ x (\mathsf{a}) - \dashint_{\mathsf{b} \ne \mathsf{a}} \dd \mathsf{a} \dd \mathsf{b} ~ \ln \left\lvert  x (\mathsf{a}) -  x (\mathsf{b}) \right\rvert \\
		    & = i 4 \tau \int \dd x \rho (x) ~x - \dashint \dd x \dd y \rho (x) \rho (y) ~ \ln \lvert x-y \rvert 
		\end{align*}
		from \eqref{eq:GSPreallong} having an imaginary coefficient. Therefore, in the planar large $N$ limit, $\mG_{N,L=\infty} (t) \approx e^{-N^2 S_{\text{eff}} ^{(\text{\rm ABC})} [\rho] }$ is dominated by complex saddle point configurations.\par
		It is known that, in general, effective actions in Quantum Field Theory and in matrix models may admit complex saddle points. We refer to \cite{Garcia:1996npC,Guralnik:2007rxC,Ferrante:2013hgC} for extensive discussion, and \cite{Marino:2008yaC,Marino:2012zqC} for a treatment of the complex saddle points in matrix models. An important feature of the models \eqref{eq:GU=GWWlong}-\eqref{eq:GSPreallong} is that the complex saddles of the effective actions are the dominant ones in the planar limit.\par
		\begin{rmk}
		The fact that the saddle point configuration may be complex does not automatically imply that the action evaluated on that saddle point is complex. In fact, it turns out not to be the case for the models we consider.
		\end{rmk}
		In conclusion, to solve \eqref{eq:SPEABCit} we gain insight from the Quantum Field Theory approach and relax the condition on the integration domain, so that we let it pass through the complex saddle points.\par
		\begin{rmk}[Replica wormholes]
		The inclusion of complex saddles in the gravitational path integral is playing a prominent role in black hole physics. The recent efforts to solve the Hawking paradox and the factorization problem build on the refined analysis of complex saddle points, known as \emph{replica wormholes} \cite{Penington:2019kkiC,Almheiri:2019qdqC}. It has been proposed that these wormholes can become the dominant saddle point, against the Hawking saddle point configuration. This change of dominant saddle is accompanied with a first order phase transition, consistent with the Page curve. See \cite{Almheiri:2020cfmC} for an overview.
		\end{rmk}

		\subsubsection*{Contour deformation and reality condition}
		In the light of the above discussion, in practice we allow
		\begin{equation*}
			\text{supp} \rho = \Gamma \subset \C ,
		\end{equation*}
		i.e. the saddle point configuration of the matrix model \eqref{eq:GSPreallong} places the eigenvalues on some arc or union of arcs $\Gamma $ in the complex plane. See e.g. \cite{David:1990skC} for an early discussion and \cite{Copetti:2020dilC} for a recent application in the context of black holes.\par
		Of course, consistency requires 
		\begin{equation*}
		    \lim_{\tau \to 0^{+}} \Gamma = [-1,1] 
		\end{equation*}
		and we will henceforth only focus on solutions $\Gamma$ that, for small enough values of $\tau$ (in a sense made precise momentarily) are homotopic to $[-1,1]$.\par
		Having relaxed the assumption on $\text{supp} \rho$, we write \eqref{eq:SPEABCit} as 
		\begin{equation}
		\label{eq:cplexSPEABC}
			\dashint_{\Gamma} \dd y \frac{\rho (y)}{z-y} = - i 2 \tau  \qquad \forall z \in \Gamma ,
		\end{equation}
		where now $z,y \in \C$ need not be real. Notice the minus sign in the right-hand side compared to \eqref{eq:SPEABCit}, according to Remark \ref{rmk:sign}.\par
		We define a coordinate $\sigma \in \Gamma $ and parametrize the embedding $\Gamma \hookrightarrow \C$ through $\sigma \mapsto z(\sigma)$, which we normalize according to $\lvert \dot{z} (\sigma) \rvert^2 =1 $. The integrals below are understood in such a parametrization, which we leave implicit.\par
		\medskip
		After having obtained a distribution $\rho (x)$ that solves \eqref{eq:cplexSPEABC}, we determine the shape of $\Gamma$ requiring the probabilities
		\begin{equation}
		\label{eq:realitycondGamma}
			0 \le \int_{z_0} ^{z_1} \dd z \rho (z) \le 1 , \qquad \forall z_0,z_1 \in \Gamma .
		\end{equation}\par
		 So, in particular, $\Gamma$ must be a level set along which the integral is real-valued.\par

		\subsubsection*{Solution in the first phase}
		 Gaining insight from the imaginary-time solution in Subsection \ref{sec:ABCGWW}, we expect that the solution in the first phase is the analytic continuation of the imaginary-time solution.\par
		 Indeed, under the ansatz that $\Gamma$ is homotopic to $[-1,1]$, the function 
		\begin{equation}
		\label{eq:rhoABCrealtPI}
			\rho (x) = \frac{1+i2 \tau x}{\pi \sqrt{1-x^2}}
		\end{equation}
		solves \eqref{eq:cplexSPEABC}. It can be checked directly by plugging \eqref{eq:rhoABCrealtPI} in the left-hand side of \eqref{eq:cplexSPEABC} and solving the integral using the assumption on $\Gamma$.\par
		\begin{prop}
		    The 1-cycle $\Gamma $ is the component homotopic to the interval $[-1,1]$ of the level set
		    \begin{equation}
		\label{eq:defGammacplexABC}
			\Re \left[ \log \left( -iz + \sqrt{1-z^2} \right) -2\tau \sqrt{1-z^2} \right] =0 .
		\end{equation}
		\end{prop}
		\begin{proof}
		The integral in \eqref{eq:realitycondGamma} with the solution \eqref{eq:rhoABCrealtPI} is easily solved and reads 
		\begin{equation*}
			\pi \int ^{z} \dd z \rho (z) = \mathrm{ArcSin} (z) - i 2 \tau \sqrt{1-z^2} \qquad \forall z \in \Gamma 
		\end{equation*}
		where $\mathrm{ArcSin}$ is the arcsine function, i.e. the extension of the inverse sine function $\sin^{-1} : [-1,1] \to [0, \pi]$ to $\C$. Using the logarithmic form 
		\begin{equation*}
		    \mathrm{ArcSin} (z) = i \log \left( -i z + \sqrt{1-z^2} \right) ,
		\end{equation*}
		the reality constraint becomes 
			\begin{equation*}
			\Im \left[ i \log \left( -iz + \sqrt{1-z^2} \right) -i2\tau \sqrt{1-z^2} \right] =0 ,
		\end{equation*}
		which is \eqref{eq:defGammacplexABC} upon multiplication by $\sqrt{-1}$.\par
		Equation \eqref{eq:defGammacplexABC} admits disconnected solutions and, by construction, we ought to retain the one component that is smoothly deformed into $[-1,1] \subset \R \subset \C$.
		\end{proof}
		
		\subsubsection*{Cross-check: consistency with known results}
		Notice that, taking $\tau \to 0$, the condition \eqref{eq:realitycondGamma} reduces to 
		\begin{equation*}
			0 \le \frac{1}{\pi} \left[ \mathrm{ArcSin} (z) -\mathrm{ArcSin} (-1) \right] \le 1 ,
		\end{equation*}
		where we have chosen as initial point $z_0=-1$. The latter equation is solved by $z \in [-1,1]$, thus
		giving back the contour 
		\begin{equation*}
		    \lim_{\tau \to 0^{+}} \Gamma = [-1,1] 
		\end{equation*}
		as it should.\par
		Furthermore, replacing $i \tau \mapsto \gamma $ in \eqref{eq:realitycondGamma}, that is, applying the argument to the imaginary-time dynamics of the system, the solution is still given by $z \in [-1,1]$ (cf. Figure~\ref{fig:realitycondGWWABC}), consistently with the standard procedure and the result of Subsection \ref{sec:ABCGWW}.\par
		\begin{figure}[ht]
		\centering
			\includegraphics[width=0.4\textwidth]{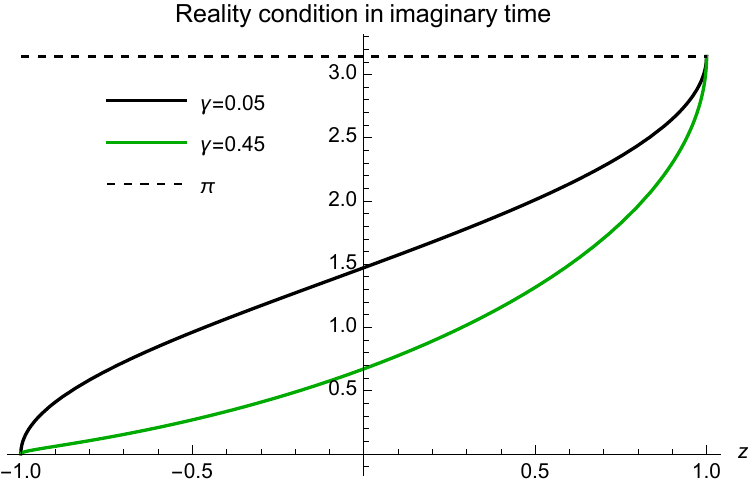}
		\caption{Solution to the imaginary-time version of the reality condition \eqref{eq:realitycondGamma} multiplied by $\pi$, for $\gamma=0.05$ (black) and $\gamma=0.45$ (green).}
		\label{fig:realitycondGWWABC}
		\end{figure}
		
		\subsubsection*{Criticality}
		At every $\tau >0$, the contour \eqref{eq:defGammacplexABC} intersects the imaginary axis in the upper half plane.\par
		The eigenvalue density $\rho (x)$ in \eqref{eq:rhoABCrealtPI} has a complex zero at 
		\begin{equation*}
			z_0 (\tau) = \frac{i}{2 \tau} ,
		\end{equation*}
		that is, placed in the positive imaginary semi-axis. Increasing $\tau >0$, $z_0 (\tau)$ descends from $i \infty$ and eventually intersects $\Gamma$ at the critical value $\tau_{\text{cr}}$, which is the (lowest positive) solution to 
		\begin{equation*}
			\Re \left[ \log \left( -i z_0 (\tau_{\text{cr}}) + \sqrt{1-z_0 (\tau_{\text{cr}})^2} \right) -2\tau \sqrt{1-z_0 (\tau_{\text{cr}}) ^2} \right] =0 ,
		\end{equation*}
		which, explicitly, is
		\begin{equation}
		\label{eq:defeqtaucritical}
			\sqrt{1 + 4 \tau_{\mathrm{cr}} ^2 } - \log \left(\frac{  1 \pm \sqrt{1 + 4 \tau_{\mathrm{cr}} ^2 } }{2 \tau_{\mathrm{cr}} } \right) = 0 .
		\end{equation}
		The solution is 
		\begin{equation*}
			\tau_{\mathrm{cr}} \approx 0.33137 .
		\end{equation*}
		At $\tau > \tau_{\mathrm{cr}}$, $\Gamma$ breaks up into a two-cut solution, with the gap in the support opening around 
		\begin{equation*}
		z=z_{0} (\tau_{\mathrm{cr}}) \approx i 1.51. 
		\end{equation*}
		See Figure~\ref{fig:ABCrealtimeGamma}. From the random matrix theory perspective, this is a one-cut to two-cut phase transition, with the peculiarity that the cut $\Gamma$ is deformed away from the real axis. This is distinct from the hard-edge to soft-edge transition we have seen for the imaginary-time (finite temperature) model in Subsection \ref{sec:ABCGWW}.\par
		\begin{figure}[ht]
		\centering
			\includegraphics[width=0.4\textwidth]{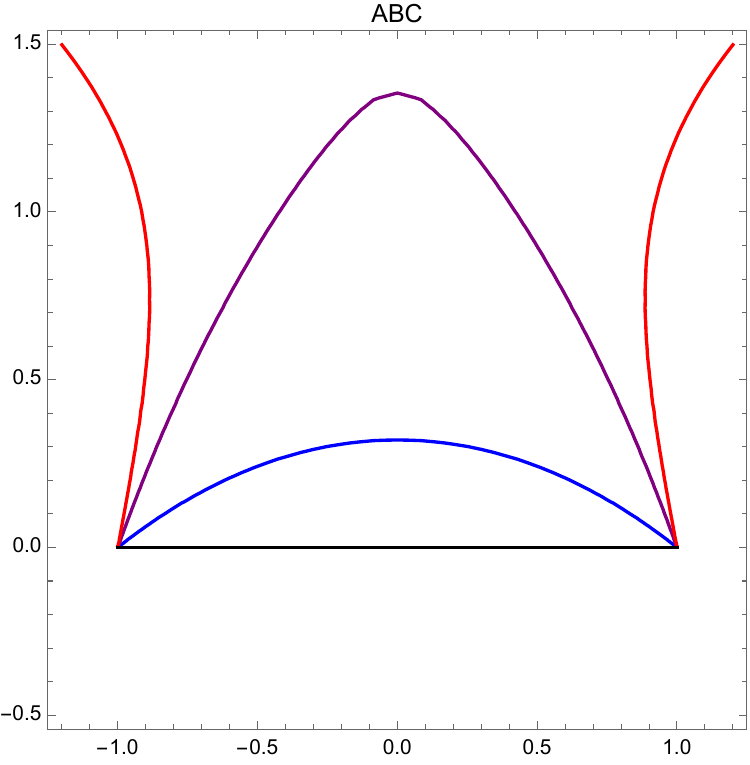}
		\caption{Solution to \eqref{eq:defGammacplexABC} homotopic to $[-1,1]$ at $\tau=$ 0(black), 0.15(blue), 0.33(purple), 0.4(red).}
		\label{fig:ABCrealtimeGamma}
		\end{figure}
		
	\subsubsection*{Solution in the second phase}
		In the new phase, the two disconnected arcs forming $\Gamma $ join $-1$ to $A$ and $B$ to $+1$, respectively, for a pair of points $(A,B)$ in the upper-half $\C$-plane. The reflection symmetry of the problem imposes $\Im A = \Im B$ and $\Re A = - \Re B$. Moreover, 
		\begin{equation*}
			\lim_{\tau \to  \tau_{\mathrm{cr}}} A = \lim_{\tau \to  \tau_{\mathrm{cr}}} B = \frac{i}{2 \tau_{\mathrm{cr}}} 
		\end{equation*}
		follows from the continuity of $\rho (x)$ at the critical point.\par
		In the new phase, we can set up a Riemann--Hilbert problem for the planar resolvent $\omega (z)$, similar to what was done in the second phase in imaginary-time, and solve it adapting the standard techniques to the present case.\par
		We omit the technical details, which are involved but standard. One can follow \cite{DiFrancesco:1993C} step by step for the two-cut phase, noting that the procedure is solely based on complex contour deformations, which go through identically when the endpoints $A,B$ lie in $\C$ and not in $\R$.\par
		The upshot is that we find the two-cut solution 
		\begin{equation}
		\label{eq:rhoABCrealtPII}
			\rho (x) = i \frac{2 \tau}{\pi} \sqrt{  \frac{(A-x)(x-B)}{1-x^2}} 
		\end{equation}
		and $\Im A = \Im B= \frac{i}{2\tau}$, while $\Re A = - \Re B$ is implicitly fixed by the normalization condition and the knowledge of $\Im A$.\par
		Let us mention that, \emph{a posteriori}, one may have leveraged the insight from two-cut matrix models with hard edges supported on subsets of $\R$, as for example in \cite{Cunden:2018C}, to derive \eqref{eq:rhoABCrealtPII}. Indeed, the na\"{i}ve ansatz for the resolvent would be precisely the planar resolvent $\omega (z)$ associated to \eqref{eq:rhoABCrealtPII}. Then, running the argument backwards, by selecting a closed contour of integration encircling $\Gamma$ one shows that such $\omega (z)$ solves the Riemann--Hilbert problem descending from \eqref{eq:SPEABCit}.
		
	\subsubsection*{Dynamical free energy}
		Having obtained the eigenvalue density, we can use it to compute the dynamical free energy \eqref{eq:defDFEtaulargeN}. In the first phase we have 
		\begin{align*}
			\left. \frac{\dd f^{\text{(ABC)}}}{\dd \tau} \right\rvert_{0 < \tau \le \tau_{\text{cr}}} & = - 4i \int_{\Gamma} \dd x ~ \rho (x) ~x \\
			& -4i \dashint_{-1} ^{1} \dd x \frac{x(1+i 2 \tau x)}{\pi \sqrt{1-x^2}}=  4 \tau ,
		\end{align*}
		where we use the form \eqref{eq:rhoABCrealtPI} of the eigenvalue density. To pass to the second line we have used that $\Gamma$ is homotopic to $[-1,1]$.\par
		\medskip
		We have not been able to obtain a closed form expression for $f(\tau)$ in the phase $\tau > \tau_{\text{cr}}$. Nevertheless, we can leverage the continuity of $ \frac{\dd f^{\text{(ABC)}}}{\dd \tau}$ at the critical point, which follows from the continuity of $\rho (x)$ and $\Gamma $. We then compute the expansion 
		\begin{subequations}
		\begin{align}
			A &=\frac{i}{2 \tau }-\left( 1-\frac{\tau_{\mathrm{cr}}}{ \tau}\right) +O(( \tau-\tau_{\mathrm{cr}})^{2}), \\
			B &=\frac{i}{2 \tau}+\left( 1-\frac{\tau_{\mathrm{cr}}}{ \tau}\right) +O(( \tau-\tau_{\mathrm{cr}})^{2}) .
		\end{align}
		\label{eq:approxABtaucr}
		\end{subequations}
		With these expressions at hand, we can compute $\left. \frac{\dd^2 f^{\text{(ABC)}}}{\dd \tau^2 } \right\rvert_{\tau > \tau_{\text{cr}}}  $ close to the critical point, which is enough to obtain the order of the phase transition.
		\begin{prop}
		    The dynamical free energy $f (\tau)$ is twice differentiable and its third derivative is discontinuous at $\tau=\tau_{\text{cr}}$.
		\end{prop}
		\begin{proof}
		We start by computing the second derivative of $f (\tau)$ close to the critical value $\tau_{\text{cr}}$:
		\begin{align*}
			\left. \frac{\dd^2 f^{\text{(ABC)}}}{\dd \tau^2 } \right\rvert_{\tau > \tau_{\text{cr}}}   \approx - 4i \frac{\dd \ }{\dd  \tau} \left( \frac{i 2 \tau x}{\pi }  \right) & \left[ \int_{-1}^{\frac{i}{ 2 \tau}-\left( 1-\frac{\tau_{\mathrm{cr}}}{ \tau} \right) } x \dd x  \sqrt{\frac{\left( \frac{i}{ 2 \tau}-\left( 1-\frac{\tau_{\mathrm{cr}}}{ \tau}\right) -x\right) \left( x-\frac{i}{ 2 \tau}-\left( 1-\frac{\tau_{\mathrm{cr}}}{ \tau}\right) \right) }{1-x^{2}}} 	 \right. \\
				& \left. +\int_{\frac{i}{ 2 \tau}+\left( 1-\frac{\tau_{\mathrm{cr}}}{ \tau}\right) }^{1}  x \dd x \sqrt{\frac{\left( \frac{i}{ 2 \tau}-\left( 1-\frac{\tau_{\mathrm{cr}}}{ \tau}\right) -x\right) \left( x-\frac{i}{ 2 \tau}-\left( 1-\frac{\tau_{\mathrm{cr}}}{ \tau}\right) \right) }{1-x^{2}}}   \right] 
		\end{align*}
		We have used the two-cut eigenvalue density \eqref{eq:rhoABCrealtPII} replacing $A$ and $B$ with their approximate form \eqref{eq:approxABtaucr}. Thus, we first differentiate and then evaluate the resulting integrals. The contributions involving $\frac{ \dd A}{\dd \tau }$ and $\frac{ \dd B}{\dd \tau }$ vanish because they multiply the integrand evaluated at the boundary, that vanishes. We finally get 
		\begin{equation*}
			\lim_{\tau \to \tau_{\text{cr}}^{+}} \frac{\dd^2 f^{\text{(ABC)}}}{\dd \tau^2 } = 4 .
		\end{equation*}
		The same argument can be improved to compute directly the first derivative of the dynamical free energy in the second phase. We obtain the result 
		\begin{align*}
			\left. \frac{\dd f^{\text{(ABC)}}}{\dd \tau } \right\rvert_{\tau > \tau_{\text{cr}}}   & \approx  \frac{ 8 \tau}{\pi} \left[ \int_{-1} ^{1} \dd x - \int_{\frac{i}{2 \tau}  -1 + \frac{\tau_{\mathrm{cr}}}{ \tau} } ^{\frac{i}{2 \tau}  +1 - \frac{\tau_{\mathrm{cr}}}{ \tau} } \dd x \right] x \sqrt{\frac{\left( \frac{i}{ 2 \tau}-\left( 1-\frac{\tau_{\mathrm{cr}}}{ \tau}\right) -x\right) \left( x-\frac{i}{ 2 \tau}-\left( 1-\frac{\tau_{\mathrm{cr}}}{ \tau}\right) \right) }{1-x^{2}}} \\
				& \approx \frac{ 8 \tau}{\pi}  \left\{  \int_{-1} ^{1} \dd x  \frac{x}{\sqrt{1-x^2}} \left( x - \frac{i}{2 \tau} \right) \left[ 1- \left( 1 - \frac{\tau_{\mathrm{cr}}}{ \tau} \right)^2 \frac{1 }{2 \left( x - \frac{i}{2 \tau} \right)^2 } \right] \right. \\
				& \ \qquad \qquad \left. -2 \left( 1 - \frac{\tau_{\mathrm{cr}}}{ \tau} \right)^2 \left[  \frac{i x}{\sqrt{1-x^2}} \right]_{x= \frac{i}{2 \tau}} \right\} \\
				& = 4 \tau - 8 \left( 1 - \frac{\tau_{\mathrm{cr}}}{ \tau} \right)^2 \left( 1 + 4 \tau_{\mathrm{cr}} ^2 \right)^{-\frac{3}{2}}  ,
		\end{align*}
		the symbol ``$\approx$'' indicating that we only retain the lowest non-trivial order in $(\tau_{\text{cr}} - \tau)$. We conclude that the phase transition is at least third order. By applying the same ideas to $\left. \frac{\dd^3 f^{\text{(ABC)}}}{\dd \tau^3 } \right\rvert_{\tau > \tau_{\text{cr}}}  $ we have explicitly checked that it is non-zero at $\tau_{\text{cr}}$.
		\end{proof}
		We conclude that the dynamical free energy signals a third order DQPT.\par
	
	\subsubsection*{Late time behaviour of the dynamical free energy}
		We can study the late time asymptotics $\tau \to \infty$ of the dynamical free energy $f (\tau)$. In this regime, the symmetry arguments are enough to constraint the form of $A$ and $B$ at leading order in $1/\tau$. One finds
		\begin{equation*}
			A = -1 +\frac{i}{2 \tau} + O(\tau^{-2}) , \qquad B = 1 +\frac{i}{2 \tau}  + O(\tau^{-2}) .
		\end{equation*}
		and a direct computation akin to the ones above yields 
		\begin{equation*}
			\left. \frac{\dd f^{\text{(ABC)}}}{\dd \tau } \right\rvert_{\tau > \tau_{\text{cr}}}  \propto \frac{1}{\tau} , \qquad \tau \to \infty ,
		\end{equation*}
		which means the dynamical free energy behaves logarithmically at late times.\par
		
	\subsubsection*{Relation with earlier works}
		Before concluding the present analysis, it is worthwhile to mention that the transition we find is essentially the one first discovered in \cite{DeanoC}, and further elaborated upon in \cite{CelsusDetC,Celsus:2020C}. Indeed, even though the model we are interested in is \eqref{eq:GSp=GWWlong}, after changing variables into \eqref{eq:GSPreallong} and applying Andr\'eief's identity, we land on a Hankel determinant which only differs by the one considered in \cite{DeanoC,CelsusDetC,Celsus:2020C} by terms which are sub-leading in $N$, cf. also Remark \ref{remarkABCsub}. Not surprisingly, then, we find agreement with the results therein, although using different mathematical techniques, namely a saddle point analysis in the complex plane instead of the asymptotics of a system of orthogonal polynomials.\par

	\subsection{Dynamical quantum phase transition: PBC}
    \label{sec:DQPTPBCLinfty}
    
    \subsubsection*{Statement of the result}
		We now solve the saddle point equation for the chain with PBC and prove the existence of a third order phase transition at the critical point \eqref{eq:defeqtaucritical}. Moreover, in agreement with the general argument presented at the end of Subsection \ref{sec:ABCGWW}, we will find by direct computation the relation 
		\begin{equation*}
			f^{\text{(ABC)}} (\tau) = 2 f^{\text{(PBC)}} (\tau) .
		\end{equation*}\par
		We consider \eqref{eq:SPEPBC} and rewrite it in exponential variables $z= e^{i \theta}$, $u=e^{i \varphi}$:
		\begin{equation}
			\dashint \frac{\dd u}{u} ~ \varrho (u) \frac{z+u}{z-u} = \tau \left( z - \frac{1}{z} \right) .
		\end{equation}
		Recall that the distribution $\varrho $ is related to $\rho$ through $\varrho (e^{i \theta}) = \rho (\theta)$.\par
		As in the real-time ABC chain, we ought to relax the assumption on the integration contour, in order to catch the contribution from the complex saddles of the model \eqref{eq:GU=GWWlong}.\footnote{Note that we do not consider a meromorphic deformation of the model at finite $N$, but rather take into account the large $N$ complex saddles of the undeformed model. Had we considered a meromorphic matrix model as in \cite{Santilli:2021eonC}, the resulting saddle point equation would be different.}\par
		
	\subsubsection*{Solution in the first phase}
		We begin with the ansatz $\text{supp} \varrho = \Gamma$ where, as opposed to the imaginary-time model of Subsection \ref{sec:PBCGWW} (i.e. the GWW model), we let $\Gamma$ be a 1-cycle in $\C$ homotopic to the unit circle $\mathbb{S}^1$. Topologically, $[\Gamma] \in H_1 (\C^{\ast})$ but the geometry (i.e. the shape) of the curve $\Gamma$ depends on $\tau$ and is subject to the constraint 
		\begin{equation*}
		    \lim_{\tau\to 0^{+}} \Gamma = \mathbb{S}^1 .
		\end{equation*}\par
		With this ansatz, and gaining insight from the imaginary-time result in Subsection \ref{sec:PBCGWW}, we obtain that 
		\begin{equation}
		\label{eq:rhoPBCrealtPI}
			\varrho (z) = \frac{1}{2 \pi} \left[ 1 + i \tau \left( z + \frac{1}{z} \right) \right] 
		\end{equation}
		solves the saddle point equation. We also observe that 
		\begin{equation*}
			\oint_{\Gamma} \frac{\dd u }{2 \pi u } \left[ 1 + i \tau \left( u + \frac{1}{u} \right) \right] \left( \frac{z+u}{z-u}  \right) = \begin{cases} -i + 2 \tau z & z \in \text{Int}(\Gamma) \\ i - 2 \tau /z & z \in \text{Ext}(\Gamma) \end{cases} 
		\end{equation*}
		where $\text{Int}(\Gamma)$ (resp. $\text{Ext}(\Gamma)$) is the interior (resp. exterior) of the closed Jordan curve $\Gamma$.\par
		\begin{prop}
		    The 1-cycle $\Gamma$ is the unique connected component homotopic to $\mathbb{S}^1$ of the level set 
		\begin{equation}
		\label{eq:defGammacplexPBC}
			\Im \left\{ -i \ln (z) + \tau \left( z - \frac{1}{z} \right)  \right\} =0 .
		\end{equation}
		\end{prop}
		\begin{proof}
		    It follows from the analogue of the reality condition \eqref{eq:realitycondGamma} by direct computation.
		\end{proof}
		We note that \eqref{eq:defGammacplexPBC} is manifestly invariant under the $\mathbb{Z}_2 ^{\mathsf{P}}$ symmetry $z \mapsto \frac{1}{z}$ and the time reversal $\mathbb{Z}_2 ^{\mathsf{T}}$ symmetry $(\tau, z) \mapsto (-\tau, -z) $ of \eqref{eq:GU=GWWlong}.\par
		
		\subsubsection*{Criticality}
		The solution \eqref{eq:rhoPBCrealtPI} for $\varrho (z)$ has zeros in $\C$ located at 
		\begin{equation}
			z_{\pm} (\tau) = \frac{i}{2 \tau} \left( 1 \pm \sqrt{1+4 \tau^2} \right) 
		\end{equation}
		and a phase transition takes place at the lowest positive value of $\tau$ at which either of these points hits $\Gamma$. The criticality condition is then 
		\begin{equation*}
			\Im \left\{ -i \ln (z_{\pm} (\tau_{\text{cr}})) + \tau_{\text{cr}} \left( z_{\pm} (\tau_{\text{cr}}) - \frac{1}{z_{\pm} (\tau_{\text{cr}})} \right)  \right\} =0 ,
		\end{equation*}
		which, using elementary manipulations, becomes 
		\begin{equation*}
			\sqrt{1 + 4 \tau_{\mathrm{cr}} ^2 } - \log \left(\frac{  1 \pm \sqrt{1 + 4 \tau_{\mathrm{cr}} ^2 } }{2 \tau_{\mathrm{cr}} } \right) = 0 .
		\end{equation*}
		This is exactly the same equation for $\tau_{\text{cr}}$ found in the ABC chain, cf. \eqref{eq:defeqtaucritical}. The phase transition then takes place at exactly the same critical value, as expected.\par
		A direct study shows that, as $\tau$ is increased, $\Gamma$ is progressively deformed away from $\mathbb{S}^1$ and eventually pinches at $\tau= \tau_{\text{cr}}$. See Figure~\ref{fig:PBCrealtimeGamma}.\par
		\begin{figure}[ht]
		\centering
			\includegraphics[width=0.4\textwidth]{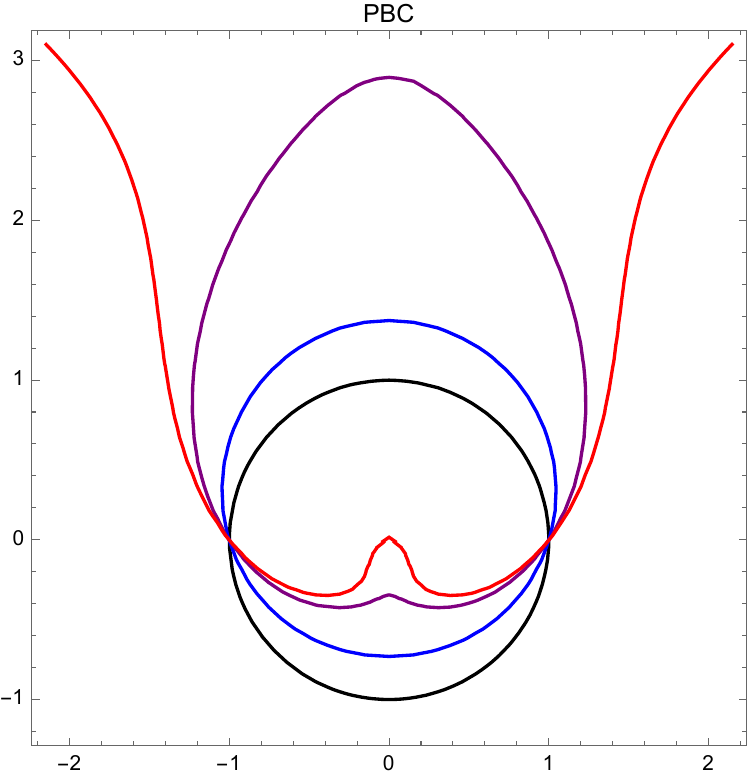}
		\caption{Solution to \eqref{eq:defGammacplexPBC} homotopic to $\mathbb{S}^1$ at $\tau=$ 0(black), 0.15(blue), 0.33(purple), 0.4(red).}
		\label{fig:PBCrealtimeGamma}
		\end{figure}\par

	\subsubsection*{Dynamical free energy}
		As in the ABC chain, we are unable to find a closed form expression for $f(\tau)$ in the phase $\tau > \tau_{\text{cr}}$. Nevertheless, we can mimic the computations above and evaluate
		\begin{equation*}
			\left. \frac{\dd f^{\text{(PBC)}}}{\dd \tau} \right\rvert_{0 < \tau \le \tau_{\text{cr}}} =2 \tau 
		\end{equation*}
		using \eqref{eq:rhoPBCrealtPI}. Besides, approximating near $\tau_{\text{cr}}$ in the second phase, we are able to prove that the transition is third order. The derivation is almost identical to the previous subsection.\par
		We conclude by showing $f(\tau)$ in Figure~\ref{fig:DynQuantFreeEnergy}. For comparison with the imaginary-time phase transition, we plot $f^{(\text{\rm PBC})} (\tau)$ and $\mathcal{F}^{(\text{\rm PBC})} (\gamma)$ together in Figure~\ref{fig:fvsFplot}. The plots make manifest that the dynamical free energy has a much more marked discrepancy from the parabola $\tau^2$ than its imaginary-time analogue, in agreement with the logarithmic behaviour at late times derived analytically.
		\begin{figure}[ht]
			\centering
				\includegraphics[width=0.45\textwidth]{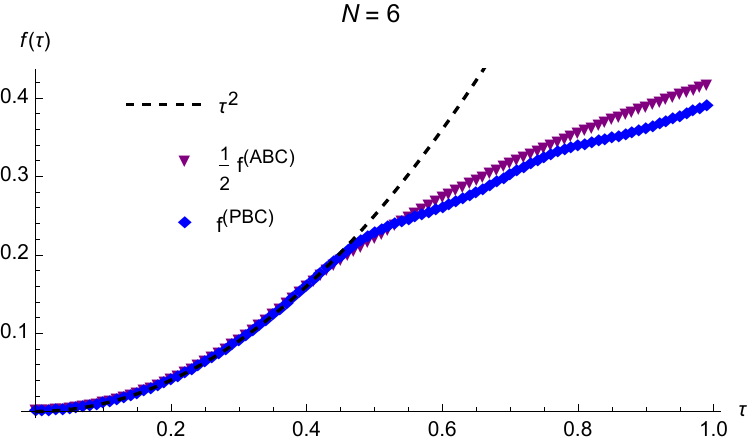}\hspace{0.08\textwidth}
				\includegraphics[width=0.45\textwidth]{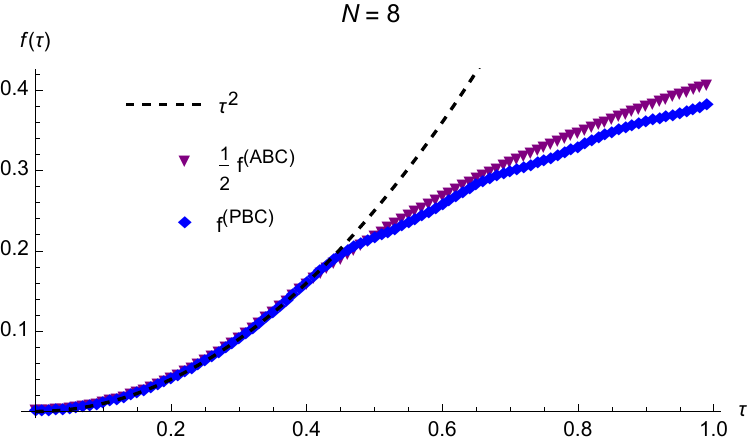}
			\caption{Plot of $f(\tau)  $ at $N=6$ (left) and $N=8$ (right).}
			\label{fig:DynQuantFreeEnergy}
			\end{figure}
			\begin{figure}[htb]
			\centering
				\includegraphics[width=0.65\textwidth]{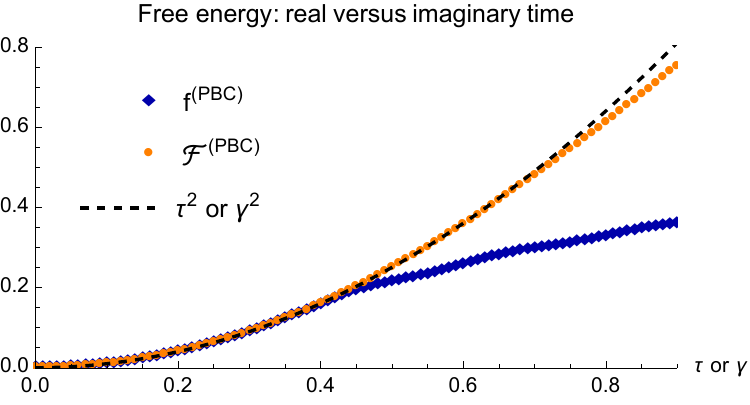}
			\caption{Comparative plot of $f(\tau) $ and $\mathcal{F} (\gamma)$ at $N=8$.}
			\label{fig:fvsFplot}
			\end{figure}

	\subsubsection*{Relation with earlier works}
	    Expression \eqref{eq:GU=GWWlong} for the Loschmidt echo was thoroughly analyzed in \cite{Krapivsky:2017sua}. However, that work did not consider the planar limit. Our result for $f^{\text{(PBC)}} (\tau) $ precisely interpolates between the large $N$ (at fixed $t$) and large $t$ (at fixed $N$) regimes of \cite{Krapivsky:2017sua}.
	    
	\subsubsection*{Mathematical status of the derivation}
        Before concluding this section, we highlight the extent to which the derivation of the DQPT is mathematically rigorous. To begin with, the map between the Loschmidt amplitude and the matrix model is an exact mathematical identity. For the most part, our manipulations based on RMT are also mathematically rigorous, especially all the statements in Section \ref{sec:S4} below.\par 
        However, we do not provide a complete proof of the convergence of $f (\tau)$ to its saddle point value. This prevents our main result from being a full-fledged theorem. Nevertheless, as explained in Remark \ref{rmk:convergence}, we substantiate our procedure by cross-checking the outcome with the existing literature. In particular, at the end of Subsection \ref{sec:DQPTABCLinfty} we have successfully compared our result on the ABC chain with the formulas in the mathematical works \cite{DeanoC,CelsusDetC,Celsus:2020C}, finding perfect agreement.
 
	\subsubsection*{Double-scaling limit and universality class}
	    It is an intriguing open mathematical problem to rigorously establish the universality class to which the DQPT belongs.\par
	    Setting up a Riemann--Hilbert problem, following e.g. in \cite{BDJ:1999,Baik:2000C}, one may take a double-scaling limit in which $\tau \to \tau_{\text{cr}} $  as $N $ goes to $\infty$, with 
	    \begin{equation*}
	        ( \tau_{\text{cr}} - \tau ) \propto N^{-\eta}
	    \end{equation*}
	    for some positive rational exponent $\eta$. However, it seems not possible to find a scaling such that the Riemann--Hilbert problem reduces to that of Painlev\'e II, the latter being the case for the double-scaling limit of the GWW model.\par
	    Therefore, we expect that the DQPT we uncover belongs to a \emph{different universality class} than GWW. From the analysis of the random matrix mechanism behind the transition, we may conjecture that the corresponding equation governing the double-scaling limit is Painlev\'e I. That is to say, the universality class would be the same as two-dimensional quantum gravity \cite{Brezin:1990rb,Gross:1989vs,Douglas:1989ve}, rather than GWW. It would be interesting to prove (or disprove) this conjecture.

\section{Robustness against impurities in the initial state}
\label{sec:S3}
It may happen that, in the preparation of the initial state $\lvert \psi_0 \rangle $, one or more of the spins are not aligned, thus producing a state that slightly differs by the single domain wall state $\lvert \downarrow  \downarrow \cdots  \downarrow \uparrow \uparrow \cdots \uparrow \rangle $. A desirable feature of a quantum quench, that makes the DQPT more easily accessible experimentally, is the robustness of the phase transition against the introduction of such impurities, or defects, in the initial state. In this section we show that it is indeed the case for the quench we are considering. Recent theoretical studies of spin chains (or lattice models) with defects include \cite{Aasen:2016dopC,Aasen:2020jwbC,Roy:2021jusC,Roy:2021xulC,Tan:2022vazC}.\par
\medskip
For concreteness, let us assume that, in the initialization of the quench, a spin $\lvert \uparrow \rangle_{N-p}$ in the position $N-p$ (i.e. the $p^{\text{th}}$ spin from the right) is not set to $\lvert \downarrow \rangle_{N-p}$. That is, we consider the Loschmidt amplitude 
\begin{equation}
\label{eq:defDefectEcho}
    \mathfrak{D}_{N,L} ^{(p)} (t) := \langle \underbrace{\downarrow  \dots \downarrow \downarrow }_{N-p-1} \uparrow  \underbrace{\downarrow  \downarrow \dots \downarrow }_{p}  \underbrace{\uparrow \uparrow \dots \uparrow }_{L-N} \left\lvert e^{-it H_{\text{XY}}} \right\rvert \psi_0 \rangle 
\end{equation}
where the notation $\mathfrak{D}^{(p)}$ stands for \emph{defect}. As usual, we assume for simplicity $L \gg N$ and set $L=\infty$ throughout the analysis. Besides, we consider the chain with PBC, being the analysis for ABC completely analogous.\par

\subsubsection*{Statement of the result}
A careful analysis of the amplitude in presence of an impurity, defined in \eqref{eq:defDefectEcho}, using random matrix theory techniques, leads to 
\begin{equation}
\label{eq:resultDefectleadingN}
    \lim_{N \to \infty} \frac{1}{N^2} \mathfrak{D}_{N+1,\infty} ^{(p)} (t) =   \lim_{N \to \infty} \frac{1}{N^2} \mG_{N, \infty} (t) ,
\end{equation}
for every $t$, meaning that the presence of an impurity does not spoil the DQPT.

\subsubsection*{Derivation of the leading order behaviour}

From the derivation in Subsection \ref{sec:derivquench}, in particular from \eqref{eq:GNintegralgen}, we have that 
\begin{equation*}
    \mathfrak{D}_{N,\infty} ^{(p)} (t) = \int_{U(N-1)} \dd U ~\chi_{\mathsf{A}_p} (U)  ~ \exp \left[-it \mathrm{Tr} \left( U + U^{-1} \right)  \right] ,
\end{equation*}
where $\mathsf{A}_p$ means the $p^{\text{th}}$ antisymmetric representation of $U(N-1)$.\par
We are interested in large $N$ behaviour of the ratio 
\begin{equation}
\label{eq:ratioDefect}
    \frac{\mathfrak{D}_{N+1,\infty} ^{(p)} (t)}{\mG_{N,\infty} (t)} = \left\langle \chi_{\mathsf{A}_p}  \right\rangle_{N,t}
\end{equation}
where the average is taken in the random matrix ensemble \eqref{eq:GU=GWWlong}. To reduce clutter, here and in the rest of the section we omit the dependence on $N,L$ and $t$ when there is no risk of confusion.\par
To begin with, we introduce the generating function of the defects:\footnote{The computations that follow are inspired by the analysis of certain order operators, the Wilson loop operators, in Quantum Field Theory in \cite{Hartnoll:2006C,Russo:2017ngfC,Santilli:2018byiC,Santilli:2021rcpC}.}
\begin{equation*}
    \mathfrak{D} (w) = \sum_{p=0} ^{N} w^p  \left\langle \chi_{\mathsf{A}_p}  \right\rangle .
\end{equation*}
Recall that, when written in terms of the eigenvalues $\left\{ e^{i \theta_a} \right\}_{a=1, \dots, N}$ of the random matrix $U$, the character in the $p^{\text{th}}$ antisymmetric representation is the $p^{\text{th}}$ elementary symmetric polynomial. The generating function then reads 
\begin{subequations}
\begin{align}
    \mathfrak{D} (w) & = \left\langle \det \left( 1+ w U \right)\right\rangle \\
    & = \left\langle \exp \left[\mathrm{Tr} \ln \left( 1+ w U \right) \right] \right\rangle . \label{eq:genfdefectlog}
\end{align}
\end{subequations}
It is explicit in the latter form that the determinant insertion in the matrix model \eqref{eq:GU=GWWlong} shifts the integrand by 
\begin{equation*}
    \exp \left[ - N^2 S_{\text{eff}} ^{(\text{\rm PBC})} (\theta_1, \dots, \theta_N )\right] \ \mapsto  \exp \left[ - N^2 S_{\text{eff}} ^{(\text{\rm PBC})} (\theta_1, \dots, \theta_N ) + N \left( \frac{1}{N} \sum_{a=1} ^{N} \log \left( 1 + w e^{i \theta_a} \right) \right) \right]  ,
\end{equation*}
with the second summand in the right-hand side being $O(N)$. Hence, its effect is negligible at large $N$ when compared to the term $N^2 S_{\text{eff}}$, and it does not alter the saddle point solution. In more detail, the procedure consists in going along the steps in Section \ref{sec:PT}, with the insertion of the term \eqref{eq:genfdefectlog} inside the integrand of the matrix model. We omit the details, as they are extremely similar to the derivation in Section \ref{sec:PT}, except that they carry this additional term. The newly inserted, $w$-dependent piece would only correct the Euler--Lagrange equations for $\rho$ by the $1/N$ contribution 
\begin{equation*}
    \frac{1}{N} \frac{w}{1+w z} 
\end{equation*}
which drops out in the planar limit. We conclude that the saddle point equation, and therefore the saddle point eigenvalue distribution, are not affected by these impurities. In the planar limit, the computation of $\mathfrak{D}^{(p)}$ reduces to
\begin{equation*}
     \mathfrak{D} ^{(p)} \approx e^{-N^2 S_{\text{eff}} ^{(\text{\rm PBC})} [\rho] }  ,
\end{equation*}
whence the advertised result \eqref{eq:resultDefectleadingN}.\par

\begin{cor}[Robustness against small variations of the initial state]
\begin{itemize}
    \item[(i)] Let $f^{(p)} (\tau)$ denote the dynamical free energy obtained replacing the state $\psi_0$ with the state $\lvert N, \kappa = \mathsf{A}_p \rangle$ in the definition of the Loschmidt echo. The partition $ \kappa = \mathsf{A}_pa$ consists of a single column of $p$ boxes, with $p<N$. In the planar limit considered so far, $f^{(p)} (\tau) \approx f (\tau)$.
     \item[(ii)] Similarly, let $\mathscr{P}= \left\{ p_1, \dots , p_m \ : \ 0 \le p_1 < \cdots < p_m <N\right\} $ be a collection of $m<N$ indices in the range $[0,N-1]$, and define the superposition state 
     \begin{equation*}
         \lvert N, \mathscr{P} \rangle:= \frac{1}{\sqrt{m}} \sum_{p \in \mathscr{P}  } \lvert N, \kappa = \mathsf{A}_p \rangle .
     \end{equation*}
     Besides, let $f^{(\mathscr{P})} (\tau)$ denote the dynamical free energy obtained replacing the state $\psi_0$ with $ \lvert N, \mathscr{P} \rangle$ in the definition of the Loschmidt echo. In the planar limit, $f^{(\mathscr{P})}  (\tau)\approx f (\tau)$.
\end{itemize}
\end{cor}
\subsubsection*{Derivation of the defect contribution}
We can go on and evaluate the first non-trivial contribution to the ratio \eqref{eq:ratioDefect}, which is the contribution of the defect in the planar limit. As we have said, the insertion of the impurity has a sub-leading effect that preserves the saddle point solution. We can thus evaluate $\mathfrak{D} (w) $ in the planar limit using the eigenvalue density computed in Subsection \ref{sec:DQPTPBCLinfty}. Rewriting \eqref{eq:genfdefectlog} in the planar limit, we obtain 
\begin{equation*}
    \frac{1}{N} \ln \mathfrak{D} (w) \approx \int_{\Gamma} \dd z \varrho (z) \log \left( 1+ w z \right) .
\end{equation*}
As with the dynamical free energy, it is easier to evaluate the derivative 
\begin{equation}\label{eq:dDdwgenFdefect}
    \frac{ \dd \ }{\dd w } \frac{1}{N} \ln \mathfrak{D} (w)  \approx \int_{\Gamma} \dd z \varrho (z) \frac{z}{1+wz} .
\end{equation}
In the first phase $\tau < \tau_{\text{cr}}$ the integral evaluates to $i 2 \tau$. Exponentiating and expanding, we finally arrive at 
\begin{equation*}
    \left. \frac{\mathfrak{D}_{N+1,\infty} ^{(p)} (t)}{\mG_{N,\infty} (t)} \right\rvert_{\tau < \tau_{\text{cr}} } \approx \frac{ (i t)^p }{p!}
\end{equation*}
valid in the large $N$ limit. In particular, the closer the impurity to the ``outer region'' with all spins up, corresponding to a lower value of $p$, the smaller is its net effect.\par
Once again, lacking an explicit expression for $\varrho (z)$ in the phase $\tau> \tau_{\text{cr}}$, we are unable to provide a closed-form expression for $\mathfrak{D}_{N+1,\infty} ^{(p)} (t)$ in this late time phase. Notwithstanding, formula \eqref{eq:dDdwgenFdefect} remains valid, from which the defect contribution is implicitly determined. In particular, the statements that the defect contribution is sub-leading in $N$ and its expectation value has a finite large-$N$ limit continue to hold.

\subsubsection*{Other impurities in the initial state}
The previous derivation is easily adapted to the case in which the impurity consists of a spin $\lvert \downarrow \rangle_{N+p}$, $p$ positions away from the domain wall. The difference in this case would be to replace $N+1$ with $N-1$, which of course has no effect in the large $N$ limit, and to replace the $p^{\text{th}}$ antisymmetric representation $\mathsf{A}_p$ with the $p^{\text{th}}$ symmetric representation $\mathsf{S}_p$. The defect generating function \eqref{eq:genfdefectlog} is replaced by 
\begin{equation*}
    \mathfrak{D} (w) = \left\langle \exp \left[ - \mathrm{Tr} \ln \left( 1- w U \right) \right] \right\rangle ,
\end{equation*}
while the rest of the argument goes through essentially unchanged.

\section{Dynamical quantum phase transition and finite system size}
\label{sec:S4}
		We now move back to the problem of analyzing the effect of a finite number of qubits on the DQPT.\par
		To this aim, we consider the finite-$L$, discrete matrix models \eqref{eq:G=MMdiscrlong} and take the scaling limit 
		\begin{equation}
		\label{eq:discrLNscaling}
			N, L \to \infty , \qquad \ell = \frac{L}{N} \ge 1 \text{ fixed},
		\end{equation}
		together with the planar limit with $\tau = \frac{t}{N}$ fixed. The additional scaling parameter $\ell$ has the meaning of (the inverse of) density of spins $\downarrow$ in the initial state $\lvert \psi_0 \rangle $.\par
		
	\subsubsection*{Statement of the results}
	The goal of the present section is to argue that the DQPT discussed above persists in the limit \eqref{eq:discrLNscaling} that accounts for the effects of finite ratio $\ell$. Two take-home messages emerge from our analysis:
	\begin{itemize}
	    \item A finite $\ell$ triggers a phase transition, which nevertheless does not spoil the DQPT previously uncovered if $\ell \gtrsim 1.2$;
	    \item In the phase $\tau < \tau_{\text{cr}}$ and with $\ell \ge 2$, the dependence on the system size is exponentially small in $(L-2N)$.
	\end{itemize}
	To distinguish the phase transitions induced by finite $N/L$ effects, as per the first bullet point, from the genuine DQPT we are interested in, we will refer to the former as \emph{finite density-induced phase transition} throughout the current section.
	
	\begin{rmk}[Naturalness of the scaling limit]
	    The scaling limit \eqref{eq:discrLNscaling} retains a finite density $\ell^{-1}$ of spins $\downarrow$. The limit of a large number of lattice sites with finite particle density is a commonly studied thermodynamic limit in condensed matter systems, and is moreover the most suited one to mimic the behaviour of an experimental setup, in which $L$ can be much larger than $N$ but their ratio is inevitably finite. Crucially, the outcomes of this section prove that there is no significant loss of generality in considering the simpler setting with $L\to \infty$ first and then $N \to \infty$. The critical behaviour of the system stays the same for every $\ell \ge 2$, up to \emph{exponentially small} corrections.
	\end{rmk}
	
	\subsubsection*{Discrete matrix models and phase transitions}
		It has been known since the work of Douglas and Kazakov \cite{Douglas:1993iiaC} that the discreteness of a matrix ensemble may induce a phase transition. The argument is simple: the discreteness imposes a lower bound on the distance between any two eigenvalues of the random matrix ensemble, namely the lattice spacing. This lower bound becomes a constraint on the eigenvalue density at large $N$. Whenever a given solution to the saddle point equation ceases to satisfy such a constraint in a region of parameter space, the model undergoes a phase transition.\par
		For the models we consider, which are discretized versions of $U(N)$ and $USp(2N)$ ensembles, we have the following result, see \cite{BaikLiuC}.
		\begin{lem}\label{lemma:discrGenBaik}
		    Let $\rho (z)$ be the eigenvalue density associated to a discrete matrix model with discrete domain $\mathscr{D}$. Then $\rho (z)$ is upper bounded according to 
		    \begin{equation*}
		        \left\lvert \rho (z) \right\rvert \le \left\lvert \rho_{\text{unif.}} ^{\mathscr{D}} \right\rvert 
		    \end{equation*}
		    where $\rho_{\text{unif.}} ^{\mathscr{D}}$ is the normalized, uniform probability density on the domain $\mathscr{D}$.
		\end{lem}
		The discrete domains in \eqref{eq:G=MMdiscrlong}, arising from the consideration of the physical spin chain, are discrete subsets of $U(1)$ subject to the following additional constraints.
		\begin{lem}\label{lemma:discrconstr}
		 The eigenvalue densities associated to the discrete matrix ensembles \eqref{eq:G=MMdiscrlong} satisfy the condition
		\begin{subequations}
		\begin{align}
			\text{\rm PBC:}& \qquad \lvert \varrho (z) \rvert \le \frac{ \ell }{2\pi} \label{eq:finiteLconstraintPBC} \\
			\text{\rm ABC:}& \qquad \lvert \varrho (z) \rvert \le \frac{ \ell }{\pi}  \label{eq:finiteLconstraintABC} 
		\end{align}
		\label{eq:finiteLformulas}
		\end{subequations}
		$\forall z \in \text{supp} \varrho$. 
		\end{lem}
		It is crucial that, by construction, a finite number of qubits induces a discretization of the matrix ensembles \eqref{eq:G=MMlong}. Then, in the ABC chain, the change of variables $x_a=\cos \theta_a$ produces a \emph{non-uniform} discretization of the interval $[-1,1]$, inherited by the uniform discretization of the semicircle. This is encapsulated in the expression \eqref{eq:GSPrealdiscr}.\par
		\begin{proof}[Proof of Lemma \ref{lemma:discrconstr}]
		Lemma \ref{lemma:discrconstr} is a corollary of the more general Lemma \ref{lemma:discrGenBaik}. Since we will only need the content of Lemma \ref{lemma:discrconstr}, we prove it explicitly, building on \cite{Jain:2013pyC,Santilli:2019wvq}.\par
		The derivation of \eqref{eq:finiteLformulas} is as follows. Consider the discrete ensembles in \eqref{eq:G=MMdiscrlong} and use the invariance of the summand to restrict the discrete angles $\left\{ \theta (s_j) \right\}_{j=1, \dots , N}$, as defined in \eqref{eq:discrthetaofs}, to the principal Weyl chamber
		\begin{equation*}
			0 \le \theta (s_1) < \theta (s_2) < \cdots < \theta (s_N) \le C , \qquad C= \begin{cases} 2 \pi & \text{\rm PBC} \\ \pi & \text{\rm ABC} \end{cases}
		\end{equation*}
		at the expense of a factor $\lvert \mathrm{W} (G) \rvert = N!$. Then, one observes that 
		\begin{equation*}
			\lvert \theta (s) - \theta (s^{\prime}) \rvert \ge \frac{C}{L} \lvert s - s^{\prime} \rvert ,
		\end{equation*}
		which, having the planar limit in mind, is conveniently rewritten in the form 
		\begin{equation*}
			\frac{ \left\lvert \frac{s}{N} - \frac{s^{\prime}}{N} \right\rvert }{\lvert \theta (s) - \theta (s^{\prime}) \rvert } \le \frac{L}{C N} .
		\end{equation*}
		Here $C$ is the length of the support onto which the eigenvalues live, which is the circumference $C=2\pi$ for PBC and the semicircle $C=\pi$ for ABC. Obviously, the constraint is ineffective if one first takes $L \gg N$. Instead, expressing everything in terms of the discrete complex variables $z(s)$ and sending $N \to \infty$ with the scaling \eqref{eq:discrLNscaling}, one gets 
		\begin{equation*}
			C \lvert \varrho (z) \rvert \le \ell ,
		\end{equation*}
		which yields \eqref{eq:finiteLformulas}. In imaginary time, the absolute value can be omitted, because $\varrho(z)=\rho (\theta)$ is positive and of real argument. However, to generalize to real time dynamics, we ought to be careful with the absolute values.
		\end{proof}\par
		\medskip
		The rest of the present section is organized as follows.
		\begin{itemize}
		    \item Before delving in the analysis of the fate of the DQPT when considering a realistic chain with $L < \infty$, we mention related results for the imaginary-time model, i.e. the discrete GWW model. These are collected at the beginning of Subsection \ref{sec:finiteLGWW}.
		    \item Then, in the body of Subsection \ref{sec:finiteLGWW} we show explicitly how to derive a finite density-induced phase transition in the imaginary-time PBC chain, and extend the result to the imaginary-time ABC chain.
		    \item Having gained insight in the (arguably simpler) imaginary-time model, we face the effects of a finite system size in the real-time models we are interested in Subsection \ref{sec:finiteLDQPT}.
		    \item The final Subsections \ref{sec:ExpNumericDecay}-\ref{sec:ExpAccurateDecay} are devoted to show that the infinite system Loschmidt echo $\mL_{N, \infty} (t)$ approximates the finite size Loschmidt echo $\mL_{N, L<\infty} (t)$ with exponential accuracy in the first phase.
		\end{itemize}
		To be clear: in Subsections \ref{sec:finiteLGWW}, \ref{sec:finiteLDQPT} \& \ref{sec:ExpAccurateDecay} we work in the scaling limit \eqref{eq:discrLNscaling}, in which $L$ and $N$ are large but their ratio $\ell$ is finite. In Subsection \ref{sec:ExpNumericDecay}, instead, we compare numerically the analytic expectations with the Loschmidt echo at finite values of $L$ and $N$.

	\subsection{Finite size effects in imaginary time}
	\label{sec:finiteLGWW}
	
	\subsubsection*{Exponential accuracy of the infinite system approximation in imaginary time}
	Consider the imaginary-time version of the discrete ensemble \eqref{eq:GU=discrlong}, that is, the discrete GWW model.
	Rigorous estimates for the rate of convergence of such discrete ensemble (and its generalizations) to its continuous analogue \eqref{eq:GU=GWWlong}, from the first phase, have been given in \cite{BaikLiuC}. See also \cite{Perez-Garcia:2013lba} for a discussion in the context of spin chains. It turns out that the error in considering the continuous matrix model instead of the finite-$L$, discrete model is exponentially small in $(L-N)$ \cite{BaikLiuC}.\par
	Moreover, the scaling limit of the discrete GWW model has been extensively analyzed in \cite{Jain:2013pyC} (and later in \cite{Santilli:2019wvq} a generalization of the model), where it is shown that the discreteness of the ensemble induces a Douglas--Kazakov-type \cite{Douglas:1993iiaC} phase transition as a function of the parameter $\ell$.\par
	
	\subsubsection*{Finite density-induced phase transition in imaginary time}
		We now reconsider the imaginary-time PBC chain taking into account the constraint \eqref{eq:finiteLconstraintPBC}. Looking back at the eigenvalue density \eqref{eq:GWWrhoPI}, the condition \eqref{eq:finiteLconstraintPBC} reads
		\begin{equation*}
			1 + 2 \gamma \le \ell.
		\end{equation*}
		It it obvious that, for $\ell \gg 1$, the finite density $\ell^{-1}$ has no effect on the GWW phase transition. More specifically, having a finite value for $\ell$ produces a Douglas--Kazakov phase transition at 
		\begin{equation*}
			\gamma = \gamma_{\star} \equiv \frac{\ell -1}{2} ,
		\end{equation*}
		which does not affect the large $N$ GWW phase transition at $\gamma=\frac{1}{2}$ if $\ell \ge  2 + \varepsilon$ for arbitrary but fixed $\varepsilon >0$. We refer to \cite{Jain:2013pyC} for an exhaustive analysis.\par
		\medskip
		It is an easy task to repeat the argument to the imaginary-time chain with ABC. From the eigenvalue density \eqref{eq:rhoPISpNGWW} and imposing \eqref{eq:finiteLconstraintABC}, we get the same critical value $\gamma_{\star} = \frac{\ell-1}{2}$ for the phase transition induced by the discreteness of the matrix ensemble.\par
		\medskip
		The special value $\ell =2$ has a neat physical meaning. Recall that $\ell =L/N$, thus a chain with PBC and $\ell=2$ in the state $\lvert \psi_0 \rangle $ has exactly half of the qubits $\lvert \downarrow \rangle$ and half $\lvert \uparrow \rangle$. For $\ell <2$, it would be convenient to rephrase the problem as starting with the state 
		\begin{equation*}
			\bigotimes_{k=1} ^{L} \lvert \downarrow \rangle_k = \lvert \downarrow \downarrow \dots \downarrow \rangle 
		\end{equation*}
		and act with the spin-flip operators $\sigma^{+} _{1, \dots, L-N}$ on the last $(L-N)$ qubits.\par
		The value $\ell=2$ is the self-dual point of the $\mathbb{Z}_2$ spin flip duality. At such point, the duality of the model becomes a symmetry.

	\subsection{Finite size effects on the DQPT}
	\label{sec:finiteLDQPT}
		
		We start by recalling that the map $x=\cos \theta$ produces a non-uniform discretization of the interval $[-1,1]$, and hence of $\Gamma$. The constraint \eqref{eq:finiteLconstraintABC} then reads 
		\begin{equation*}
			\lvert \pi \sqrt{1-x^2} \varrho^{\text{(ABC)}} (x) \rvert \le \ell .
		\end{equation*}
		In particular, define 
		\begin{align*}
			m^{\text{(PBC)}} (\tau) & := \max_{z \in \Gamma} \lvert 2 \pi \varrho^{\text{(ABC)}} (z) \rvert , \\ 
			m^{\text{(ABC)}} (\tau) & := \max_{z \in \Gamma} \lvert \pi \sqrt{1-z^2} \varrho^{\text{(ABC)}} (z) \rvert  ,
		\end{align*}
		so that a finite density-induced DQPT takes place at the critical value $\tau_{\star}$ which solves 
		\begin{equation*}
			 m^{\text{(PBC)}} (\tau_{\star}) = \ell \qquad \text{ or } \qquad m^{\text{(ABC)}} (\tau_{\star}) = \ell .
		\end{equation*}
		Both eigenvalue densities have local maxima at $z= \pm 1$, and attain the same extremal value $\sqrt{1+4\tau^2}$. The other local maxima over $\C$ do not satisfy the initial condition, namely do not belong to the undeformed contour in the limit $\tau \to 0$. We then have 
		\begin{equation*}
			\tau_{\star} = \frac{\sqrt{\ell^2 -1}}{2} 
		\end{equation*}
		for both PBC and ABC. Requiring that this finite density-induced phase transition at $\tau_{\star}$ takes place at a late time $\tau_{\star} > \tau_{\text{cr}}$, not to invalidate the derivation of the DQPT discussed above, we are led to the condition $\ell > \ell_{\star} $ where 
		\begin{equation*}
		\frac{\sqrt{\ell_{\star}^2 -1}}{2}  = \tau_{\text{cr}} \quad \Longrightarrow \quad \ell_{\star} \approx 1.2 
		\end{equation*}

	\subsection{Finite vs infinite system: a numerical study}
	\label{sec:ExpNumericDecay}
		To give a quantitative estimate of the error in approximating a realistic chain, consisting of a number $L$ of qubits, with $L=\infty$, we define
		\begin{subequations}
		\begin{align}
			\err_N (\ell, \gamma) & = \ln \frac{  \left\lvert \langle \psi_0 \lvert e^{-\beta H_{\text{XY}}} \rangle \psi_0 \rangle_{L<\infty} \right\rvert^2 }{ \left\lvert \langle \psi_0 \lvert e^{-\beta H_{\text{XY}}} \rangle \psi_0 \rangle_{L=\infty} \right\rvert^2  }   =  2 N^2 \left[ \mathcal{F} (\gamma) \vert_{L< \infty} -  \mathcal{F} (\gamma) \vert_{L= \infty} \right] \label{eq:deferrorthermal} \\ 
			\err_N (\ell, \tau) & = \ln \frac{  \left\lvert \langle \psi_0 \lvert e^{- it H_{\text{XY}}} \rangle \psi_0 \rangle_{L<\infty} \right\rvert^2 }{ \left\lvert \langle \psi_0 \lvert e^{- it H_{\text{XY}}} \rangle \psi_0 \rangle_{L=\infty} \right\rvert^2  }   =  - 2 N^2  \left[ f(\tau)\vert_{L< \infty}  - f(\tau)\vert_{L= \infty} \right] .  \label{eq:deferrorquantum}
		\end{align}
		\end{subequations}
		From \cite{BaikLiuC}, in the imaginary-time scenario we have 
		\begin{equation*}
			\err_N (\ell, \gamma) \approx \ln \left( 1 + e^{-c(\gamma) (L-N)} \right) \approx e^{-c(\gamma) N(\ell -1)} \qquad \forall \ 0 \le \gamma \le \frac{1}{2}, \ \ell > 1+ \delta
		\end{equation*}
		for a function $c (\gamma) >0$ which does not depend on $\ell$ and a suitable $\delta >0$. We have derived above $\delta=1$ at leading order in the planar limit.\par
		Thus, before the GWW phase transition, the error in approximating the imaginary-time system by an infinite number of qubits (i.e. $L\to \infty$ from the onset) is exponentially small in the system size.\par
		\medskip
		One of the outcomes of Section \ref{sec:PT} is that, in the first phase, the real-time and imaginary-time dynamics are related by analytic continuation, a fact that does not hold beyond the respective phase transitions. From that, we are tempted to expect that the finite-$L$ dependence in the first phase is exponential also in the real-time system.\par
		The plots in Figures~\ref{fig:errvstau} \& \ref{fig:errvsell} make manifest that the error \eqref{eq:deferrorquantum} behaves identically to its imaginary-time analogue \eqref{eq:deferrorthermal}. In particular, from Figure~\ref{fig:errvsell} we observe an exponential decay of $\err_N (\ell, \tau)$ as a function of $\ell$.
		\begin{figure}[ht]
			\centering
				\includegraphics[width=0.5\textwidth]{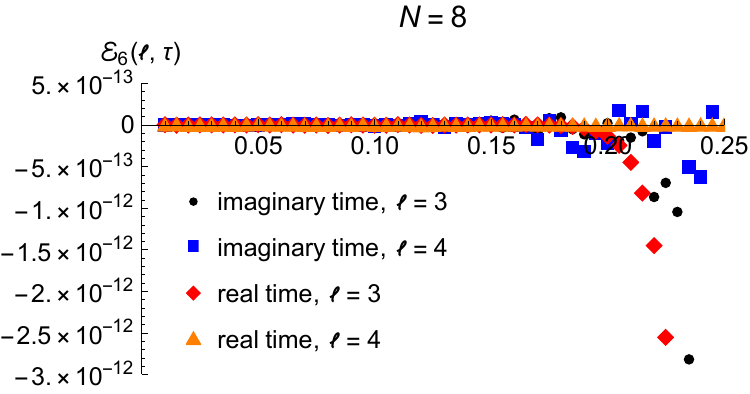}
			\caption{Plot of $\err_N$ as a function of $\tau$ (real time) or $\gamma$ (imaginary time).}
			\label{fig:errvstau}
			\end{figure}\\
			\begin{figure}[ht]
			\centering
				\includegraphics[width=0.5\textwidth]{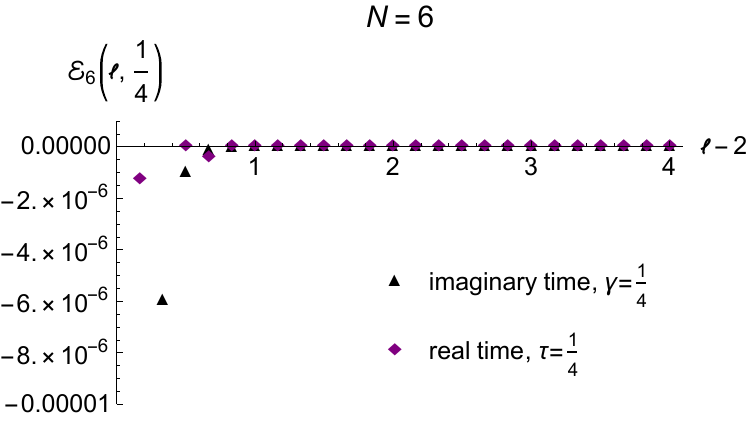}
			\caption{Plot of $\err_N$ as a function of $\ell$.}
			\label{fig:errvsell}
			\end{figure}\par

		\subsection{Exponentially accurate thermodynamic limit}
		\label{sec:ExpAccurateDecay}
			As already mentioned above and elucidated in \cite{Perez-Garcia:2013lba}, borrowing rigorous estimates from the work \cite{BaikLiuC} one proves that, in imaginary time, the ratio 
			\begin{equation*}
				\frac{\lvert \langle \psi_0 \vert e^{- \beta H_{\text{XY}}} \vert \psi_0 \rangle_{L< \infty} \rvert^2 }{ \lvert \langle \psi_0 \vert e^{- \beta H_{\text{XY}}} \vert \psi_0 \rangle_{L= \infty} \rvert^2  } = 1 + O(e^{-c(\gamma )N (\ell-1)}) , \qquad L,N \to \infty \text{ with } \ell-1 >  \delta 
			\end{equation*}
			for a $c(\gamma)>0, 0 \le \gamma< \frac{1}{2}$ and a suitable $\delta >0$. Therefore, for a sufficiently large system, the Loschmidt echo can be approximated by its infinite-$\ell$ limit with error exponentially small in $(L-N)$.\par
			Although the procedure of \cite{BaikLiuC} does not straightforwardly extend to the real time dynamics, Figure~\ref{fig:errvsell} suggests that similar estimates hold. In the rest of this subsection, we argue that this is indeed the case.\par
			\medskip
			We work with the discrete model \eqref{eq:GSPrealdiscr}, i.e. we henceforth focus on the chain with ABC. Following \cite{BaikLiuC}, we first define an auxiliary function 
			\begin{equation*}
				\nu (z) = \left( z+i \sqrt{1-z^2} \right)^L -1
			\end{equation*}
			such that 
			\begin{equation*}
			    \nu (z) =0 \qquad  \forall z \in \scS_L 
			\end{equation*}
			where $\scS_L$ s the discrete domain in \eqref{eq:GSPrealdiscr}. We then apply Cauchy's theorem to write 
			\begin{equation}
			\label{eq:entryHankeldiscr}
				\frac{1}{L} \sum_{x \in \scS_L} x^{a+b-2} e^{-4it x} \sqrt{1-x^2} = \frac{1}{L} \oint_{\Sigma} \frac{\dd x }{2 \pi i } ~ \frac{\nu^{\prime} (x)}{\nu (x)} ~ x^{a+b-2} ~ e^{-4it x} \sqrt{1-x^2} .
			\end{equation}
			The integration contour, shown in Figure~\ref{fig:contourSigma}, is 
			\begin{equation*}
				\Sigma = \Sigma_{-} \cup \Sigma_{1} \cup \left( - \Sigma_{+} \right) \cup \left( - \Sigma_{-1} \right) ,
			\end{equation*}
			where 
			\begin{align*}
				\Sigma_{\pm} & = \left[ -1 + \varepsilon , 1-\varepsilon \right] \pm i \varepsilon^{\prime} \\
				\Sigma_{\pm 1} &= \pm (1-\varepsilon ) + i \left[ - \varepsilon^{\prime}, \varepsilon^{\prime} \right] 
			\end{align*}
			for arbitrarily small $\varepsilon, \varepsilon^{\prime} >0$. Notice that we are cutting an arbitrarily small neighbourhood of the edges $x= \pm 1$ of size $\varepsilon$. This choice excises the accumulation points $x=\pm 1$ of the zeros of the orthogonal polynomials we will introduce momentarily. Because of the square root behaviour of the measure near the edges, the integrand vanishes linearly in $\varepsilon$.\par
			The contour integral on the right-hand side of \eqref{eq:entryHankeldiscr} is solved by residues and, thanks to the suitable choice of function $\nu (z)$, yields the left-hand side.\par
			Along the way, observe that, writing $z= e^{i \mathrm{ArcCos} (x)}$, with $\mathrm{ArcCos} (x)$ the arccosine function, 
			\begin{equation*}
				\tilde{\nu} (x) := \sqrt{1-x^2}\frac{\nu^{\prime} (x)}{i L \nu (x)} = - \frac{\nu (x)-1}{\nu (x)} = \frac{z^L}{1-z^L} .
			\end{equation*}
			Here $\tilde{\nu} (x) $ is a shorthand notation. For $x \in \Sigma_{-}$ (resp. $\Sigma_{+}$), $z= e^{i \mathrm{ArcCos} (x)}$ is mapped to a contour slightly inside (resp. outside) the semicircle $\left\{ e^{i \theta} \ : \ 0 \le \theta \le \pi \right\}$, see the plot in Figure~\ref{fig:backmapSigma}.\par
			\begin{figure}[ht]
			\centering
				\includegraphics[width=0.75\textwidth]{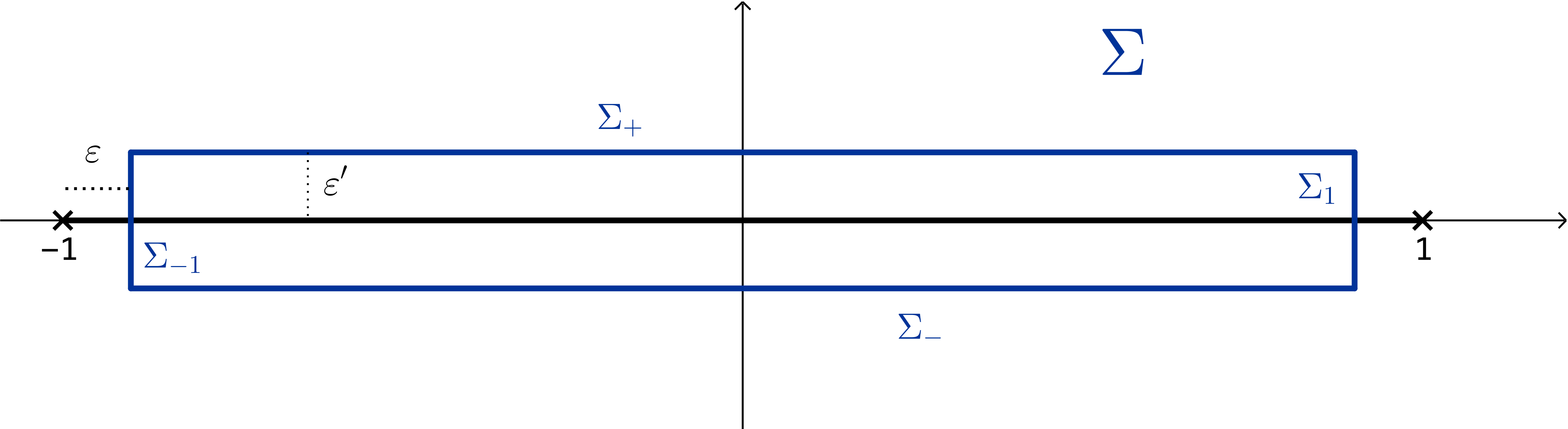}
			\caption{Choice of contour $\Sigma$ for the application of Cauchy's theorem.}
			\label{fig:contourSigma}
			\end{figure}\par
			\begin{figure}[ht]
			\centering
				\includegraphics[width=0.4\textwidth]{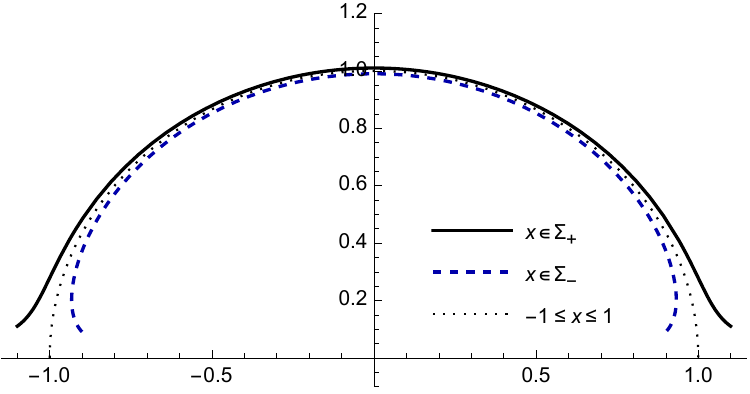}
			\caption{Image of $\Sigma_{\pm}$ under $\mathrm{ArcCos}$.}
			\label{fig:backmapSigma}
			\end{figure}
			Taking the Hankel determinant of \eqref{eq:entryHankeldiscr} and using standard row and column manipulations we get 
			\begin{equation}
			\label{eq:HankeldetCauchy}
				\mG_{N,L < \infty} ^{\text{(ABC)}} (t)  = \det_{1 \le a,b, \le N} \left[  \oint_{\Sigma} \frac{\dd x }{2 \pi } ~ P_{a-1} (x) P_{b-1} (x) ~ e^{-4it x} \tilde{\nu} (x) \right] ,
			\end{equation}
			where $\left\{ P_j \right\}_{j \in \mathbb{N}}$ is a system of monic polynomials. In practice we use the orthogonal polynomials with respect to the measure $e^{i N \lambda_{\text{D}} x} \dd x$ studied in \cite{DeanoC}, with identification $\lambda_{\text{D}} =-4 \tau$.\par
			Note that, in \eqref{eq:HankeldetCauchy}, we have used the invariance of the determinant under addition of rows and columns to replace the factor $x^{a+b-2}$, that arises in taking the Hankel determinant, with $P_{a-1} (x) P_{b-1} (x) =x^{a-1} x^{b-1} +$ lower terms.\par
			\medskip
			From the analytic properties of the integrand in each entry of \eqref{eq:HankeldetCauchy}, it follows that 
			\begin{equation*}
				\int_{\Sigma_{+}} \frac{\dd x }{2 \pi } ~ P_{a-1} (x) P_{b-1} (x) ~ e^{-4it x} = h_{a-1} \delta_{ab} ,
			\end{equation*}
			where $h_a = h_a (t)$, valid for $\tau$ distinct from every point at which any of the $\left\{ h_{a} \right\}_{a=0, \dots , N-1}$ vanishes. Importantly, a sufficient condition at leading order in $N$ is $\tau<\tau_{\text{cr}}$, in perfect agreement with the rest of our discussion so far.\par
			With the planar limit in mind, we restrict to $\tau<\tau_{\text{cr}}$ and write \cite{BaikLiuC}:
			\begin{equation}
			\label{eq:Hankeldiscdet2}
				\mG_{N,L < \infty} ^{\text{(ABC)}} (t)  = \det_{1 \le a,b, \le N} \left[ h_{a-1} \delta_{ab} + \int_{\Sigma_{+} \cup \Sigma_{-}} \frac{\dd x }{2 \pi }  P_{a-1} (x) P_{b-1} (x) ~ e^{-4it x}  v (x) \right] ,
			\end{equation}
			where we have neglected the $O(\varepsilon^{\prime})$ pieces, have changed the orientation of the integral along $- \Sigma_{+}$ and eventually introduced the shorthand notation 
			\begin{equation*}
				v (x) = \begin{cases} -1 - \tilde{\nu} (x) & x \in \Sigma_{+} \\ \tilde{\nu} (x) & x \in \Sigma_{-} . \end{cases} 
			\end{equation*}
			The auxiliary function $v(x)$ accounts for both the change of orientation and the separation of the diagonal term $h_{a-1} \delta_{ab}$. In terms of the variable $z=e^{i \mathrm{ArcCos} (x)}$, it reads 
			\begin{equation*}
				v (x) = \begin{cases} \frac{1}{z^L -1} & x \in \Sigma_{+} \\  \frac{z^L}{1-z^L}  & x \in \Sigma_{-}  \end{cases} 
			\end{equation*}
			thus is exponentially suppressed in $L$ in both cases.\par
			Next, we bring out $\left\{ h_{a-1} \right\}_{a=1, \dots , N}$ from the corresponding columns in \eqref{eq:Hankeldiscdet2}, noticing that 
			\begin{equation*}
				\mG_{N,L = \infty} ^{\text{(ABC)}} (t)  = \prod_{a=0} ^{N-1} h_a 
			\end{equation*}
			which stems from the basics of random matrix theory. We emphasize again that we are using polynomials that are orthogonal with respect to the \emph{continuous measure} on $[-1,1]$, not with respect to the discrete one.\par
			We finally arrive at 
			\begin{equation*}
				\frac{ \mG_{N, L < \infty} ^{\text{(ABC)}} (t)  }{ \mG_{N, L = \infty} ^{\text{(ABC)}} (t)  } = \det \left[  \id_N + \mathcal{K} \right]
			\end{equation*}
			where $\id_N $ is the identity $N \times N $ matrix and, denoting $\left\{  p_{j} \equiv \frac{1}{\sqrt{h_j}} P_j \right\}$ the orthonormal polynomials, 
			\begin{equation*}
				\mathcal{K}_{ab} = \int_{\Sigma_{+} \cup \Sigma_{-}} \frac{\dd x }{2 \pi } ~ p_{a-1} (x) p_{b-1} (x) ~ e^{-4it x}  v (x) .
			\end{equation*}\par
			\medskip
			So far we have mostly adapted and extended \cite{BaikLiuC} to the non-uniformly discretized Hankel determinant \eqref{eq:GSPrealdiscr}. It remains to show that the matrix $\mathcal{K}$ is exponentially damped in $N,L \gg 1$.\par
			\medskip
			We have already observed that $v(x)$ is exponentially small as a function of $L$: 
			\begin{equation*}
				v(x) \approx \begin{cases} z^{L}  & \lvert z \rvert < 1 \\ z^{-L}  & \lvert z \rvert > 1 \end{cases} \qquad L \to \infty , \quad z= e^{i \mathrm{ArcCos} (x)} .
			\end{equation*}
			For the polynomials $p_j$, we use the estimate of \cite[Thm.2.2]{DeanoC}, which holds on $\Sigma_{\pm}$ for $\varepsilon >0$ arbitrarily small but finite, 
			\begin{equation*}
				p_j (x) \approx \frac{ q  (x)^{j+ \frac{1}{2}} }{(x^2-1)^{\frac{1}{4}}}  e^{\frac{ i t}{q (x)}} , \qquad j \to \infty ,
			\end{equation*}
			where $q (x) \equiv \frac{1}{2} \left[ x + (1-x^2)^2 \right]$. In the regime $0 < \varepsilon^{\prime} \ll \varepsilon \ll \frac{1}{2}$ we can safely assume $\lvert x \rvert <1$, which yields $\lvert q (x) \rvert <1$ and hence $\lvert \mathcal{K}_{ab} \rvert < c_1 e^{-c_2 L} $ for suitable positive, $L$-independent quantities $c_1,c_2$.\par
			However, in order to recover the original system \eqref{eq:GSPrealdiscr} we must modify the contour $\Sigma$ to involve the edges $x=\pm 1$. In a $\varepsilon$-neighbourhood of $x=\pm 1$, that is, $\lvert x \rvert = 1 + \varepsilon$, we have that $v(x)$ is suppressed approximately as $1/\lvert x \rvert^{L}$, whereas the product of the orthonormal polynomials grows proportionally to $\lvert x \rvert^{a+b-2}$. Because $a,b \le N$, a sufficient condition for $\lvert \mathcal{K}_{ab} \rvert$ to decay exponentially in $L$ is 
			\begin{equation*}
				L > 2  N -2 ,
			\end{equation*}
			which in the scaling limit \eqref{eq:discrLNscaling} becomes $\ell \ge 2$.\par
			
			\subsubsection*{Summary: finite versus infinite system}
			
			Let us summarize our finding for $L< \infty$, its thermodynamic limit with scaling \eqref{eq:discrLNscaling}, and its relation with the system in which we set $L= \infty$ from the onset.
			\begin{itemize}
				\item The Loschmidt echo in the infinite system approximation undergoes a third order DQPT. This is shown in Subsections \ref{sec:DQPTABCLinfty}-\ref{sec:DQPTPBCLinfty}.
				\item A system of $L$ qubits and finite ratio $\ell$ experiences the same phase transition, in the scaling limit \eqref{eq:discrLNscaling}, if the density satisfies $\ell > \ell_{\star} \approx 1.2$. The bound $\ell_{\star}$ is derived in Subsection \ref{sec:finiteLDQPT}.
				\item Furthermore, in the phase $\tau <\tau_{\text{cr}}$, the approximation of the finite density Loschmidt echo $\mL_N (t)\vert_{L< \infty}$ by $\mL_N (t) \vert_{L= \infty}$ in the thermodynamic limit \eqref{eq:discrLNscaling} is exponentially accurate if $\ell \ge 2$. This bound is proven in Subsection \ref{sec:ExpAccurateDecay}.
			\end{itemize}
			Let us stress the distinction between $\ell > \ell_{\star}$, which guarantees that the DQPT is not overturned by finite-$\ell$ effects (i.e. finite density of spins $\downarrow$ in the initial state), and the requirement $\ell \ge 2$, which guarantees that the results in the $L = \infty$ approximation reliably describe a finite chain.

\section{Odd \texorpdfstring{$N$}{N} and the quantum speed limit}
\label{sec:Nodd}
It was highlighted at the beginning of the note that the Loschmidt echo $\mL_{N,L} (t)$ with $N \in 2 \mathbb{N} +1$ shows damped oscillations at late times, vanishing at infinitely many points. By definition, the lowest time at which the Loschmidt echo vanishes is the time at which the evolved state $e^{i t H_{\text{XY}}} \lvert \psi_0 \rangle $ becomes orthogonal to the initial state $ \lvert \psi_0 \rangle$, and goes under the name of \emph{Quantum Speed Limit (QSL)}. We denote the corresponding value of $\tau $ as $\tau_N ^{\text{QSL}}$. Explicitly:
\begin{equation}
\label{eq:defqsltN}
    \tau_N ^{\text{QSL}} := \min \left\{ \tau > 0 \ : \ \mL_{N,L} (N \tau) = 0 \right\} .
\end{equation}
The interconnection between the quantum speed limit and dynamical phase transitions was elucidated in \cite{Zhou:2021bfo}, see also \cite{Lupo:2016C,Fogarty:2020C} for other related works on QSL.\par
The presence of zeros in $\mL_{2n+1,L} (t)$, $\forall n \in \mathbb{N}$, implies that the dynamical free energy 
\begin{equation*}
    f (\tau) = - \frac{1}{2 (2n+1)^2} \ln \mL_{2n+1,L} (t) ,
\end{equation*}
and even more so its limiting value for $n \to \infty$, has radically different analytic properties compared to the dynamical free energy along $N=2n \in 2 \mathbb{N}$. Concretely, $\tau_N ^{\text{QSL}} $ signals the first point at which $f(\tau)$ is singular, for every finite odd $N$.\par
Evaluating $f (\tau)$ exactly as a function of $\tau$ for several odd values of $N$, one is lead to the observation that $\tau_N ^{\text{QSL}} $ is progressively moved toward $\tau_{\text{cr}}$, thus establishing a connection with the dynamical free energy computed at even values of $N$. Inspired by the numerical evidence, we are able to prove the following result.
\begin{thm}\label{thm:QSL}
    For every $N \in 2\mathbb{N}+1$, let $\tau_N ^{\text{QSL}} $ be defined as in \eqref{eq:defqsltN}. For either choice of boundary conditions, $\exists$ $\tau^{\text{QSL}} \ge 0 $ such that 
    \begin{equation}
    \label{eq:existencelimit}
        \lim_{n \to \infty} \tau^{\text{QSL}} _{2n+1} = \tau^{\text{QSL}} .
    \end{equation}
    Moreover, 
    \begin{equation}
    \label{eq:tauqsl=taucr}
        \tau^{\text{QSL}} = \tau_{\text{cr}} ,
    \end{equation}
    with $\tau_{\text{cr}}$ defined in \eqref{eq:defeqtaucritical}.
\end{thm}
The rest of the current section is devoted to the proof of Theorem \ref{thm:QSL}. It is divided in three steps:
\begin{itemize}
    \item[1.] Prove the existence of the limit \eqref{eq:existencelimit};
    \item[2.] Prove that $ \tau^{\text{QSL}} \ge \tau_{\text{cr}} $;
    \item[3.] Given the lower bound, prove that $ \tau^{\text{QSL}} = \tau_{\text{cr}} $.
\end{itemize}
The proof is done assuming ABC, and the result for PBC is a consequence of it and the identity \eqref{eq:ABCequals2PBC}. Furthermore, we assume $L= \infty$ (cf. Subsection \ref{sec:ExpAccurateDecay}).

\subsubsection*{Step 1. Existence of the limit}
The first step is the existence of the limit \eqref{eq:existencelimit} and of the corresponding limiting value $\tau^{\text{QSL}}$.\par
Consider $\mG_{N,\infty} ^{(\text{\rm ABC})} $ as given in \eqref{eq:GSPreallong}. By Andreief's identity, it equals a Hankel determinant, as already discussed in Subsection \ref{sec:ExpAccurateDecay}. For every $N$, the Hankel determinant is associated to a system of orthogonal polynomials, which is closely related to those studied in \cite{DeanoC,CelsusDetC,Celsus:2020C}. Indeed, the difference between \eqref{eq:GSPreallong} and the Hankel determinants studied in \cite{CelsusDetC} is sub-leading in $N$ and we neglect it, having the planar limit in mind.\par
More precisely, denoting $\left\{ P_j \right\}$ the system of monic orthogonal polynomials satisfying 
\begin{equation*}
    \int_{-1} ^{1} \dd x ~ P_a (x) P_b (x) e^{-i4t x} = h_{a} \delta_{ab} ,
\end{equation*}
one has 
\begin{equation*}
    \lim_{N \to \infty} \frac{1}{N^2} \ln \mG_{N,\infty} ^{(\text{\rm ABC})} = \lim_{N \to \infty} \frac{1}{N} \sum_{a=1} ^{N} \ln \left( \frac{ h_{a-1}}{N} \right) ,
\end{equation*}
where the dependence of $\left\{ h_a \right\}_{a \ge 0}$ on $\tau$ is left implicit. We have arranged the inverse powers of $N$ in the right-hand side to highlight the ratios that have a finite large $N$ limit.\par
Zeros of the Loschmidt amplitude, and hence of the echo, are associated to degeneration of the system of orthogonal polynomials. Indeed, if for some  $a=1, \dots, N$ it happens that 
\begin{equation*}
    \left. h_{a-1} \right\rvert_{\tau_{0} ^{a}} =0 \qquad \text{ for a } \tau_{0} ^{a} >0 ,
\end{equation*}
the corresponding polynomial $P_{a-1} (x)$ ceases to be orthogonal to $P_{a-2} (x)$ \cite{CelsusDetC}. As a consequence, at $\tau_N ^{\text{QSL}}$ the zeros of one of the polynomials $\left\{ P_{a-1} \right\}_{a=1, \dots, N}$ coincide with the zeros of the polynomial preceding it in the hierarchy. By definition of QSL, it is the first value of $\tau$ for which this degeneracy takes place.\par
\medskip
It is proved in \cite{CelsusDetC} that, for large $t$ and fixed $N=2n+1$, the zeros of $h_{2n+1}$ are equispaced as functions of $t$. Regardless of the zeros before the late time regime becomes valid, expressing the zeros as functions of $\tau=t/(2n+1)$ moves all such zeros to the left as $n$ is increased. We infer that, for every fixed $n \in \mathbb{N}$, $\tau^{\text{QSL}} _{2n+1} < \infty$ and the sequence is non-increasing asymptotically in $n$. Stated more formally, we have derived the existence of a positive upper bound on $\tau^{\text{QSL}} _{2n+1} $ which is non-increasing in $n$:
\begin{equation*}
    \exists \tau^{\text{u.b}} _{2n+1} \ \forall n \in \mathbb{N} \quad \text{such that } \quad  \tau^{\text{QSL}} _{2n+1} \le \tau^{\text{u.b}} _{2n+1} \ \text{ and } \  \tau^{\text{u.b}} _{2n+1} \le \tau^{\text{u.b}} _{2n+1+p} \ \forall p \in \mathbb{N} .
\end{equation*}
\par
Note that the argument does not exclude the possibility that $\lim_{n \to \infty} \tau^{\text{QSL}} _{2n+1}  $ vanishes. We rule out this possibility in the next step.

\subsubsection*{Step 2. Lower bound on the QSL}
The second and third step rely on the approximate identity 
\begin{equation}
\label{eq:Todaid}
	\mG_{N+1}^2 \frac{ \partial^2 \ }{\partial t^2} \ln \mG_{N+1} \approx - 4 \mG_N \mG_{N+2} ,
\end{equation}
which holds for ABC at leading order in $N$. (Here we omit the dependence on $L$ and on the boundary conditions to reduce clutter). In fact, it is possible to write down an exact identity, with no approximation in $N$, which however we will not use. Identity \eqref{eq:Todaid} is a form of the Toda equation, appeared earlier in e.g. \cite{Basor:2018C,CelsusDetC}, to which we refer for the details of the proof.\footnote{See also \cite{Adler99,Adler03C} for more on the relation between the matrix ensembles \eqref{eq:G=MMlong} and the Toda integrable hierarchy.}\par
\medskip
For the second step, we recall the extension of Szeg\H{o}'s theorem \cite{Johansson:1997C} quoted in Theorem \ref{thm:Szego}. Because it does not assume that $N$ is even, nor the reality of the matrix model, it applies to the present situation. We thus have 
\begin{equation*}
    \lim_{N \to \infty} \ln \mG_{N} (t) = - 2 t^2 
\end{equation*}
in the strict large $N$ limit, without scaling $t$. By analyticity of the Loschmidt amplitude, we have that the estimate extends to the limit with scaling, as long as $\tau $ remains away from $ \tau^{\text{QSL}}$. That is, 
\begin{equation*}
    \lim_{N \to \infty} \frac{1}{N^2}\ln \mG_{N} ^{\text{(ABC)}} (N \tau ) = - 2 \tau ^2  \qquad \tau < \tau^{\text{QSL}} - \varepsilon 
\end{equation*}
for an arbitrary $\varepsilon >0 $.\par
We now look back at \eqref{eq:Todaid}, which we consider for $N=2n+1\in 2 \mathbb{N}+1$. Because $N+1=2(n+1)$ is even, the results of the previous sections apply to $\mG_{N+1} ^{\text{(ABC)}}$. Therefore, after dividing both sides of \eqref{eq:Todaid} by $(-4)$ and taking the logarithm, the left-hand side reads at large $N$ 
\begin{equation*}
   \lim_{n \to \infty} \frac{1}{(2n+2)^2} \ln \left[ - \frac{1}{4} \mG_{2n+2}^2 \frac{ \partial^2 \ }{\partial t^2} \ln \mG_{2n+2} \right] = -4 \tau^2  \qquad \tau < \tau_{\text{cr}} ,
\end{equation*}
whereas the right-hand side reads 
\begin{equation*}
     \lim_{n \to \infty} \frac{1}{(2n+2)^2}\ln \left[ \mG_{2n+1} \mG_{2n+3} \right] = - 4 \tau^2  \qquad \tau < \tau^{\text{QSL}} - \varepsilon .
\end{equation*}
Besides, we know that the left-hand side is analytic for $\tau < \tau_{\text{cr}}$. Imposing the equality of the two limiting values we get the constraint 
\begin{equation*}
    \tau^{\text{QSL}} - \varepsilon \ge \tau_{\text{cr}} \qquad \forall \varepsilon >0  .
\end{equation*}

\subsubsection*{Step 3. Identify the QSL}
We have accomplished step 2 of the proof exploiting the relation \eqref{eq:Todaid}, which bridges between even and odd $N$. Let us recall that it is an approximate identity, which holds up to terms that are sub-leading in $N$, which is sufficient for our purposes.\par
Let us rewrite \eqref{eq:Todaid} by sifting $N+1 \mapsto N$:
\begin{equation}
\label{eq:rewriteToda}
	\mG_{N}^2 \frac{ \partial^2 \ }{\partial t^2} \ln \mG_{N} \approx - 4 \mG_{N-1} \mG_{N+1} ,
\end{equation}
again understood in the limiting sense 
\begin{equation*}
    \lim_{n \to \infty } \frac{1}{(2n+1)^2} \ln \left[ - \frac{1}{4} \mG_{2n+1}^2 \frac{ \partial^2 \ }{\partial t^2} \ln \mG_{2n+1} \right] =  \lim_{n \to \infty } \frac{1}{(2n+1)^2} \ln \left[ \mG_{2n} \mG_{2n+2} \right] .
\end{equation*}
We also observe that, near the first zero, 
\begin{equation*}
    \mG_N (t) \propto  (t-N \tau^{\text{QSL}} _N) ,
\end{equation*}
which again is a consequence of the analysis in \cite{CelsusDetC} and we neglect sub-leading terms. The linearity near the zero implies that 
\begin{equation*}
    	\lim_{\tau \to \tau^{\text{QSL}} _{2n+1}} \mG_{2n+1}^2 \frac{ \partial^2 \ }{\partial t^2} \ln \mG_{2n+1} = \text{constant} .
\end{equation*}
Taking one further derivative of \eqref{eq:rewriteToda}, we have on the left-hand side 
\begin{equation*}
    \frac{\partial \ }{\partial t} \left[ 	\mG_{N}^2 \frac{ \partial^2 \ }{\partial t^2} \ln \mG_{N} \right] = 2 \mG_N \left( \frac{\partial \mG_N }{\partial t} \right) \left( \frac{ \partial^2 \ }{\partial t^2} \ln \mG_{N}  \right) + \mG_N ^2 \left( \frac{ \partial^3 \ }{\partial t^3} \ln \mG_{N}  \right) 
\end{equation*}
and on the right-hand side 
\begin{equation*}
     -4 \frac{\partial \ }{\partial t} \left[ 	\mG_{N-1} \mG_{N+1} \right] = -4 \mG_{N-1} \mG_{N+1} \left[ \frac{\partial \ }{\partial t} \ln \mG_{N-1} + \frac{\partial \ }{\partial t} \ln \mG_{N+1}  \right] .
\end{equation*}
We now set $N=2n+1$, take the logarithm of both expressions, divide by $(2n+1)^2$, and eventually take the limit $n \to \infty$. Equating the two sides, as dictated by having differentiated \eqref{eq:rewriteToda}, we have that
\begin{itemize}
    \item[(lhs)] the left-hand side is continuous function of $\tau$, with is analytic in $\tau < \tau^{\text{QSL}}$ and has its first non-analytic point there.
    \item[(rhs)] In turn, the right-hand side is known from the derivation in Subsection \ref{sec:DQPTABCLinfty} and is smooth in $\tau < \tau_{\text{cr}}$ and non-analytic at $\tau= \tau_{\text{cr}}$.
\end{itemize}
Imposing the match of the two sides we conclude that the equality \eqref{eq:tauqsl=taucr} must hold.\par
\medskip
\begin{rmk}[Breakdown of factorization]
The existence of a QSL is related to the breakdown of the factorization at large $N$ in the random matrix ensembles \eqref{eq:G=MMlong}. Indeed, from \eqref{eq:Gmatrixintegralform} we have 
\begin{equation*}
    \frac{\partial \ }{\partial t } \frac{1}{N} \ln \mG_{N, L=\infty } = -i \left\langle \frac{1}{N} \mathrm{Tr} \left( U + U^{-1} \right)   \right\rangle_G .
\end{equation*}
Taking one further derivative and recalling that $t=N \tau$ we get 
\begin{align}
    \left\langle \left[ \frac{1}{N} \mathrm{Tr} \left( U + U^{-1} \right)  \right]^2 \right\rangle_G - \left[  \left\langle \frac{1}{N} \mathrm{Tr} \left( U + U^{-1} \right) \right\rangle_G  \right]^2  & \equiv   \left\langle \left[ \frac{1}{N} \mathrm{Tr} \left( U + U^{-1} \right)  \right]^2 \right\rangle_G ^{\text{conn.}} \notag \\
    & = \frac{1}{N^2} \frac{\partial^2 f}{\partial \tau^2 } \label{eq:conncorrderiv}
\end{align}
with $f(\tau)$ defined as in \eqref{eq:defDFEtaulargeN}, except that now we allow $N$ to be odd. (The minus sign in the definition cancels the $(\sqrt{-1})^2$ from the two derivatives). Note that \eqref{eq:conncorrderiv} is exact and holds for every $N \in \mathbb{N}$.\par
Whenever $f(\tau)$ is finite, \eqref{eq:conncorrderiv} implies the well-known factorization result at large $N$, because the right-hand side vanishes. This is the case for the large $N$ limit with $N=2n \in 2 \mathbb{N}$ and also for $N=2n+1$ in the region $\tau < \tau^{\text{QSL}}$. However, at the value $\tau = \tau^{\text{QSL}} _{2n+1}$, $f (\tau)$ develops a singularity at finite $n$ which prevents the factorization in the limit $n \to \infty $ of \eqref{eq:conncorrderiv}.\par
In summary, the QSL is associated with the breakdown of factorization at large $N$ in the current model. Intriguingly, this effect is reminiscent of wormholes in gravitational theories. Wormholes are connected geometries that prevent the factorization of the Hilbert space, and their kicking in triggers a phase transition. Albeit the phase transition associated with the dominance of the connected wormholes geometries is typically first order, the analogy seems worthwhile of further exploration.
\end{rmk}

\section{Discussion}

The physics of coherent nonequilibrium real-time evolution is very rich yet still vastly unexplored. In this work we have uncovered a novel dynamical quantum phase transition, with many distinctive features. Such DQPT is detected by the real-time Loschmidt echo of the XY model but, differently from previous works, we have analyzed it in the planar limit, guided by QFT and RMT.\par
The DQPT takes place as a function of time in the scaling limit \eqref{eq:deftau} and its characterization has required the study of a random matrix ensemble with complex weight, for which we developed tailored analytical methods. In particular, the DQPT is not merely a Wick rotation of the well-established GWW transition.\par
The DQPT arises in the thermodynamic limit, but we have managed to fully identify the finite chain effects and gain analytic control on the exponentially small discrepancies.\par
While the above points hold true for $N$ large and even, we have also discussed a phase transitions induced by the zeros of the Loschmidt echo when $N$ is odd and have established a relation with the quantum speed limit. Every analytical result has been supported by numerical evaluations, shown in the Figures.\par
Third order DQPTs have never been measured experimentally. We have shown here how not only it is within experimental reach, but it is predicted for one of the most fundamental strongly-correlated systems, as is the isotropic XY Heisenberg chain. Crucially, the spin interactions of this model have already been engineered with quantum simulators \cite{zhang2017observation}.\par
\medskip
The experimental verification of the DQPT discovered here is certainly the most appealing potential follow-up of our work. Another problem open for future study is to establish the universality class of the transition. Based on heuristic arguments, we expect that it is governed by Painlev\'e I equation. This would mean the universality class of 2d quantum gravity \cite{Brezin:1990rb,Gross:1989vs,Douglas:1989ve}, as opposed to that of GWW, characterized by Painlev\'e II \cite{Periwal:1990gf}. It would be interesting to rigorously confirm or refute this conjecture.\par
From a broader perspective, in this paper we have opened a new route to analytically study DQPTs and to uncover their universality classes, based on the power of RMT. As the natural next step, inspired by the non-analyticity result of \cite{Piroli:2018amn}, it is conceivable, although technically challenging, that an exhaustively characterization of DQPT as done presently can be achieved in more general 1d models.

\vspace{6pt}
\subsubsection*{Acknowledgments}
The work has been financially supported by Comunidad de Madrid (grant QUITEMAD-CM, ref. P2018/TCS-4342), by the Ministry of Economic Affairs and Digital Transformation of the Spanish Government through the QUANTUM ENIA project call - QUANTUM SPAIN project, and by the European Union through the Recovery, Transformation and Resilience Plan - NextGenerationEU within the framework of the Digital Spain 2025 Agenda.\\
\indent DPG also acknowledges financial support from the Spanish Ministry of Science and Innovation (“Severo Ochoa Programme for Centres of Excellence in R\&D” CEX2019- 000904-S and grant PID2020-113523GB-I00) as well as from CSIC Quantum Technologies Platform PTI-001.\\
\indent LS acknowledges warm hospitality at the Departamento de An\'alisis Matem\'atico y Matem\'atica Aplicada, Universidad Complutense de Madrid, and financial support from the Funda\c{c}\~{a}o para a Ci\^{e}ncia e a Tecnologia (FCT) through the Project PTDC/MAT-PUR/30234/2017, from the Beijing Natural Science Foundation through the project IS23008 ``Exact results in algebraic geometry from supersymmetric field theory'', and from the Shuimu Scholars program of Tsinghua University.

{\small
\bibliography{spin3,spinC}
}
\end{document}